\documentclass[aps, pra, twocolumn,superscriptaddress, 10pt]{revtex4-1}
\usepackage[utf8]{inputenc}
\usepackage{graphicx}
\usepackage{dcolumn}
\usepackage{bm}
\usepackage{multirow}
\usepackage{amssymb}
\usepackage{amsfonts}
\usepackage{amsmath}
\usepackage{amsthm}
\usepackage{enumerate}
\usepackage{color}
\usepackage[dvipsnames]{xcolor}
\usepackage{marvosym}
\newcommand{\mathbbm}[1]{\text{\usefont{U}{bbm}{m}{n}#1}}
\usepackage{qcircuit}

\usepackage{mathtools}

\usepackage{hyperref}

\usepackage{soul,xcolor}
\setstcolor{purple}

\usepackage{tikz}
\usetikzlibrary{matrix,positioning,decorations.pathreplacing}

\newtheorem{theorem}{Theorem}

\newtheorem{observation}[theorem]{Observation}
\newtheorem{corollary}[theorem]{Corollary}

\newcommand{\tr}{\text{tr}}

\newcommand{\rank}{\text{rank}}

\newcommand{\bra}[1]{\langle #1|}
\newcommand{\ket}[1]{|#1\rangle}
\newcommand{\braket}[2]{\langle #1|#2\rangle}

\newcommand{\C}{\ensuremath{\mathbbm C}}
\newcommand{\be}{\begin{equation}}
\newcommand{\ee}{\end{equation}}
\newcommand{\bea}{\begin{eqnarray}}
\newcommand{\eea}{\end{eqnarray}}

\newcommand{\kommentar}[1]{}

\newcommand{\identity}{\mathbbm{1}}

\newtheorem{lemma}[theorem]{Lemma}

\newcommand{\forget}[1]{}
\begin{document}
	
	\title{State transformations within entanglement classes containing permutation-symmetric states}

	\date{\today}
	
	\author{Martin Hebenstreit}
	\affiliation{Institute for Theoretical Physics, University of Innsbruck, Technikerstr. 21A, 6020 Innsbruck, Austria}
	\author{Cornelia Spee}
	\affiliation{Institute for Theoretical Physics, University of Innsbruck, Technikerstr. 21A, 6020 Innsbruck, Austria}
	\author{Nicky Kai Hong Li}
	\affiliation{Institute for Theoretical Physics, University of Innsbruck, Technikerstr. 21A, 6020 Innsbruck, Austria}
	\author{Barbara Kraus}
	\affiliation{Institute for Theoretical Physics, University of Innsbruck, Technikerstr. 21A,  6020 Innsbruck, Austria}
	\author{Julio I. de Vicente}
	\affiliation{Departamento de Matemáticas, Universidad Carlos III de Madrid, Avda. de la Universidad 30, E-28911, Leganés (Madrid), Spain}
	
	\begin{abstract}
		The study of state transformations under local operations and classical communication (LOCC) plays a crucial role in entanglement theory. While this has been long ago characterized for pure bipartite states, the situation is drastically different for systems of more parties: generic pure qudit states cannot be obtained from nor transformed to any state (i.e., they are isolated), which contains a different amount of entanglement. We consider here the question of LOCC convertibility for permutation-symmetric pure states of an arbitrary number of parties and local dimension, a class of clear interest both for physical and mathematical reasons and for which the aforementioned result does not apply given that it is a zero-measure subset in the state space. While it turns out that generic $n$-qubit symmetric states are also isolated, we consider particular families for which we can determine to be, on the contrary, endowed with a rich local stabilizer, a necessary requirement for LOCC convertibility to be possible. This allows us to identify classes in which LOCC transformations among permutation-symmetric states are possible. Notwithstanding, we provide several results that indicate severe obstructions to LOCC convertibility in general even within these highly symmetrical classes. In the course of the study of LOCC transformations, we also characterize the local symmetries of symmetric states.
	\end{abstract}

	\maketitle

	\section{Introduction}
	
	Entanglement is a purely quantum effect of multicomponent systems, which is behind quantum technologies outperforming their classical counterparts. Therefore, it has been a major subject of study in quantum information theory since its inception \cite{reviewe}. In addition to this, the tools developed in this field have been found to play a crucial role in the understanding of condensed matter physics \cite{reviewcm}. Entanglement theory aims at characterizing which states possess this property and the different forms in which it can manifest, providing protocols to manipulate it and quantifying to which extent it can be used. This theory is formulated within the framework of quantum resource theories \cite{reviewqrt} where local operations assisted by classical communication (LOCC) (see e.g.\ \cite{locc1,locc2,HeEn21}) are considered to be free operations, which assigns a pivotal role to this notion. First, entanglement can be understood as a resource that allows implementation of tasks which would otherwise be impossible for distant parties that are constrained to the use of LOCC. Second, LOCC transformations induce an operationally meaningful partial order in the set of entangled states under which any entanglement measure must be monotonic. If a state $\Psi$ can be deterministically transformed under this class of operations into $\Phi$, the former cannot have less entanglement than the latter as any protocol that can be implemented within this paradigm with $\Phi$ can also be achieved with $\Psi$ but not necessarily the other way around. Thus, characterizing LOCC state transformations is a crucial problem in this theory. It also provides basic primitives for quantum information protocols and, more importantly, identifies which states are potentially more useful in applications. In this regard, Nielsen's theorem \cite{nielsen} provides a milestone result as it characterizes LOCC convertibility among pure bipartite states in term of a simple majorization relation and shows that simple LOCC protocols with a single round of classical communication are sufficient \cite{LoPopescu}. Furthermore, from the theorem, it follows that there exists a unique (up to local unitary (LU) transformations) maximally entangled state $|\phi_+\rangle=\sum_{i=0}^{d-1}|ii\rangle/\sqrt{d}$ that can be transformed to any other for any fixed value of the local dimension $d$. This means that $|\phi_+\rangle$ has to be the most useful state for any task subject to parties implementing LOCC protocols regardless of its specific goal, be it teleportation, metrology or other tasks.
	
	Regrettably, this picture is much less satisfactory when one considers pure states shared by more than two parties. First, it has been long known that in this case there exist pairs of states which are not related by LOCC even if one lifts the premise of deterministic transformations and allows an arbitrary non-zero probability of success \cite{slocc}. States that can be interconverted in this latter paradigm are said to belong to the same stochastic LOCC (SLOCC) class. While there is a unique SLOCC class for pure entangled bipartite states, genuinely multipartite entangled states of three qubits give rise to two classes \cite{slocc}, identified respectively by the paradigmatic GHZ and W states
	\begin{align}
		|GHZ\rangle&=\frac{1}{\sqrt{2}}(|000\rangle+|111\rangle),\nonumber\\
		|W\rangle&=\frac{1}{\sqrt{3}}(|001\rangle+|010\rangle+|100\rangle).
	\end{align}
	The situation has been investigated for a larger number of parties and/or local dimension where it becomes much more complex since in these cases there can be infinitely many SLOCC classes \cite{VeDe02,BrLu04}. This establishes that contrary to the bipartite case, multipartite entanglement appears in inequivalent forms. Specifically, SLOCC classes pinpoint different equivalence classes but provide no sense of comparison of the relative usefulness of states. Thus, one could try to study the ordering induced by LOCC within each SLOCC class in order to identify the most entangled states for each family. Nevertheless, the analysis of LOCC convertibility in the multipartite regime appears to be a formidable mathematical problem: the simple structure of protocols that is sufficient in the bipartite case does not apply here and it is not even known whether it is necessary to consider LOCC transformations involving infinitely many rounds of classical communication \cite{SpdV17,dVSp17,HeEn21}. More importantly, multipartite entanglement displays an extreme behaviour in this respect that dooms the above program to failure. It turns out that almost all states are isolated under LOCC for most configurations of number of parties and local dimension \cite{GoKr17,SaWa18}. In more detail, these works establish that a generic pure state of homogeneous systems constituted of $n>3$ parties with local dimension $d>2$ (as well as systems for which $d=2$, $n>4$) cannot be transformed to nor obtained from a non-LU-equivalent state by LOCC. This implies that the counterpart of the maximally entangled bipartite state, the maximally entangled set \cite{dVSp13} is of full measure in the Hilbert space. Interestingly, the local stabilizer of each SLOCC class plays a crucial role in characterizing LOCC transformations \cite{GoWa11} and the aforementioned result follows by proving that for almost all states this stabilizer is trivial (i.e.,\ it only contains the identity). In sum, it turns out that the resource theory of entanglement as formulated by the paradigm of LOCC is generically trivial in the sense that almost every pair of states (even if restricted to the same SLOCC class) is incomparable. Obviously, this statement does not forbid the existence of particular zero-measure SLOCC classes that might be free of isolation and display a rich LOCC structure in which every state can be converted to a less entangled state. Notwithstanding, to the best of our knowledge only two such examples are known in the literature: the $n$-qubit W and GHZ classes \cite{turgutW,turgutGHZ}. On the contrary, all classes studied in 4-qubit states \cite{dVSp13,SpdV16} and 3-qutrit states \cite{HeSp16} have $LOCC_\mathbb{N}$ isolated states, i.e., states that can neither be reached nor converted via a non-trivial finite-round LOCC transformation.
	
	We believe that the above results call for a more specific analysis. On the one hand, this suggests that the study of multipartite LOCC convertibility should be restricted to non-generic subsets that are motivated by physical or informational interest. On the other hand, studying classes that have a relevant mathematical structure might make the state-transformation problem more amenable and lead to new families that do not display isolation, thus identifying potentially useful states for applications. Actually, these two alternatives can be brought together under the umbrella of symmetry. Symmetry principles are a cornerstone in the development of physical theories and quantum mechanics is no exception and, at the same time, they provide a mathematically pleasant structure. Whereas the symmetries of stabilizer states \cite{EnKr20} and of translational invariant matrix product states (of low bond dimension) \cite{SaMo19,He21} have been studied, we focus here on another very relevant class of multipartite states, the one corresponding to permutation symmetry, i.e.\ states in the symmetric subspace. These states not only describe systems of many bosons for indistinguishable particles but they are often encountered in quantum optics (such as in the study of superradiance \cite{Dicke}) and appear as the ground states of natural Hamiltonians (e.g.,\ in the Lipkin-Meshkov-Glick model \cite{Orus}). Moreover, permutation-symmetric (referred to in what follows just as symmetric) states such as the $|\phi_+\rangle$, W, GHZ and Dicke states correspond to canonical states in quantum information theory and its applications (see e.g.\ \cite{Aulbach} and references therein) and a huge experimental effort has been put forward to prepare these states in the lab both for qubits \cite{expqubits} and for higher local dimensions \cite{expqudits}. In fact, their symmetry structure has enabled an in-depth analysis of several of their properties regarding entanglement \cite{Aulbach,symentanglement}. Most noticeably in our context, SLOCC classes of symmetric states have been carefully studied and several remarkable mathematical properties have been identified from this point of view \cite{BaKr09,MaKr10,MiRo13}. In particular, some of these classes have been found to have elements in their local stabilizer, which suggest that the latter is highly non-trivial. We refer to these subclasses of symmetric states as exceptionally-symmetric (ES) classes (a proper definition is provided in the next section). Lastly, it should be pointed out that so far, all the known examples of SLOCC classes that do not display isolation happen to contain symmetric states (the W and GHZ qubit states).
	
	Thus, the main goal of this paper is to study state convertibility in general within SLOCC classes that contain symmetric states. Due to the technical difficulty of the problem at hand, in several instances we will restrict to LOCC protocols under the very mild and natural assumption that they consist of a finite (but otherwise arbitrary) number of rounds of classical communication, which we will denote by LOCC$_{\mathbb{N}}$. In passing, we will determine the properties and structure of the local stabilizer group of these families, which is interesting in its own right \cite{SaMo19,EnKr20} and could find applications beyond the study of state transformations. 
	
	The outline of the remainder of the article is the following.
	In Section \ref{sec:preliminaries}, we introduce notation and briefly review results on the symmetric subspace, SLOCC classes containing symmetric states, and the stabilizer of symmetric states. We first present the preliminaries as these allow us to then  provide the overview of our results in Section \ref{sec:results} in a comprehensible way.
	In Section \ref{sec:reachability}, we generalize results from \cite{SpdV17,dVSp17} concerning $LOCC_\mathbb{N}$-reachability as well as $LOCC_1$-convertibility (i.e.,\ one-round protocols) of states with a finite stabilizer to the case of an infinite stabilizer, which is indispensable for addressing symmetric classes.
	In Section \ref{sec:symmetric}, we consider $n$ qubits and we argue that, although SLOCC classes containing symmetric states form a zero-measure set in the Hilbert space, generically, they have trivial stabilizer for $n\geq 5$ and thus, can neither be reached nor converted via LOCC. This allows us, in turn, to completely characterize the entanglement content of those states via easily computable entanglement measures. In Section \ref{sec:non_ES}, we argue that non-ES classes do not seem promising with respect to a rich LOCC structure. However, a subclass of  ES classes, the non-derogatory ES classes (for the definition see Section \ref{sec:preliminaries}), do turn out to possess a promising stabilizer which we derive in Section \ref{sec:non_derog}. We further show in this section that indeed a rich LOCC structure is present. More precisely, we show that  LOCC transformations between two different symmetric states are possible (Theorem \ref{thm:symtosym}). However, we also show that there always exist states that are isolated (Theorem \ref{theo:sumsisolation}) in a sense that we will explain there, except for the outstanding GHZ- and W-class for $n$-qubit states, which are free from isolation for all $n$.
	Finally, we deal with derogatory ES classes (for the definition see Section \ref{sec:preliminaries}) in Section \ref{sec:derog}. We argue that fully characterizing the stabilizer of derogatory ES classes is beyond the scope of the present work. We show that some of the SLOCC classes somehow resemble multi-copy scenarios of the aforementioned non-derogatory ES classes, and therefore inherit a non-trivial LOCC structure. Moreover, we construct a 5-qutrit derogatory ES class within which all states are $LOCC_{\mathbb{N}}$-isolated, despite having a non-trivial stabilizer. Finally, as illustrative examples, we scrutinize the three- and four-qutrit cases. 
	
	\section{Preliminaries}
	\label{sec:preliminaries}
	
	We consider $n$-partite states with $n\geq 3$ in homogeneous systems described by the Hilbert space $\mathcal{H} = \mathcal{H}_1 \otimes \ldots \otimes   \mathcal{H}_n = \mathbb{C}^d \otimes \ldots \otimes \mathbb{C}^d$. We only consider fully entangled states, i.e., states for which all reduced density matrices have full rank. States $\ket{\Psi}$ and $\ket{\Phi}$ are SLOCC-equivalent iff there exist local invertible operators $A_i \in GL(d, \mathbb{C})$ such that $\ket{\Psi} = A_1 \otimes \ldots \otimes A_n \ket{\Phi}$. We will write states within the same SLOCC class as local invertible operators acting on a representative of the SLOCC class, which we will also call a seed state, $\ket{\Psi_s}$. When considering state transformations, we will usually denote the initial state as $\ket{\Psi} =g_1 \otimes \ldots \otimes g_n \ket{\Psi_s}$ and the final state as $\ket{\Phi} = h_1 \otimes \ldots \otimes h_n \ket{\Psi_s}$. Moreover, we use the notation $g = g_1 \otimes \ldots \otimes g_n$, $G =  G_1 \otimes \ldots \otimes G_n = g^\dagger g$ and analogue notation for $h_1 \otimes \ldots \otimes h_n$. For convenience, we work with unnormalized states. Furthermore, let us denote the stabilizer, or (local) symmetries, of a state $\ket{\Psi}$ by $S_{\Psi} = \{ S^{(1)} \otimes \ldots \otimes S^{(n)} \in GL(d, \mathbb{C})^{\times n} \ | \ S^{(1)} \otimes \ldots \otimes S^{(n)} \ket{\Psi} = \ket{\Psi}\}$. If not stated differently a superindex $(i)$ indicates as above that we refer to the part of the symmetry that is acting on party $i$. In order to denote that the matrix $B$ is acting on party $i$ we will use the notation $B_{(i)}$. Moreover, we will use a superindex $M^{[b]}$ to indicate the submatrix matrix $(M)_{bb}$ for the block $b$.

	\subsection{The symmetric subspace}
	Let us briefly recall some results about permutation-symmetric states, which we will simply call symmetric states in the following. We denote the symmetric subspace of $(\mathbb{C}^d)^{\otimes n}$ as $\operatorname{Sym}^n(\mathbb{C}^d)$, which is defined as
	\begin{align}
		\operatorname{Sym}^n(\mathbb{C}^d) = \left\{  \ket{\Psi} \in (\mathbb{C}^d)^{\otimes n} \ | \   P_{i,j} \ket{\Psi} = \ket{\Psi} \ \forall i,j     \right\},
	\end{align}
	where $P_{i,j}$ is a non-local operator permuting parties $i$ and $j$ \cite{Ha13,referenceBookWatrous}. Its complex dimension equals the binomial coefficient $\begin{pmatrix}n + d -1\\d-1 \end{pmatrix}$ \cite{Ha13,referenceBookWatrous}. As a simple example, the symmetric subspace for two qubits is spanned by $\{\ket{00},\ket{11},(\ket{01} + \ket{10})/\sqrt{2}\}$ and hence 3-dimensional. In general, a state in $\ket{\Psi} \in \operatorname{Sym}^n(\mathbb{C}^d)$ can be written as 
	\begin{align}
		\label{eq:symstate}
		\ket{\Psi} =  \sum_{\vec{n}} \sum_{\vec{i}:\#i = n_i} \alpha_{\vec{n}}\ \ket{i_1, \ldots, i_n},
	\end{align}
	where the first sum is over all $d$-dimensional vectors $\vec{n}$, whose entries sum to $n$, the second sum is over all $n$-dimensional vectors $\vec{i}$ in which the entry $i$ occurs exactly $n_i$ times, and $\alpha_{\vec{n}} \in \mathbb{C}$ \cite{Ha13}.
	
	Alternatively to Eq.\;(\ref{eq:symstate}), any $\ket{\psi} \in \operatorname{Sym}^n(\mathbb{C}^d)$ can also be written as \cite{comon}
	\begin{equation}\label{srank}
	|\psi\rangle=\sum_{i=1}^r|w_i\rangle^{\otimes n},
	\end{equation}
	where $|w_i\rangle\in\mathbb{C}^d$ $\forall i$. The smallest value of $r$ is called symmetric tensor rank, which we denote by $\rank_S(\psi)$. Since two symmetric states $|\psi\rangle$ and $|\phi\rangle$ are in the same SLOCC class iff there exists an invertible $A\in\mathbb{C}^{d\times d}$ such that $|\phi\rangle=A^{\otimes n}|\psi\rangle$ \cite{MiRo13,  MaKr10}, the symmetric tensor rank is the same for all symmetric states in the same SLOCC class.

	In the case of symmetric qubit states we will also use the Majorana representation. Ignoring normalization, every $|\psi\rangle\in$\;$\operatorname{Sym}^n(\mathbb{C}^2)$ can be written as \cite{comon}
	\begin{equation}\label{majorana}
	|\psi\rangle=\sum_\pi P_\pi(|\epsilon_1\cdots\epsilon_n\rangle),
	\end{equation}
	where each $|\epsilon_j\rangle$ is a single-qubit state, i.e., $|\epsilon_j\rangle=\cos\alpha_j|0\rangle+\exp(i\beta_j)\sin\alpha_j|1\rangle$ with $\alpha_j\in[0,\pi/2]$ and $\beta_j\in[0,\pi]$ $\forall j$, and the summation runs over all possible permutations. The Majorana representation is unique in the sense that the $2n$ real parameters $\{\alpha_j,\beta_j\}$ uniquely identify every state $|\psi\rangle\in$\;$\operatorname{Sym}^n(\mathbb{C}^2)$. Given a state $|\psi\rangle$, the diversity degree $m$ is the number of pairs $(\alpha_j,\beta_j)$ that take different values and the degeneracy configuration $\{k_1,\ldots,k_m\}$ labels how often a state, $\ket{\epsilon_j}$, (i.e., a pair $(\alpha_j,\beta_j)$) occurs in $|\epsilon_1\cdots\epsilon_n\rangle$ [see Eq. (\ref{majorana})]. Since the Majorana representation is unique and two symmetric states $|\psi\rangle$ and $|\phi\rangle$ are in the same SLOCC class iff there exists an invertible $A\in\mathbb{C}^{2\times2}$ such that $|\phi\rangle=A^{\otimes n}|\psi\rangle$, the diversity degree and the degeneracy configuration are SLOCC invariant \cite{BaKr09}.

	In the following, we will be mainly concerned with SLOCC classes that contain at least one symmetric state. Note that the union of those SLOCC classes is indeed a measure zero set, and thus, having a potentially rich LOCC structure there does not contradict with Refs.\;\cite{GoKr17,SaWa18}.
	This can be seen via a simple parameter counting argument as follows. As mentioned before, the dimension of the symmetric subspace $\operatorname{Sym}^n(\mathbb{C}^d)$ is $\begin{pmatrix}n + d -1\\d-1 \end{pmatrix}$. The maximal dimension of an SLOCC orbit is $n (d^2-1)$. Thus, the union of SLOCC classes containing at least one symmetric state has a dimension of no more than $n (d^2-1) + \begin{pmatrix}n + d -1\\d-1 \end{pmatrix}$ \cite{footnote1}, while the full Hilbert space has dimension $d^n$. Given a fixed $d$, the union of SLOCC classes containing symmetric states will thus be of measure zero for sufficiently large $n$.

	\subsection{Exceptionally symmetric and non-exceptionally symmetric states}\label{subsec:Prelim_ESandNES}
	In \cite{MiRo13,  MaKr10}, it has been shown that any two SLOCC-equivalent symmetric states $\ket{\Psi}$ and $\ket{\Phi}$, i.e., symmetric states for which there exists $A_1, \ldots, A_n$ such that $\ket{\Phi} =A_1 \otimes \ldots \otimes A_n  \ket{\Psi}$, are always related by a local operation of the form $A^{\otimes n}$, i.e., $\ket{\Phi} = A^{\otimes n} \ket{\Psi}$ for some $A$. Hence, given a symmetric seed state $\ket{\Psi_s}$, it suffices to consider operators of the form $A^{\otimes n}$ in order to obtain all symmetric states within the SLOCC class of $\ket{\Psi_s}$.
	
	Another result in \cite{MiRo13, MaKr10} addresses the stabilizer of symmetric states. Some symmetric states have a stabilizer of the form $B_{(i)} \otimes B^{-1}_{(j)}$ for a matrix $B  \not\propto \identity$, i.e., $B$ acting on some party $i$ and $B^{-1}$ acting on some other party $j$. Clearly, this is a property that is common to the whole SLOCC class. Note that the eigenvalues of $B$ do not play a role (only their degeneracies), as whenever $B \otimes B^{-1}$ is a symmetry of a state, then so is $f(B) \otimes f(B)^{-1}$ for any analytic function $f$ \cite{MiRo13}. This allows one to construct a new symmetry with $\tilde{B}$ which has the same Jordan structure but has arbitrarily changed eigenvalues (respecting the degeneracies) \cite{MiRo13}. One can then distinguish two different types of symmetric states (SLOCC classes containing a symmetric state), which we call  \emph{exceptionally symmetric} [ES] and   \emph{non-exceptionally symmetric} [non-ES] states (classes). Symmetric states that have non-trivial symmetries of the form $B \otimes B^{-1}$ belong to ES classes and states that do not have such symmetries belong to non-ES classes. In fact, the latter states solely have symmetries of the form $A^{\otimes n}$.
	
	Let us mention that to solve for the symmetries of a symmetric state $\ket{\psi}$, it is sometimes more convenient to solve the equation $A^{\otimes n}\ket{\psi}=\lambda\ket{\psi}$ for $A\in SL(d,\mathbb{C})$ (or other suitable choices of normalization, e.g., $[A]_{1,1}=1$). Note further that the equation $B\otimes B^{-1}\otimes\identity^{\otimes n-2}\ket{\psi}=\lambda'\ket{\psi}$ holds only if $\lambda'=1$ \cite{footnote2}. 
	In the following, we will therefore consider the symmetry defining equation with a proportionality factor whenever it eases the presentation.

	Clearly, considering ES classes, the Jordan normal form of $B$ is invariant under SLOCC \cite{MiRo13}. Hence,  we will, without loss of generality, consider in the remainder of this section and in Section \ref{subsec:Prelim_DandND}  $B$ in Jordan normal form. Here and in the following, $J_k$ denotes a $(k \times k)$ Jordan block with eigenvalue 1. If $B$ is a single Jordan block of size $(k+1) \times (k+1)$, then the state
	\begin{align}
		\label{eq:ek}
		\ket{E_k} = \sum_{i_1 + \ldots + i_n = k} \ket{i_1 i_2 \ldots i_n} \in {\mathbb{C}^{k+1}}^{\otimes n}
	\end{align}
	is the (up to SLOCC) unique state stabilized by $B\otimes B^{-1}$ \cite{MiRo13}. In the following, we call these states states with $k$ excitations. For $k=1$, they are identical to the Dicke states with one excitation. However, for $k>1$, these states differ from the Dicke states as not all terms in the superposition are the same up to permutation (e.g., $\ket{E_2}$ for $n=2$ is given by  $\ket{E_2}=\ket{11}+\ket{02}+\ket{20}$). If $B$ contains more than one Jordan block, but all eigenvalues have geometric multiplicity 1, i.e., eigenvalues corresponding to different Jordan blocks are different, then
	\begin{align}
		\label{eq:sumek}
		\bigoplus_{b=1}^K \ket{E_{k_b}}
	\end{align}
	is (up to SLOCC) the unique state stabilized  by $B\otimes B^{-1}$  \cite{MiRo13}, where $K$ is the total number of Jordan blocks in $B$ of respective sizes $k_b + 1$. Recall that only the block structure and the multiplicities but not the specific values of the eigenvalues are relevant here.
	
	\subsection{Derogatory and non-derogatory exceptionally symmetric states}\label{subsec:Prelim_DandND}
	
	The matter becomes more involved when matrices $B$ with eigenvalues whose geometric multiplicity is larger than 1, so-called \textit{derogatory} matrices, are considered. In this case, states stabilized by $B \otimes B^{-1}$  are not unique any more \cite{MiRo13}. To discuss these classes further, we use the same notation as in \cite{MiRo13} and consider first the case where all Jordan blocks in $B$ share the same eigenvalue. The state $\ket{i_j^{(b_j)}}$ denotes here and in the following the $(i_j+1)$-th standard basis vector in the Jordan block $b_j$. An excitation $l$ is assigned to the $(l+1)$-th standard basis vector in a Jordan block. With this notation one can write the states of fixed excitation number $k$, which are stabilized by a  derogatory $B \otimes B^{-1}$ as \cite{MiRo13}
	
	\begin{align}
		\label{eq:ekderog}
		\ket{E_k^{n_1, \ldots, n_K}} = \sum_{\vec{b}:\#b = n_b}  \sum_{i_1 + \ldots + i_n = k} \ket{i_1^{(b_1)} i_2^{(b_2)} \ldots i_n^{(b_n)}}, 
	\end{align}
	where $n_b \geq 0$ for all $b \in \{1, \ldots, K\}$, and $n_1 + \ldots + n_K = n$. The total number of excitations is denoted by $k$. The first sum runs over all $n$-dimensional vectors $\vec{b}$ comprised of integers between 1 and $K$ such that each integer $b\in \{1, \ldots, K\}$ appears exactly $n_b$ times within $\vec{b}$. In other words, $n_b$ denotes the number of parties for which $\ket{i_j^{(b_j)}}$ has to lie in the range of the $b$th Jordan block. 
	To illustrate Eq.\;(\ref{eq:ekderog}), let us consider the three-qutrit case with a Jordan block of size 1 and a Jordan block of size 2. We have that $\ket{E_0^{3,0}}=\ket{000}$ since there is no excitation and all local vectors lie in the range of the first Jordan block. As for the state $\ket{E_0^{2,1}}=\ket{W}$, there is still no excitation, but since the zero-excitation state for the first block is $\ket{0^{(1)}}=\ket{0}$ and the one for the second block is $\ket{0^{(2)}}=\ket{1}$, we have that within each term, the state $\ket{0}$ appears twice ($n_1=2$) and $\ket{1}$ appears once ($n_2=1$).
	
	A full classification of SLOCC classes for derogatory $B$ is unknown. In particular, the set $\{\ket{E_k^{n_1, \ldots, n_K}}\}_{k, (n_1 \ldots, n_K)}$ does not constitute a set of representatives. However, it has been shown that a representative of any SLOCC classes stabilized by some $B\otimes B^{-1}$, where $B$ has only a single eigenvalue $\lambda$, can be written as
	\begin{align}
		\label{eq:ekderogsuperpos}
		\ket{\Phi_{\boldsymbol{\alpha}}^\lambda}=\sum_{k=0}^{\max_b(k_b)} \sum_{n_1+ \ldots + n_K = n} \alpha_{k, n_1, \ldots, n_K} \ket{E_k^{n_1, \ldots, n_K}},
	\end{align}
	where $k_b$ denotes the individual block sizes, and in the sum over $n_b$, those $n_b$ for which $k_b$ is smaller than $k$ are zero, and $\alpha_{k, n_1, \ldots, n_K} \in \mathbb{C}$ \cite{footnote3}.

	If $B$ constitutes of $K-m$ Jordan blocks with different eigenvalues (to which we refer to via $b\in \{1, \ldots, K-m\}$) and $m$ Jordan blocks sharing $d$ degenerate eigenvalues $\{\lambda_1,\ldots,\lambda_d\}$, then the states that are stabilized by $B\otimes B^{-1}$ can be written as the direct sum
	\begin{align}
		\label{eq:sumek2}
		\bigoplus_{b=1}^{K-m} \ket{E_{k_b}}\bigoplus_{j=1}^d \ket{\Phi_{\boldsymbol{\alpha_j}}^{\lambda_j}},
	\end{align}
	where $\ket{\Phi_{\boldsymbol{\alpha_j}}^{\lambda_j}}$ involves only the Jordan blocks that share the eigenvalue $\lambda_j$ \cite{footnote4}.
	
	Going beyond the distinction between ES and non-ES states and classes, we will now categorize ES states and classes into two distinct families. We will call the ES classes stabilized by some non-derogatory $B\otimes B^{-1}$ \emph{non-derogatory ES} SLOCC classes and those that are stabilized solely by derogatory $B \otimes B^{-1}$ \emph{derogatory ES} SLOCC classes. Equivalently, the ES SLOCC classes represented by states in Eqs. (\ref{eq:ek}) and (\ref{eq:sumek}) are non-derogatory, while all remaining ES SLOCC classes are derogatory. Let us remark that non-derogatory ES classes may very well be stabilized by some derogatory $B\otimes B^{-1}$. However, derogatory ES classes may not be stabilized by any non-derogatory $B\otimes B^{-1}$.
	
	Let us consider here a few simple examples to illustrate this classification. In particular, we provide examples which show that superpositions as in Eq.\;(\ref{eq:ekderogsuperpos}) have to be taken into account as not all of these are within an SLOCC orbit of states of the form in Eq.\;(\ref{eq:ekderog}). However, it should be noted that not any state of the form in Eq.\;(\ref{eq:ekderogsuperpos}) results in a different SLOCC class. Moreover, there exist non-derogatory ES states that can be written in this form. A very prominent example of non-derogatory ES states would be the three-qubit W state, $\ket{W} = \ket{001} +\ket{010} +\ket{100}$, which is of the form given in Eq.\;(\ref{eq:ek}) for $k=1$. It can be easily verified that, $B\otimes B^{-1}$ with $B = \begin{pmatrix}
	x&1      \\
	0 &x
	\end{pmatrix}$ is a symmetry. Moreover, the 3-qubit GHZ state $\ket{GHZ} = \ket{000} +\ket{111}$ is of the form in Eq.\;(\ref{eq:sumek}) for $K=2$ and $k_1=k_2 = 1$, so it is also ES. Indeed, the corresponding symmetry is with $B = \begin{pmatrix}
	x&0      \\
	0 &y
	\end{pmatrix}$. Actually, these are the only two possible non-trivial Jordan normal forms of $(2 \times 2)$ matrices. Hence, the W and the GHZ class are the only two ES classes of 3- (or $n$-) qubit states, so there exists no derogatory classes in the $n$-qubit case. To further clarify the notation for the excited states in Eq.\;(\ref{eq:ekderog}), we consider a rather lengthy but illuminating example of three four-level systems in derogatory ES classes.
	Let us consider $B = J_2 \oplus J_2$, where as before, we use the notation $J_k$ for a $(k \times k)$-Jordan block with eigenvalue 1. Then, $ \ket{E_1^{3,0}}$, $\ket{E_1^{0,3}}$, $\ket{E_1^{2,1}}$, $\ket{E_1^{1,2}}$, $\ket{E_0^{3,0}}$, $\ket{E_0^{0,3}}$, $\ket{E_0^{2,1}}$, and $\ket{E_0^{1,2}}$ are the states given by  Eq.\;(\ref{eq:ekderog}) for all possible configurations $k,n_1,n_2$ for the given $B$ matrix. Using the mapping $|i^{(b)}\rangle\to|b-1\rangle\otimes|i\rangle$ these states can be rewritten as
	\begin{align}
		\ket{E_0^{3,0}} &= \ket{000}_{A_1B_1C_1} \otimes \ket{000}_{A_2B_2C_2} \nonumber\\
		\ket{E_0^{2,1}} &= \ket{W}_{A_1B_1C_1} \otimes \ket{000}_{A_2B_2C_2} \nonumber\\
		\ket{E_0^{1,2}} &= \ket{\bar{W}}_{A_1B_1C_1} \otimes \ket{000}_{A_2B_2C_2} \nonumber\\
		\ket{E_0^{0,3}} &= \ket{111}_{A_1B_1C_1} \otimes \ket{000}_{A_2B_2C_2}\nonumber \\
		\ket{E_1^{3,0}} &= \ket{000}_{A_1B_1C_1} \otimes \ket{W}_{A_2B_2C_2} \nonumber\\
		\ket{E_1^{2,1}} &= \ket{W}_{A_1B_1C_1} \otimes \ket{W}_{A_2B_2C_2}\nonumber \\
		\ket{E_1^{1,2}} &=\ket{\bar{W}}_{A_1B_1C_1} \otimes \ket{W}_{A_2B_2C_2} \nonumber\\
		\ket{E_1^{0,3}} &= \ket{111}_{A_1B_1C_1} \otimes \ket{W}_{A_2B_2C_2},\label{eq:derogEG1}
	\end{align}
	where $\ket{\bar{W}} = \sigma_x^{\otimes 3} \ket{W}$, $\sigma_x$ denotes the Pauli X matrix. Here, the first party consists of the subsystems $A_1$ and $A_2$, the second has $B_1$ and $B_2$ and the third party holds the systems $C_1$ and $C_2$.
	As mentioned before, in contrast to the non-derogatory case, it is not the case that the set $\{\ket{E_k^{n_1, \ldots, n_K}}\}_{k, (n_1 \ldots, n_K)}$ gives representatives for all SLOCC classes in the derogatory case \cite{footnote5}. Consider the state
	\begin{align}
		&\ket{E_1^{2,1}} + \ket{E_0^{3,0}} + \ket{E_0^{0,3}}  \nonumber \\
		& \qquad = \ket{GHZ} \otimes \ket{000} +  \ket{W} \otimes \ket{W},\label{eq:derogEG2}
	\end{align}
	which as well gives rise to a derogatory class, but is not SLOCC-equivalent to any of the $\ket{E_k^{n_1,n_2}}$ listed above. This illustrates that superpositions as in Eq.\;(\ref{eq:ekderogsuperpos}) indeed need to be considered. Moreover, one additional difficulty arises. There is no guarantee that the superpositions in Eq.\;(\ref{eq:ekderogsuperpos}) would give rise to derogatory classes. In contrast to the example above, they may happen to coincide with a non-derogatory class. Consider the example
	\begin{align}
		\ket{E_1^{3,0}} + \ket{E_1^{0,3}} &= \ket{GHZ} \otimes \ket{W},\label{eq:derogEG3}
	\end{align}
	or $\ket{E_1^{2,1}} + \ket{E_1^{1,2}}$, which are both equivalent to $\ket{W} \oplus \ket{W}$. Therefore, both states are in the same non-derogatory ES SLOCC class as $\ket{W} \oplus \ket{W}$ which is non-derogatory.

	\subsection{Entanglement and state transformations}
	
	As mentioned in the introduction, the study of LOCC transformation is central in entanglement theory, as LOCC cannot increase entanglement. That is, if one state $\ket{\Psi}$ can be transformed into another state $\ket{\Phi}$ via LOCC (deterministically), then the entanglement, measured by some entanglement measure, $E$, of $\ket{\Psi}$ is larger or equal to the entanglement contained in $\ket{\Phi}$, i.e. $E(\ket{\Psi})\geq E(\ket{\Phi})$. Importantly, this holds for any entanglement measure, $E$. Hence, the study of LOCC transformations allows the comparison of the entanglement content of states. The simplest example of such transformations are local unitaries. However, as unitary operators are invertible, they do not change entanglement and are hence irrelevant here. 
	
	The existence of a non-trivial transformation presupposes the existence of local symmetries of the states, which can be easily seen as follows. Let us first consider the case of finitely many rounds of LOCC, i.e., $LOCC_{\mathbb{N}}$. If $\ket{\Psi}$ can be transformed non-trivially into $\ket{\Phi}$ via $LOCC_{\mathbb{N}}$, then there have to exist at least two distinct local operators, $M^{(1)}_i\otimes M^{(2)}_i\otimes \ldots \otimes  M^{(n)}_i$, for $i=1,2$ such that $M^{(1)}_i\otimes M^{(2)}_i\otimes \ldots \otimes  M^{(n)}_i \ket{\Psi}=\alpha_i 
	\ket{\Phi}$, with $\alpha_i\neq 0$. As we consider here fully entangled states, the local operators have to be invertible. Hence, the equations above imply that $\ket{\Psi}$ needs to possess a local symmetry. 
	
	Furthermore, as the LOCC transformation has to be deterministic, the measurement operators need to satisfy the completeness relation. To complicate matters, the measurements need to be implementable via LOCC. That is, one party performs a measurement and communicates the outcome $i$ to the other parties who then perform a measurement depending on $i$. This process terminates once all possible states which are generated during the process coincide with the desired final state. The fact that all local measurements have to obey the completeness relation make the characterization of LOCC protocols so complicated. 
	
	As one can see from the explanation above, it is the local symmetries of the states which make a transformation possible or not. In case there exist only finitely many local symmetries, a simple characterization of all states which can be reached from any other state via $LOCC_{\mathbb{N}}$ have been presented in Refs.\;\cite{SpdV17,dVSp17}. Moreover, the necessary and sufficient conditions for the existence of a LOCC transformation from a given state to any other state have been derived there for a single round of LOCC, i.e., $LOCC_1$, where one party implements a non-trivial measurement and the other parties apply (depending on the measurement outcome) a LU. We will extend these results to infinite symmetry groups here. 
	
	It should also be noted here that in case infinitely many rounds of LOCC are considered, it is still unclear whether non-invertible local matrices need to be considered as well \cite{HeEn21}. However, even in that case, it has been shown that for homogeneous systems, almost no state can be transformed into any other state \cite{GoKr17,SaWa18}. That is, almost all states are isolated. The reason is that almost no state possesses a local symmetry (other than the identity operator). 
	
	Let us mention here that the reason for considering $LOCC_1$ transformations in Refs.\;\cite{SpdV17,dVSp17}  and also here is that, apart from one exception, all previously studied LOCC transformations are all-deterministic. That is, in each round of the LOCC protocol the state is transformed deterministically into another state. Despite the fact that we will also present here some examples of pairs of states which can not be transformed into each other with such a protocol, we will mainly focus on transformability via $LOCC_1$. Due to that, we introduce here a notion of weak isolation. A state is called \textit{weakly isolated} if it is neither $LOCC_{\mathbb{N}}$-reachabe nor  $LOCC_1$-convertible. 
	
	The fact that LOCC enables to sort the states according to the entanglement they contain motivates a further investigation of LOCC transformation among states which are physically relevant. Knowing which state is more entangled than another one will also ease the discovery of new (operational) entanglement measures. In this context, it is worth mentioning that the study of LOCC and their probabilistic counterpart, SLOCC transformations, led to a complete set of easily computable entanglement measures for almost all $n$ qudit states \cite{SaSc18}.
	
	In this work, we study the entanglement of symmetric states, more precisely, we analyze which state transformations are possible via LOCC within SLOCC classes that contain symmetric states. As mentioned before, the motivation for this is physical and manifold. On the one hand, symmetric states naturally appear in physically relevant systems. In certain contexts, they display interesting physical features such as e.g. superradiance \cite{Dicke,HeRi}. Furthermore, prominent symmetric states such as e.g. the GHZ-state, the W-state, or Dicke states have found applications in quantum information (see e.g. \cite{Ca02}). All this suggests that a thorough study of entanglement in symmetric states in general will prove to be very fruitful. 
	
	\section{Results}
	\label{sec:results}
	
	\begin{figure*}[htb]
		\includegraphics[width=1.00\linewidth]{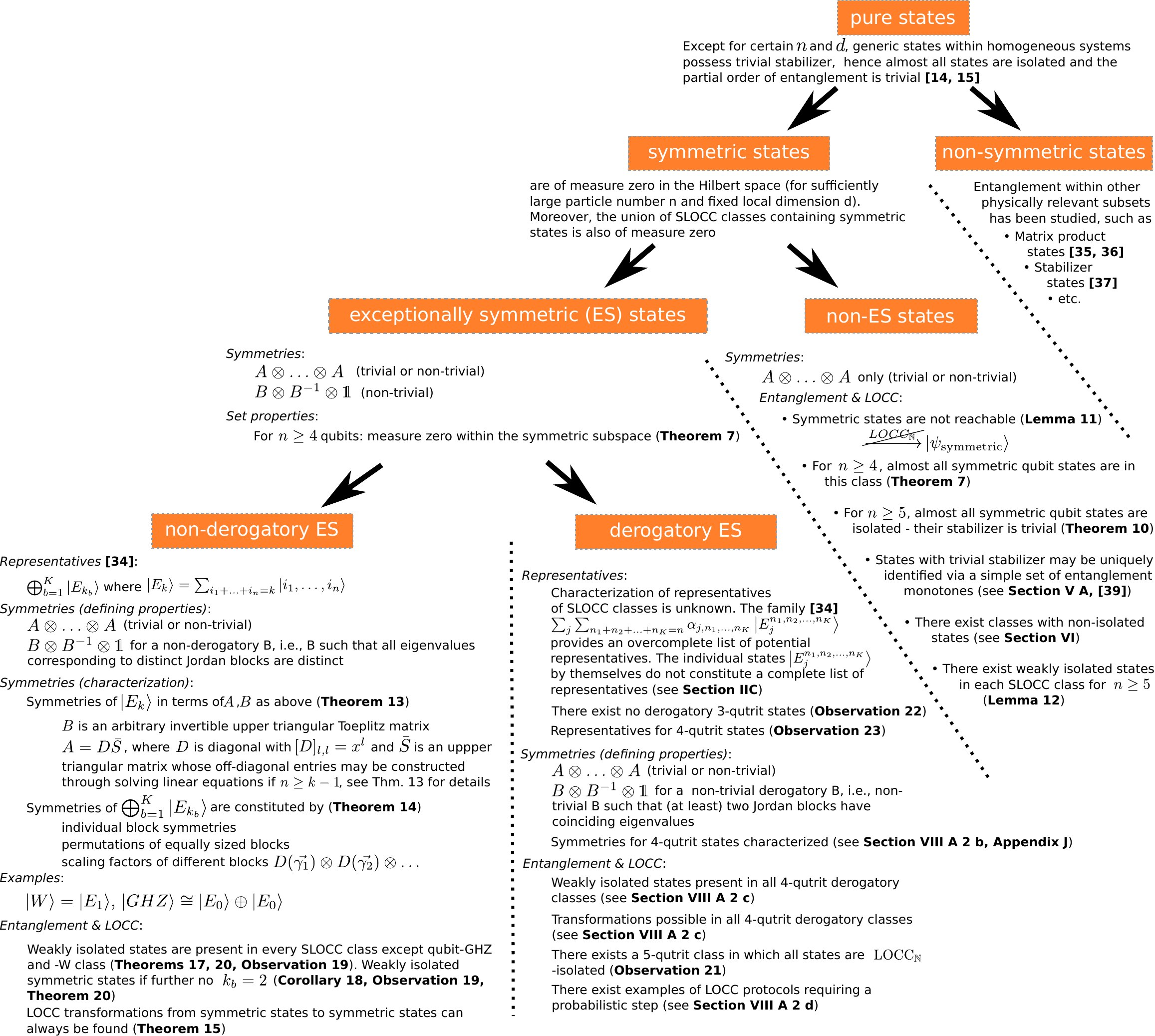}
		\caption{In this figure, we illustrate the classification of symmetric states and summarize the results. Where not stated otherwise, the statements hold for arbitrary local dimension $d$ and arbitrary particle numbers $n>3$.}
		\label{fig:overview}
	\end{figure*}
	
	In this section, we summarize the main results of this work [also see Fig.\;\ref{fig:overview} for a compact summary]. We focus on studying the possible LOCC transformations among and within SLOCC classes of pure symmetric states. We are particularly interested in finding SLOCC classes that exhibit a rich structure of possible LOCC transformations. As explained above, these transformations impose a partial order on the Hilbert space, which must be respected by any entanglement measure. In the course of this investigation, we characterize the local symmetries of those subsets of states, which are most promising regarding the set of possible LOCC transformations. 
	
	As mentioned in the introduction, it has been shown that in homogeneous systems constituted of $n>3$ parties with local dimension $d>2$ (as well as systems for which $d=2$, $n>4$), LOCC transformations among pure states are almost never possible \cite{GoKr17,SaWa18}. This is a consequence of the fact that generic states of such systems possess a trivial stabilizer. Note that this statement is not necessarily true for subsets of the Hilbert space of measure zero and, clearly, the set of symmetric states forms such a subspace of measure zero. However, we show that---considering qubits---the statement applies to the symmetric subspace considering $n\geq5$ particles (Theorem \ref{thm:genericsymmetricisolation}). To spell it out, state transformations involving symmetric qubit states are almost never possible as almost no such state possesses a non-trivial local symmetry. In this sense, the set of symmetric qubit states reflects the properties of the full Hilbert space. In turn, this result provides a simple and complete characterization of entanglement of generic symmetric qubit states.
	
	Compared to the generic case, the situation in the SLOCC classes corresponding to the $n$-qubit GHZ and W states is as different as it could possibly be. It is well-known that these two classes do not contain any isolated state, that is, a state that can neither be reached from any other state nor converted into any other state via LOCC (non-trivially). In other words, there, every state can take part in non-trivial state transformations (while, generically, no state can). This motivates a detailed investigation of zero measure subsets within the symmetric subspace. The classification of symmetric states into ES and non-ES states (see Section \ref{subsec:Prelim_ESandNES}) is well suited to identify the desired subset. The reason for that is the following. Note that in case the stabilizer is non-trivial, non-trivial LOCC transformations typically do exist. As explained before, whether a state can be transformed into another state via LOCC highly depends on the properties of the stabilizer \cite{GoWa11, dVSp13,HeEn21}. Hence, the classification into ES and non-ES states, which is purely based on properties of the stabilizer, is ideally suited here. For ES states, the symmetry group is non-unitary and infinite. In Section \ref{sec:reachability}, we hence generalize the tools for studying finite-round LOCC (LOCC$_\mathbb{N}$) transformations, which have been introduced for the case of a finite stabilizer \cite{SpdV17,dVSp17} \cite{footnote6}, to the case of an infinite (not necessarily unitary) stabilizer (Theorem \ref{thm:reachability} and Lemma \ref{lemma:convertibility}). We study the implications of these conditions  on the symmetries in Observations \ref{obs:quasicommutation2} and  \ref{obs:quasicommutation3}. Moreover, we give a simple necessary and sufficient condition for when a state is weakly isolated, i.e., isolated with respect to  LOCC$_\mathbb{N}$-reachability and LOCC$_1$-convertibility (Lemma \ref{lemma:isolation}).
	
	Using these results we then study the symmetries and possible LOCC transformations within the aforementioned classes of symmetric states. 
	
	{\it Non-ES states:} Recall that non-ES states are those states that have symmetries of the form $A^{\otimes n}$ only. Note that symmetric states with only the trivial stabilizer fall into this class. As mentioned before, for these states a complete, easily computable set of entanglement measures is known and the maximal success probability of transforming one state into the other can be easily computed \cite{SaSc18}.
	We show that non-ES states are never reachable via $LOCC_\mathbb{N}$ (Lemma \ref{lemmaSym}), even if the stabilizer is non-trivial, which implies that non-trivial finite round LOCC transformations among two non-ES states are never possible. Moreover, we show that every SLOCC class containing an non-ES state for $n\geq5$ contains weakly isolated states (Lemma \ref{lemma:non-es-isolation}). In the following we will thus focus on the most interesting case, the ES states.

	{\it ES states:} ES states have typically a much larger stabilizer compared to non-ES states, as the symmetry groups are generated by symmetries of the form $B \otimes B^{-1}$ in addition to $A^{\otimes n}$ (in the sense explained in the preliminaries). Recall that ES states are partitioned into the two sub-categories of non-derogatory ES and derogatory ES states, depending on the properties of $B$ (see Section \ref{sec:preliminaries}). We deal with these two types separately.
	
	{\it Non-derogatory ES states:} Regarding a broad variety of possible LOCC transformations, non-derogatory ES states appear particularly promising, as the W and the GHZ state, whose SLOCC classes do not contain any isolated states, are examples of non-derogatory ES states. Recall that all the SLOCC classes and corresponding representatives are known (see Section \ref{sec:preliminaries}). We characterize the stabilizer for all of them (Theorems \ref{lemma:eksyms} and \ref{lemma:eksumsyms}). As the W and GHZ states are two particular cases of non-derogatory ES states, the derived stabilizer serves as a generalization of the (particularly rich) symmetry groups of both the W and GHZ states.
	We show that within each of the SLOCC classes of non-derogatory ES states, LOCC transformations among two symmetric states are possible (Theorem \ref{thm:symtosym}). Nevertheless, we also show that in each of these SLOCC classes, except for the qubit-W and the qubit-GHZ classes, weakly isolated states are present (Theorems \ref{theo:isolation} and \ref{theo:sumsisolation}). This shows that despite the existence of the exceptionally large stabilizer within all the non-derogatory ES SLOCC classes, the qubit GHZ and qubit W classes are the only classes that are free of weak isolation.

	{\it Derogatory ES states:} Finally, we deal with derogatory ES states. Recall that the symmetries of the form $B \otimes B^{-1}$ are more restricted here than in the respective non-derogatory ES counterparts, as at least one eigenvalue needs to be degenerate. Several difficulties arise, as pointed out in the preliminaries. In particular, the characterization of derogatory ES SLOCC classes is very delicate. This is due to the facts that neither the representatives of derogatory ES SLOCC classes are known for arbitrary system sizes, nor is there a simple and sufficient condition for a state to fall into the derogatory ES class. Due to that, our main focus here will be to demonstrate the properties of those states with the help of examples. Among $n$-qubit and three-qutrit states (Observation\;\ref{Obs:3Qutrit}), there does not exist any derogatory ES state. The smallest system comprising derogatory ES states is a four-qutrit system, which we analyze in detail in Sec. \ref{subsec:4qutritDerog}. First, we characterize the symmetry groups of all derogatory ES four-qutrit states. Then, we show that in each of the classes, symmetric weakly isolated states are present. Moreover, we show that certain non-trivial LOCC transformations are possible (Observation\;\ref{Obs:4Qutrit}). Further, we identify a derogatory ES 5-qutrit state representing an SLOCC class in which all states are isolated with respect to LOCC$_\mathbb{N}$-reachability and -convertibility (Observation \ref{obs:derogatoryisolation}). This SLOCC class together with the SLOCC classes of the GHZ and W--state (qubits) therefore represent the two extreme cases of SLOCC classes. 
	
	We are presenting these results by first considering qubit states, for which we show that the generic results hold also true for symmetric states (Sec.\;\ref{sec:symmetric}). Then, we move on to non--ES states, where we show that the symmetries are still too restrictive to allow a rich LOCC structure to exist. Nevertheless, we also present a family of symmetric states which are indeed LOCC-convertible. States having the most symmetries and hence the most promising candidates for rich LOCC structures are the ES states and among them actually the non-derogatory ES states. We scrutinize this class of states and derive all local symmetries of them. Their symmetries are so rich that even LOCC transformations among symmetric states can be derived. However, also in these classes, with the exception of the GHZ and W classes, we show that weakly-isolated states always exist.

	As mentioned above, the symmetries of the form $B \otimes B^{-1}$ of derogatory ES states are always strict subsets of their non-derogatory ES counterparts. However, this is not necessarily true for the symmetries of the form $A^{\otimes n}$. Thus, the conclusion that derogatory ES classes exhibit a smaller stabilizer and hence  restricted possibilities of LOCC transformations compared to their non-derogatory counterparts cannot be  drawn. However, all the results presented here indicate that, typically, the structure of possible LOCC transformations is richer within non-derogatory ES classes compared to derogatory ES classes. To gain insight into the possible structures within derogatory ES classes, we analyze in detail small system sizes and show that, despite the fact that local symmetries exist, there exist SLOCC classes where all states are isolated.

	\section{$LOCC_{\mathbb{N}}$-reachability for non-unitary stabilizer}
	\label{sec:reachability}
	
	In this section, we generalize two results obtained by some of us in \cite{SpdV17,dVSp17} upfront. There, LOCC protocols involving a finite number of rounds of classical communication, denoted by $LOCC_\mathbb{N}$, are studied. Due to practical reasons, we consider here $LOCC_\mathbb{N}$ despite the fact that LOCC protocols may, in general, include an infinite number of such rounds. For more details on the difference between LOCC and $LOCC_\mathbb{N}$, see \cite{SpdV17,dVSp17,Ch11}.
	
	In  Theorem 1 in \cite{SpdV17,dVSp17} necessary and sufficient conditions for reachability under $LOCC_\mathbb{N}$ have been derived for the case of a unitary stabilizer. Here, we generalize the theorem to an arbitrary stabilizer.
	Moreover, we also generalize Lemma 3 therein, providing necessary and sufficient conditions for $LOCC_1$-convertibility, i.e., convertibility under a nontrivial operation performed by one party  \cite{footnoteLOCC1}, to the scenario of an infinite stabilizer. We heavily make use of these generalization derived here, which clearly have significance beyond the present work.
	
	\begin{theorem} 
		\label{thm:reachability}
		A state $\ket{\Phi}\propto h \ket{\Psi_s}$ is reachable via $LOCC_{\mathbb{N}}$, iff there exists $S \in S_{\Psi_s}$ such that the following conditions hold up to permutations of the particles:
		\begin{itemize}
			\item[(i)] For any $i\geq 2$, $ (S^{(i)})^\dagger H_i S^{(i)} \propto H_i$ and
			\item[(ii)] $  (S^{(1)})^\dagger H_1 S^{(1)} \not\propto H_1$.
		\end{itemize}
	\end{theorem}
	
	The proof works very similarly to the case when the stabilizer is finite (and can therefore be chosen to be unitary) \cite{SpdV17} and is presented in Appendix \ref{app:proof_reach_conv}.

	\begin{lemma}
		\label{lemma:convertibility}
		A state $\ket{\Psi}\propto g \ket{\Psi_s}$ is convertible via $LOCC_{1}$ iff there exist $m$ symmetries $S_k \in S_{\Psi_s}$, with $m>1$ and $H_1>0$ and $p_k > 0$ with $\sum_{k=1}^m p_k = 1$, such that the following conditions hold up to permutations of the particles:
		\begin{itemize}
			\item[(i)] $(S_k^{(i)})^\dagger G_i S_k^{(i)} \propto G_i$ for any $i\geq 2$  and for all $k \in \{1, \ldots, m\}$ and
			\item[(ii)] $G_1 = \sum_{k=1}^m p_k \left( S_k^{(1)} \right)^\dagger H_1 S_k^{(1)}$ and $H_1 \not\propto \left(S^{(1)}\right)^\dagger G_1 S^{(1)}$ for any $S \in S_{\Psi_s}$ fulfilling $(S^{(i)})^\dagger G_i S^{(i)} \propto G_i$ for all $i\geq 2$.
		\end{itemize}
	\end{lemma}
	The proof is analogous to the proof of Lemma 3 in \cite{dVSp17} and can be found in Appendix \ref{app:proof_reach_conv}.
	
	Let us consider a simple example to illustrate how these conditions allow one to decide reachability and convertibility. Assume that the symmetries of the state $\ket{\Psi}$ are $P_z^{\otimes n}$ where $P_z=\operatorname{diag}(z, 1/z)$ and $z\in \mathbb{C}$ and consider the state $\otimes_i h_i\ket{\Psi}$  with $h_i$ diagonal and  $h_i\neq \identity$  for all $i$. Then, using Theorem \ref{thm:reachability}, it is straightforward to show that these states cannot be reached with $LOCC_\mathbb{N}$ as all $H_i$  are diagonal and therefore either (for $|z|\neq 1$) $P_z^\dagger H_i P_z\not\propto H_i$ for all $i$  or  (for $|z|=1$) $P_z^\dagger H_i P_z\propto H_i$ for all $i$. However, this state is convertible within a single round to another pure state, which can be easily seen as follows. The condition (i) in Lemma \ref{lemma:convertibility} holds true if e.g. $S_0=\identity^{\otimes n}$ and $S_1=\sigma_z^{\otimes n}$. Then, by choosing $\tilde{H}_1=H_1+\alpha \sigma_x$ with $\alpha\neq 0$ (but small enough such that $\tilde{H}_1>0$) it holds that $H_1 = 1/2 \sum_{k=1}^m \left( S_k^{(1)} \right)^\dagger \tilde{H}_1 S_k^{(1)}$ and $\tilde{H}_1\not\propto   \left(S^{(1)}\right)^\dagger H_1 S^{(1)}$. Hence, also condition (ii) holds true.
	
	As seen from above, a crucial property for reachability (and deterministic convertibility within one round) is the existence of a symmetry $S=\bigotimes_i S^{(i)}$ such that 
	\bea \label{eq_quasi}(S^{(i)})^\dagger H_i S^{(i)} \propto H_i
	\eea(and analogously for $G_i$). We call this relation a quasi-commutation relation. It can only hold true if the operators $S^{(i)}$ and $H_i$ are related to each other in a very specific way as the following observations show.
	
	\begin{observation}
		\label{obs:quasicommutation2}
		Let $A$ be a positive $k \times k$ matrix and $B$ be an arbitrary $k \times k$ matrix. Then, $B^\dagger A B \propto A$ if and only if
		$B \propto a^{-1} U a$ for some unitary $U$ and some $a$ s.t. $a^\dagger a = A$.
	\end{observation}
	
	\begin{observation}
		\label{obs:quasicommutation3}
		Let $A$ be a positive $k \times k$ matrix and $B$ be an arbitrary $k \times k$ matrix. Then, $B^\dagger A B \propto A$ if and only if the following two conditions are met. 
		\begin{enumerate}[(i)]
			\item $B$ is (up to proportionality) similar to a unitary matrix, i.e., $B \propto R \operatorname{diag}(e^{i \phi_1}, \ldots, e^{i \phi_k}) R^{-1}$ for some $k \times k$ matrix $R$ and $\phi_1, \ldots, \phi_k \in [0,2\pi)$.
			\item $A \propto R^{-\dagger} X R^{-1}$, where $X$ is a direct sum of positive matrices acting on the degenerate subspaces of $\operatorname{diag}(e^{i \phi_1}, \ldots, e^{i \phi_k})$.
		\end{enumerate}
	\end{observation}
	The proofs of these observations can be found in Appendix \ref{app:quasicommutation2}.
	Note that in the particular case in which the eigenvalues of $B$ are non-degenerate, i.e., $\phi_i$ are mutually different, $X$ must be a diagonal matrix with positive entries. Note further, that these observations directly imply the following corollary.
	
	\begin{corollary}
		\label{lemma:quasicommutation}
		For a full rank, non-diagonalizable $k \times k$ matrix $B$, there exists no positive definite $k \times k$ matrix $A$ such that the equation $B^\dagger A B \propto A$ is satisfied. 
	\end{corollary}

	Hence, in order for Eq. (\ref{eq_quasi}) to hold true $S^{(i)}$ has to be diagonalizable. Note that the symmetries $B\otimes B^{-1}$ of the ES classes are in general not necessarily diagonalizable and therefore it is a priori not clear that non-trivial transformations are possible in ES classes. An exception here is e.g. the GHZ state (of arbitrary dimension), for which the symmetry $B\otimes B^{-1}$ is diagonal.
	
	The conditions in Theorem \ref{thm:reachability} and Lemma \ref{lemma:convertibility} can be employed to identify states that are neither reachable with a finite-round LOCC protocol nor deterministically transformable within one round to another pure state. We call such states in the following weakly isolated states. The following lemma provides a complete characterization of such states. 
	
	\begin{lemma} 
		\label{lemma:isolation}
		A state $\ket{\Phi}\propto g \ket{\Psi_s}$ is weakly isolated if and only if there exists no $S \in S_{\Psi_s} \setminus\{\identity\}$ such that $ (S^{(i)})^\dagger G_i S^{(i)} \propto G_i$ for at least $n-1$ sites $i$.
	\end{lemma}
	\begin{proof}
		We prove this by showing that a state is not weakly isolated iff there exists some $S \in S_{\Psi_s} \setminus\{\identity\}$ such that $ (S^{(i)})^\dagger G_i S^{(i)} \propto G_i$ for $n-1$ sites $i$. Obviously, if $\ket{\Phi}$ is non-isolated, then there must exist some $S \in S_{\Psi_s}$ such that $ (S^{(i)})^\dagger G_i S^{(i)} \propto G_i$ for $n-1$ sites $i$ due to Theorem \ref{thm:reachability} and Lemma \ref{lemma:convertibility}.
		Let us now prove the converse direction. To this end, we assume that there exists some $S \in S_{\Psi_s}$ such that $ (S^{(i)})^\dagger G_i S^{(i)} \propto G_i$ for all $i \in \{1, \ldots, n\} \setminus \{j\}$ and show that the state $\ket{\Phi}$ is non-isolated. Let us distinguish two cases, $ (S^{(j)})^\dagger G_j S^{(j)} \not\propto G_j$ and $ (S^{(j)})^\dagger G_j S^{(j)} \propto G_j$. In the first case, $ (S^{(j)})^\dagger G_j S^{(j)} \not\propto G_j$, it follows from Theorem \ref{thm:reachability} that this state is reachable. In the second case, $ (S^{(j)})^\dagger G_j S^{(j)} \propto G_j$, we can construct an operator $H_j$ such that $G_j = \sum_{k} p_k \left( S_k^{(j)} \right)^\dagger H_j S_k^{(j)}$ with $S_k^{(j)}= (S^{(j)})^k$ and $H_j \not\propto \left(S^{(j)}\right)^\dagger G_j S^{(j)}$.
		Then either condition (ii) in Lemma \ref{lemma:convertibility} is fulfilled and $\ket{\Phi}$  is convertible, or another symmetry $\bar{S}$ fulfilling the conditions in Theorem \ref{thm:reachability} can be constructed and $\ket{\Phi}$  is reachable.
		Note first that due to the fact that $S^{(j)}$ fulfills a quasi-commutation relation it has to be  of the form $S^{(j)}\propto RDR^{-1}$ with $D=\operatorname{diag}(1, e^{i \alpha},\ldots)$ with $\alpha \neq 0$ (up to reordering of the eigenvalues; see Observation \ref{obs:quasicommutation3}).  Then we choose the operator $H_j=G_j + \beta X$ with $X=R^{-\dagger} (\sigma_x\oplus 0) R^{-1}$, $\beta\neq 0$ such that $H_j>0$, and the probability distribution $p_k$ such that $\sum p_k e^{ik\alpha}=0$. Note that such a probability distribution always exists for any $\alpha \neq 0$. With this we have that $\sum p_k (S^{(j) \dagger})^k X(S^{(j)})^k=0$ and we obtain $G_j = \sum_{k} p_k \left( S_k^{(j)} \right)^\dagger H_j S_k^{(j)}$. Using that $ (S^{(j)})^\dagger G_j S^{(j)} \propto G_j$ it is  easy to see that for $\beta \neq 0$ it holds that $H_j \not\propto    \left(S^{(j)}\right)^\dagger G_j S^{(j)}$ If there does not exist any (other) symmetry $\tilde{S} \in S_{\Psi_s}$ such that $\left(\tilde{S}^{(j)}\right)^\dagger G_j \tilde{S}^{(j)} \propto H_j$ and $\left(\tilde{S}^{(i)}\right)^\dagger G_i \tilde{S}^{(i)} \propto G_i$ for all $i \in \{1, \ldots, n\} \setminus \{j\}$, condition (ii) in Lemma \ref{lemma:convertibility} is fulfilled and $\ket{\Phi}$ is hence convertible. Let us finally consider the case that such an $\tilde{S}$ exists. Then, $\bar{S} = S \tilde{S}^{-1}$ is also a symmetry of $\ket{\Psi_s}$ and we have that $\left(\bar{S}^{(i)}\right)^\dagger G_i \bar{S}^{(i)} \propto G_i$ for all $i \in \{1, \ldots, n\} \setminus \{j\}$. Moreover, we have that $\left(\bar{S}^{(j)}\right)^\dagger G_j \bar{S}^{(j)} = \left(\tilde{S}^{(j)}\right)^{-\dagger } \left(S^{(j)}\right)^\dagger G_j S^{(j)} \left(\tilde{S}^{(j)}\right)^{-1} \not\propto \left(\tilde{S}^{(j)}\right)^{-\dagger } H_j  \left(\tilde{S}^{(j)}\right)^{-1} \propto G_j$. Hence, the symmetry $\bar{S}$ fulfills the conditions in Theorem \ref{thm:reachability} and  $\ket{\Phi}$ is hence reachable.
		\end {proof}

		\section{Generic stabilizer of qubit SLOCC classes containing permutation symmetric states}
		\label{sec:symmetric}
		
		As the set of symmetric qubit states is of measure zero, it could well be that symmetric states do have non--trivial stabilizers generically. We will show now that this is in fact not the case and that the stabilizer for generic states in $\operatorname{Sym}^n(\mathbb{C}^2)$ is trivial when $n\geq5$.
		Hence, in the qubit case the statement that almost no state transformations are possible also applies to the union of SLOCC classes containing symmetric states. It should be noted here that the claim of \cite{Sh18} that almost all states in $\operatorname{Sym}^n(\mathbb{C}^2)$ for $n\geq5$ have a trivial stabilizer relies on extra assumptions. We will present here our own proof of this fact and additionally prove that generic states in $\operatorname{Sym}^4(\mathbb{C}^2)$ have a non-trivial stabilizer. To this end, we first show that non-ES states, i.e., states which only possess symmetries of the form $S^{\otimes n}$, are generic and then prove the statement that the stabilizer is generically trivial. Finally, we investigate the entanglement contained in symmetric qubit states.

		Let us start out by showing that ES qubit states are of measure zero.

		\begin{theorem}\label{esmeasurezero}
			The set of ES symmetric states has Lebesgue measure zero in $\operatorname{Sym}^n(\mathbb{C}^2)$ when $n\geq4$.
		\end{theorem}
		\begin{proof}
			ES symmetric qubit states belong either to the GHZ SLOCC class or to the W SLOCC class (as there exist only two different Jordan forms for $2\times 2$ matrices). The lemma in Appendix \ref{app:GHZ} shows that the first class is zero-measure when $n\geq 4$. From the Majorana representation, it follows that the states with degeneracy configuration different from $\{1,\ldots,1\}$ have measure zero, which implies that the set of states in the W SLOCC class, which all have degeneracy configuration $\{1,n-1\}$, is of zero measure.
		\end{proof}
		
		Due to Theorem \ref{esmeasurezero}, generic states in $\operatorname{Sym}^n(\mathbb{C}^2)$ only possess symmetries of the from $S^{\otimes n}$. However, we show now that the stabilizer is generically even trivial, i.e., there exists no symmetry other than the identity operator, when $n\geq5$. Before presenting a proof for this statement, we find it insightful to explain why this is not the case for $n=3$ and $n=4$.
		
		\begin{observation}\label{obsn=3}
			Generic states in $\operatorname{Sym}^3(\mathbb{C}^2)$ have a non-trivial stabilizer.
		\end{observation}
		\begin{proof}
			Let $|\psi\rangle\in\operatorname{Sym}^3(\mathbb{C}^2)$ have a Majorana representation as in Eq.\;(\ref{majorana}). The assumption of genericity implies that the degeneracy configuration is $\{1,1,1\}$. Thus, all $\{|\epsilon_j\rangle\}$ are different and every pair is a basis of $\mathbb{C}^2$. Hence, $|\epsilon_3\rangle=\alpha|\epsilon_1\rangle+\beta|\epsilon_2\rangle$ for some $\alpha,\beta\in\mathbb{C}$ such that $\alpha,\beta\neq0$. Let now $a_1=\beta/\alpha$ and $a_2=\alpha/\beta$ and define the (unique) matrix $A\in\mathbb{C}^{2\times2}$ such that
			\begin{equation}
			A|\epsilon_1\rangle=a_1|\epsilon_2\rangle,\quad A|\epsilon_2\rangle=a_2|\epsilon_1\rangle.
			\end{equation}
			The fact that $\{|\epsilon_1\rangle,|\epsilon_2\rangle\}$ are linearly independent implies that $A$ is invertible and different from the identity. Moreover, it follows that $A|\epsilon_3\rangle=|\epsilon_3\rangle$. Thus, $A^{\otimes 3}|\psi\rangle=a_1a_2|\psi\rangle=|\psi\rangle$, which proves the claim.
		\end{proof}
		Note that the observation above can also be seen by noting that any three-qubit state has a non-trivial stabilizer \cite{Verstraete2002}.
		\begin{observation}\label{obsn=4}
			Generic states in $\operatorname{Sym}^4(\mathbb{C}^2)$ have a non-trivial stabilizer.
		\end{observation}
		\begin{proof}
			Let $|\psi\rangle\in\operatorname{Sym}^4(\mathbb{C}^2)$ have a Majorana representation as in Eq.\;(\ref{majorana}). Again, the assumption of genericity implies that the degeneracy configuration is $\{1,1,1,1\}$, all $\{|\epsilon_j\rangle\}$ are different and every pair is a basis of $\mathbb{C}^2$. Let $|\epsilon_3\rangle=\alpha|\epsilon_1\rangle+\beta|\epsilon_2\rangle$ and $|\epsilon_4\rangle=\alpha'|\epsilon_1\rangle+\beta'|\epsilon_2\rangle$ with $\alpha,\beta,\alpha',\beta'\neq0$ and define now
			\begin{equation}
			a_1=\sqrt{\frac{\beta\beta'}{\alpha\alpha'}},\quad a_3=\sqrt{\frac{\alpha\beta}{\alpha'\beta'}},
			\end{equation}
			$a_2=1/a_1$ and $a_4=1/a_3$. We consider similarly as before the (unique) matrix $A\in\mathbb{C}^{2\times2}$ such that
			\begin{equation}
			A|\epsilon_1\rangle=a_1|\epsilon_2\rangle,\quad A|\epsilon_2\rangle=a_2|\epsilon_1\rangle,
			\end{equation}
			which must be again invertible and different from the identity. With the aforementioned conditions, it follows that $A|\epsilon_3\rangle=a_3|\epsilon_4\rangle$ and $A|\epsilon_4\rangle=a_4|\epsilon_3\rangle$. Thus, $A^{\otimes 4}|\psi\rangle=a_1a_2a_3a_4|\psi\rangle=|\psi\rangle$, which proves the claim.
		\end{proof}
		We will next address the question whether in the n-qubit case with $n\geq5$ generic symmetric states have a trivial stabilizer. This question was already addressed  in \cite{Sh18}. We nevertheless provide here our own proof which is not based on any assumption \cite{footnote7}.
		\begin{theorem}
			\label{thm:genericsymmetricisolation}
			Generic states in $\operatorname{Sym}^n(\mathbb{C}^2)$ with $n\geq5$ have a trivial stabilizer.
		\end{theorem}
		\begin{proof}
			As mentioned before, for generic states it suffices to consider states with Majorana representation with degeneracy configuration $\{1,\ldots,1\}$, which have the property that any pair $\{|\epsilon_i\rangle,|\epsilon_j\rangle\}$ ($i\neq j$) is linearly independent and spans $\mathbb{C}^2$. Moreover, due to Theorem \ref{esmeasurezero}, to prove our claim it is enough to show that generic states do not have a non-trivial symmetry of the form $A^{\otimes n}|\psi\rangle=|\psi\rangle$ since we only need to consider non-ES states. Now, looking into the Majorana representation, the above equation translates to
			\begin{equation}
			\sum_\pi P_\pi(|\epsilon_1\cdots\epsilon_n\rangle)=\sum_\pi P_\pi(A|\epsilon_1\rangle\cdots A|\epsilon_n\rangle).
			\end{equation}
			As \cite{Sh18} observes the uniqueness of this representation imposes then that
			\begin{equation}\label{condepsilon}
			A|\epsilon_j\rangle=a_j|\epsilon_{\pi(j)}\rangle
			\end{equation}
			where $\pi$ is a permutation and $\{a_j\neq0\}$ can be freely chosen in $\mathbb{C}$ \cite{footnote8}. Notice that the permutation $\pi$ must decompose into $t$ $s_i$-cycles such that $\sum_{i=1}^ts_i=n$.
			
			If there are two 2-cycles (say without loss of generality for $j=1,2,3,4$), then arguing as in Observation \ref{obsn=4}, this fixes $A$ uniquely up to an irrelevant proportionality constant and, thus, Eq.\ (\ref{condepsilon}) can only hold for $j>4$ for a set of states of measure zero.
			
			If there is a 1-cycle and a 2-cycle (without loss of generality for $j=1,2,3$), the argument used in Observation \ref{obsn=3} fixes again $A$ uniquely up to a proportionality constant and again Eq.\ (\ref{condepsilon}) can only hold for $j>3$ for a set of states of measure zero.
			
			If there is a 3-cycle (without loss of generality for $(123)$ in cyclic notation), we take again $|\epsilon_3\rangle=\alpha|\epsilon_1\rangle+\beta|\epsilon_2\rangle$ for some $\alpha,\beta\in\mathbb{C}$ such that $\alpha,\beta\neq0$. Then, Eq.\ (\ref{condepsilon}) can only hold for $j=3$ if $a_1=-a_3\beta/\alpha^2$ and $a_2=a_3/(\alpha\beta)$. Thus, the cases $j=1,2$ of the aforementioned equation fix $A$ uniquely up to a proportionality constant and forbid the equation to hold generically when $j>3$.

			The proof is finished by noticing that if $\pi$ contains one $s$-cycle with $s\geq4$, then Eq.\ (\ref{condepsilon}) cannot hold generically for the indices involving this cycle. To see this, we take without loss of generality this cycle to be $(12\cdots s)$ and take $|\epsilon_3\rangle=\alpha|\epsilon_1\rangle+\beta|\epsilon_2\rangle$ and $|\epsilon_4\rangle=\alpha'|\epsilon_1\rangle+\beta'|\epsilon_2\rangle$ with $\alpha,\beta,\alpha',\beta'\neq0$. Equation (\ref{condepsilon}) applied to $j=3$ imposes that
			\begin{equation}
			a_2=\frac{\alpha'}{\alpha\beta}a_3,\quad a_1=\frac{\alpha\beta'-\alpha'\beta}{\alpha^2}a_3.
			\end{equation}
			Notice that $a_1\neq0$ due to the fact that $|\epsilon_3\rangle$ and $|\epsilon_4\rangle$ are linearly independent. The above expressions when plugged into Eq.\ (\ref{condepsilon}) for $j=1,2$ fix again $A$ uniquely up to a proportionality constant. Thus, Eq.\ (\ref{condepsilon}) cannot hold generically for $4\leq j\leq s$ (this is immediate for $s\geq5$, in the case $s=4$ some straightforward algebra shows that generically $A|\epsilon_4\rangle$ will not be proportional to $|\epsilon_1\rangle$ as this requires that $2\alpha\beta'=\alpha'\beta$).
		\end{proof}
		
		\subsection{Entanglement of SLOCC classes of symmetric qubit states}
		
		The results presented above imply that it is generically impossible to transform one symmetric qubit state into another via LOCC (even if infinitely many rounds of LOCC are allowed \cite{GoKr17}). Hence, the partial order of entanglement provided by the study of LOCC transformations is trivial in this case. However, for states with a trivial stabilizer, a complete set of entanglement measures is presented in \cite{SaSc18}. To recall this result here, we denote by $\ket{\Psi_s}$ a representative of an SLOCC class of states which only have trivial symmetries (note that also for higher dimensions, the union of such SLOCC classes is of full measure in the Hilbert space). For an arbitrary normalized state $\ket{\Psi} = g\ket{\Psi_s}$ in such an SLOCC class, a complete set of entanglement monotones is given by the set
		$\{E^{\Psi_s}_{\vec{x_i}}\}_{i\in I}$, with a finite index set $I$ and 
		\bea \label{eq:EnMons} E^{\Psi_s}_{\vec{x}}(\ket{\Psi})=\bra{\vec{x}}G\ket{\vec{x}}.\eea
		Here, $\ket{\vec{x}}$ denotes a product state, i.e., $\ket{\vec{x}}=\ket{x_1}\otimes
		\ldots \otimes \ket{x_n}$, with $\ket{x_i}\in \C^d$ and $G=g^\dagger g=\otimes_{i=1}^d G_i$.  
		As $G$ in Eq. (\ref{eq:EnMons}) is local and $\ket{\vec{x}}$ is a product state, the entanglement monotones can be very easily computed.
		
		Moreover, in the case of generic states, it holds that the maximal success probability of transforming $\ket{\Psi}$ into $\ket{\Phi}$ is given by 
		\bea
		P(\Psi,\Phi) = \min_{\vec{x}} \frac{E^{\Psi_s}_{\vec{x}}(\Psi)}{E^{\Psi_s}_{\vec{x}}(\Phi)}=\frac{\|\Phi\|^2}{\|\Psi\|^2}
		\frac{1}{\lambda_{max}(G^{-1} H)}, \label{eq:optprob}
		\eea
		where $\lambda_{max}(X)$ denotes the maximal eigenvalue of $X$. 
		That is, given the SLOCC class to which a state belongs to, the finite set of entanglement monotones given above completely characterizes the entanglement contained in the state. This is in complete analogy to the bipartite entanglement monotones \cite{Vidal99}. Moreover, in the multipartite case the minimization in Eq. (\ref{eq:optprob}) does not even need to be computed, and is simply given by the RHS of that equation.
		
		For states of the form $g^{\otimes n} \ket{\Psi}$, with $\ket{\Psi}$ to be an arbitrary but known state with trivial stabilizer we have that the quantities 
		\bea \bra{\vec{x}}G\ket{\vec{x}}, \mbox {for } x\in S,\eea 
		where $S$ denotes a tomographically complete set, e.g., the eigenbases of the Pauli operators, completely characterizes the entanglement of a symmetric qubit state.

		\section{Non-exceptionally symmetric classes\label{sec:non_ES}}
		
		Whereas no LOCC transformation is generically possible among qubit states, they might well exist among states in higher dimensional Hilbert spaces. Before we consider ES classes in Sec.\;\ref{sec:es}, let us show here that non-ES classes (with all symmetries of the form  $A^{\otimes n}$) do not allow for many LOCC transformations.  We first show that no symmetric state within those classes is reachable via $LOCC_\mathbb{N}$. Then, we prove that  for $n\geq 5$ in each such SLOCC class weakly isolated states exist. Finally, we show that the symmetries of non-ES states can nevertheless be rich enough to allow for non--trivial transformations. We do so by presenting a family of states within which all symmetric states are LOCC-convertible. 
		
		Due to the restricted symmetries of non--ES state, the following Lemma, which shows that no symmetric non--ES state is reachable via LOCC, follows straightforwardly from Theorem \ref{thm:reachability}.  
		
		\begin{lemma} \label{lemmaSym} Let $\ket{\Psi_s}$ be a non-ES $n$--qudit state ($n>2$). Then any symmetric state within the same SLOCC class, $\ket{\Psi}$, is not reachable via $LOCC_{\mathbb{N}}$.\end{lemma}
		\begin{proof}
			As both, $\ket{\Psi_s}$ and $\ket{\Psi}$ are symmetric states and in the same SLOCC it is possible to find $g\in SL(d)$ such that $\ket{\Psi}=g^{\otimes n} \ket{\Psi_s}$ \cite{MaKr10, MiRo13}. Moreover, as they are non-ES, all the symmetries of $\ket{\Psi_s}$ are of the form $S^{\otimes n}$. Hence, Theorem \ref{thm:reachability} implies that $\ket{\Psi}$ is reachable via $LOCC_{\mathbb{N}}$ iff there exists a symmetry $S^{\otimes n}$ such that $S^\dagger G S \propto G$ and $S^\dagger G S \not\propto G$, which clearly cannot be fulfilled. Hence, none of these states is reachable.
		\end{proof}
		
		We show in Appendix \ref{App:proof_iso_NONES} that in each SLOCC class considered here, there exist weakly isolated states, as stated by the subsequent lemma. 
		
		\begin{lemma}
			\label{lemma:non-es-isolation}
			Let $\ket{\Psi_s}$ be a non-ES $n$--qudit state ($n\geq5$). Then, there exist weakly isolated states within the SLOCC class of $\ket{\Psi_s}$.
		\end{lemma}
		
		The subsequent example will illustrate the difference a single additional symmetry can make in the context of LOCC transformations. We will consider the SLOCC classes given by n-qubit Dicke states, $\ket{D_k^n}$, which are defined as the equal superposition of all computational basis states containing $k$ times $1$. Only in case $k=n/2$ an additional symmetry arises. In all other cases these states represent examples of symmetric qubit states which are isolated. In contrast to that, for $k=n/2$, all symmetric states are LOCC-convertible. 
		
		As shown in \cite{BaKr09}, Dicke states with different $k$ are in different SLOCC classes. Notice that here we consider $k\leq\lfloor n/2\rfloor$ as $\sigma_x^{\otimes n}|D_k^n\rangle=|D_{n-k}^n\rangle$. Moreover, we are going to impose that $k>1$ so that we do not consider ES SLOCC classes (i.e.,\ the W class). We are going to determine now the stabilizer of these classes. These Dicke states belong to non-ES SLOCC classes and any of their symmetry must be of the form $S^{\otimes n}$. Let $s_{ij}$ denote the matrix entries of $S$. Then,
		\begin{align}
			|D_k^n\rangle&=S^{\otimes n}|D_k^n\rangle\nonumber\\
			&=\sum_{j_1,\ldots,j_n}\left(\sum_{i_1+\cdots+i_n = k}\prod_{m=1}^ns_{j_mi_m}\right)|j_1\cdots j_n\rangle.
		\end{align}
		Since $|0\ldots 0\rangle$ is orthogonal to $|D_k^n\rangle$ (recall that $1<k\leq\lfloor n/2\rfloor$), we have that
		\begin{equation}
		0=\sum_{i_1+\cdots+i_n = k}\prod_{m=1}^ns_{0i_m}\propto s_{00}^{n-k}s_{01}^k.
		\end{equation}
		Thus, either $s_{00}=0$ or $s_{01}=0$. On the other hand, using analogously the orthogonality with $|1\ldots 1\rangle$, we obtain that either $s_{10}=0$ or $s_{11}=0$. Since, moreover, $S$ must be invertible, this leaves us with two options:
		$$S=\left(
		\begin{array}{cc}
		s_{00} & 0 \\
		0 & s_{11} \\
		\end{array}
		\right)\textrm{ or }S=\left(
		\begin{array}{cc}
		0 & s_{01} \\
		s_{10} & 0 \\
		\end{array}
		\right).$$
		
		In the case where $S$ is diagonal we see that
		\begin{equation}
		S^{\otimes n}|D_k^n\rangle=s_{00}^{n-k}s_{11}^k|D_k^n\rangle.
		\end{equation}
		Thus, all considered Dicke states $|D_k^n\rangle$ have the symmetry $S_\lambda^{\otimes n}$ with $S_\lambda=\operatorname{diag}(\lambda,\lambda^{(k-n)/k})$ for any $\lambda\in\mathbb{C}$.
		
		In the case where $S$ is not diagonal we have that
		\begin{equation}
		S^{\otimes n}|D_k^n\rangle=s_{10}^{n-k}s_{01}^k|D_{n-k}^n\rangle.
		\end{equation}
		Hence, this cannot be a symmetry unless $k=n/2$.
		
		Let us now show that in the SLOCC class of $|D_{n/2}^n\rangle$ (and, thus, $n$ is even) it turns out that every symmetric state is not isolated. The analysis above tells us that $P_z^{\otimes n}\sigma_x^{\otimes n}$ and $P_z^{\otimes n}$ with $P_z=\operatorname{diag}(z,z^{-1})$ ($z\in\mathbb{C}$) are symmetries. Using that we can always choose a $z$ such that $P_z^\dagger G P_z\in\operatorname{span}\{\identity,\sigma_x\}$, every symmetric state $g^{\otimes n}|D_{n/2}^n\rangle$ can be written as $g_x^{\otimes n}|D_{n/2}^n\rangle$ where $[G_x,\sigma_x]=0$. Lemma \ref{lemma:convertibility}  tells us then that all these states are LOCC-convertible. Let us see now that this is not the case for $k\neq n/2$. The key difference is that $\sigma_x^{\otimes n}$ is not in the stabilizer and all symmetries are of the form $P_z^{\otimes n}$ as established above. Hence, using Lemma \ref{lemma:isolation} it is easy to see that any symmetric state apart from $\ket{D_k^n}$ is weakly isolated. This follows on the one hand from Lemma \ref{lemmaSym}, which implies that none of these states is reachable and on the other hand from the fact that $G_x \neq \identity$ does not quasicommute with a diagonal matrix, which implies that these states are not convertible.

		\section{Non-derogatory exceptionally symmetric classes\label{sec:non_derog}}
		
		As seen from above, non-ES states do not allow many LOCC transformations, even if the symmetries are non-trivial. Let us hence consider the case of ES classes and study their symmetries and possible LOCC transformations. We focus here on the non-derogatory case and study the more involved derogatory case in Sec.\;\ref{sec:derog}. As mentioned in Sec.\;\ref{sec:results}, the classes considered here seem to possess the richest LOCC structure. 
		
		We characterize first the symmetries of non-derogatory ES states (Sec.\;\ref{sec:stabNDES}) and then show that in these classes LOCC transformations among symmetric states are indeed possible (Sec.\;\ref{sec:es}). Note that this enables us to compare the entanglement contained within these states. Moreover, we show also here that, weak isolation always exists. Indeed, symmetric states which are weakly isolated also exist. 
		
		\subsection{The stabilizer}
		\label{sec:stabNDES}
		
		Although ES symmetric states---by definition---have symmetries, those symmetries do not yet guarantee that state transformations are possible. Considering the $LOCC_\mathbb{N}$-reachability conditions from Theorem \ref{thm:reachability}, Corollary \ref{lemma:quasicommutation} shows that symmetries of the form $B \otimes B^{-1}$ (and products thereof) do not allow to satisfy condition (i) of Theorem \ref{thm:reachability}, unless $B$ is diagonalizable, as it is for instance in the case of the generalized GHZ classes. Hence, symmetries of the form $B \otimes B^{-1}$ do not necessarily guarantee $LOCC_\mathbb{N}$-reachability.
		
		However, non-derogatory ES classes do have additional symmetries. From the proof of Theorem 1 in \cite{MiRo13}, it becomes clear that symmetries of the form $A^{\otimes n}$ and $B \otimes B^{-1}$ generate the whole stabilizer (see also Appendix \ref{app:generators}). The following two theorems, which we prove in Appendix \ref{app:nonderogatorysyms}, characterize the symmetries within non-derogatory ES SLOCC classes.

		\begin{theorem}
			\label{lemma:eksyms}
			For $n\geq 3$ the stabilizer of the states $\ket{E_k}$ is generated by 
			operators of the form $B \otimes B^{-1}$ and $A^{\otimes n}$, where
			\begin{itemize}
				\item $B$ is an arbitrary invertible upper triangular Toeplitz matrix. 
				\item $A^{\otimes n}$ can be written up to a proportionality factor as $A=D\bar{S}$ where $D$ is a diagonal matrix with $[D]_{l,l} = x^l$ for some $x \in \mathbb{C}$ and $l \in \{0, 1, \ldots, k\}$ and $\bar{S}$ is upper triangular with $[\bar{S}]_{l,l}=1$ for all $l$. Furthermore, $\bar{S}$ is characterized by $[\bar{S}]_{i+1,l}=\sum_{j=1}^{l-i} [\bar{S}]_{i,l-j} y_j$ for some complex parameters $y_j$ and $\bra{E_l}\bar{S}^{\otimes n}\ket {E_k}=0$ for $0\leq l< k$. For $n\geq k-1$ the symmetries $\bar{S}^{\otimes n}$ can be determined by solving linear equations and (except for a measure zero subset) any choice of $y_j$ leads to a solution. Therefore, $\bar{S}$ is a $k$-parameter group (depending on $n$ and clearly on $k$).
			\end{itemize}
		\end{theorem}
		This provides a complete (implicit) characterization of the symmetries of all non-derogatory ES classes. In Appendix \ref{app:nonderogatorysyms} we present a systematic construction on how the linear equations determining $\bar{S}^{\otimes n}$ can be obtained and solved and we present an example. Note that Theorem \ref{lemma:eksyms} implies that all symmetries $S^{(1)}\otimes S^{(2)} \otimes \ldots \otimes S^{(n)}$ are upper triangular and normalizing the matrices to $[S^{(i)}]_{0,0} = 1$, the diagonal entries of the symmetries are given by $[S^{(i)}]_{l,l} = x^l$ for $l \in \{0,1,...,k\}$, $i \in \{0,1,...,n\}$ and for some $x \in  \mathbb{C}$. Note that this normalization implies a proportionality factor $x^k$, i.e., $S^{(1)}\otimes S^{(2)} \otimes \ldots \otimes S^{(n)}\ket{E_k}= x^k \ket{E_k}$.
		
		We will next provide a characterization of the symmetries in the general non-derogatory case.
		As we will see, we will get additional symmetries through acting over different blocks. Consider for example the (generalized) GHZ state, which has an arbitrary permutation matrix acting on all parties as a symmetry. In the following, however, we show that not many new symmetries appear in the non-derogatory case. In fact, we show that the only new symmetries come from permuting blocks of equal size and a scaling factor of the individual blocks. In particular, if no blocks have equal size, e.g., $\ket{E_3} \oplus \ket{E_2}$, no non-block-diagonal symmetries appear at all.
		
		\begin{theorem}
			\label{lemma:eksumsyms}
			For $n \geq 3$ the symmetries of states $\bigoplus_{b=1}^K \ket{E_{k_b}}$ are generated by
			\begin{enumerate}
				\item $\bigoplus_{b=1}^K S_{k_b}$, with $S_{k_b}$ a symmetry of the individual $\ket{E_{k_b}}$, i.e. $S_{k_b}\ket{E_{k_b}}=\alpha \ket{E_{k_b}}$ for any $b\in \{1,\ldots K\}$  
				\item diagonal matrices $D=\otimes_i D(\vec{\gamma}_i)$ with $D(\vec{\gamma}_i)= \bigoplus_{b=1}^K (\vec{\gamma}_i)_b \identity_{k_b}$ for some $\vec{\gamma}_i \in \mathbb{C}^K$ such that $\prod_i (\vec{\gamma}_i)_b=\gamma$ for all $b$  \cite{footnote9}.
				\item simultaneous permutations of whole blocks that have equal size for all parties (i.e. permutations of the form $X_\sigma^{\otimes n}$).
			\end{enumerate}
		\end{theorem}
		
		Note that the symmetries in the theorem above are normalized such that $S\bigoplus_{b=1}^K \ket{E_{k_b}}=\alpha\gamma\bigoplus_{b=1}^K \ket{E_{k_b}}$.

		\subsection{$LOCC_{\mathbb{N}}$ convertibility}
		\label{sec:es}
		
		In this section, we study state transformations within non-derogatory ES classes and show that indeed a rich LOCC structure becomes apparent. 
		In particular, we will show that within each non-derogatory ES class, transformations from symmetric states to symmetric states are possible. Nevertheless, we also show that apart from qubit systems, non-derogatory ES classes always contain weakly isolated states. Moreover, in certain non-derogatory ES classes, symmetric states that are weakly isolated exist. It thus turns out that the qubit W and qubit GHZ classes are the only non-derogatory ES SLOCC classes, which are completely free of isolation.

		Having characterized the stabilizer of  non-derogatory ES classes in Theorem \ref{lemma:eksyms} and \ref{lemma:eksumsyms}, we are now in the position to show that within each non-derogatory ES class, transformations from symmetric states to symmetric states are possible. Let us start out with a simple example of such a transformation. 
		To this end, we consider the $W$ state, i.e., $\ket{W}=1/\sqrt{n}\ket{E_1}$. The stabilizer of this state is given by \cite{Verstraete2002}
		\begin{align}
			\frac{1}{x}\bigotimes_{i=1}^{n-1}\begin{pmatrix}
				1& y_i    \\
				0 &  x
			\end{pmatrix} \otimes \begin{pmatrix}
				1& -\sum_{i=1}^{n-1} y_i    \\
				0 &  x
			\end{pmatrix}
		\end{align}
		with $x, y_i \in \mathbb{C}$ and $x\neq 0$.

		Every symmetric state in this SLOCC class can be written up to LUs as $\sqrt{p} \ket{0, \ldots, 0} + \sqrt{1-p} \ket{W}$ \cite{turgutW}. A transformation from this state to $\sqrt{p'} \ket{0, \ldots, 0} + \sqrt{1-p'} \ket{W}$ is possible by $LOCC_\mathbb{N}$ iff $p' \geq p$ \cite{turgutW}. In the LOCC protocol achieving this transformation, all parties act once consecutively. For simplicity, we recall here the protocol, where the initial state is $\ket{W}$. First, note that the final state $\sqrt{p'} \ket{0, \ldots, 0} + \sqrt{1-p'} \ket{W}$ can be written as
		\begin{align}
			\begin{pmatrix}
				1&0      \\
				0 &   \sqrt{1-p'}
			\end{pmatrix} \otimes \begin{pmatrix}
				1& 0       \\
				0 &   \sqrt{1-p'}
			\end{pmatrix} \otimes \ldots \otimes \begin{pmatrix}
				1&  \sqrt{n p'}    \\
				0 &  \sqrt{1-p'}
			\end{pmatrix} \ket{W}.
		\end{align}

		We will now show that the transformation can be performed in the following $n$-round protocol 
		\begin{align}
			\label{eq:trafos}
			\identity \otimes \identity \otimes \ldots \otimes \identity \ket{W} \nonumber\\
			\longrightarrow \underbrace{\begin{pmatrix}
					1& x_1       \\
					0 &   b
				\end{pmatrix}\otimes \identity \otimes \identity \otimes \ldots \otimes \identity}_{h(1)}   \ket{W} \nonumber\\
			= \underbrace{ \begin{pmatrix}
					1& 0       \\
					0 &   b
				\end{pmatrix} \otimes \begin{pmatrix}
					1& x_1       \\
					0 &   1
				\end{pmatrix} \otimes \identity \otimes \identity \otimes \ldots \otimes \identity}_{g(1)} \ket{W} \nonumber\\
			\longrightarrow \underbrace{ \begin{pmatrix}
					1& 0       \\
					0 &   b
				\end{pmatrix}\otimes \begin{pmatrix}
					1& x_2       \\
					0 &   b
				\end{pmatrix} \otimes \identity \otimes \identity \otimes \ldots \otimes \identity}_{h(2)} \ket{W}\nonumber \\
			\vdots \quad \nonumber \\
			\longrightarrow \underbrace{\begin{pmatrix}
					1& 0       \\
					0 &   b
				\end{pmatrix} \otimes \ldots \otimes \begin{pmatrix}
					1& 0       \\
					0 &   b
				\end{pmatrix}  \otimes \begin{pmatrix}
					1& x_n     \\
					0 &   b
			\end{pmatrix}}_{h(n)} \ket{W},
		\end{align}
		where $b=  \sqrt{1-p'}$, and $x_k = \sqrt{k p'}$. That the $k$th step of the protocol is possible by a local measurement by party $k$ can be seen as follows (see also \cite{turgutW}). Let $h(k) = \bigotimes_i h_i(k)$ denote the operators that give the state after the $k$th round through $h_1(k) \otimes \ldots \otimes h_n(k) \ket{W}$ and the operators describing the state before the next round as $g(k) = \bigotimes_i g_i(k)$, as in Eq. (\ref{eq:trafos}), and by convention $g_i(0) = \identity$. Note that $g(k)$ and $h(k)$ are actually two ways of representing the same state, i.e. $g(k) \ket{W} = h(k) \ket{W}$ due to the symmetries of $\ket{W}$. Let us moreover use the convention $G_i(k) = g_i(k)^\dagger g_i(k)$ and $H_i(k) = h_i(k)^\dagger h_i(k)$.
		
		Let us now prove that the protocol sketched in Eq. (\ref{eq:trafos}) can indeed be implemented via LOCC. To this end, let us consider round $k$ of the protocol.
		It is clear that it is possible to perform a measurement at site $k$ (followed by local unitaries by the remaining parties) in order to obtain $h_1(k) \otimes \ldots \otimes h_n(k) \ket{W}$ from $g_1(k-1) \otimes \ldots \otimes g_n(k-1) \ket{W}$ if it is possible to find symmetries $S_m$ such that
		\begin{align}
			[S_m^{(i)}, G_i(k)] = 0 \ \forall i \in \{1, \ldots, n\} \setminus \{k\}, \text{ and} \\
			\sum_m q_m (S_m^{(k)})^\dagger H_k(k) S_m^{(k)} = G_k(k-1)
		\end{align}
		for $q_m \geq 0$, $\sum_m q_m =1$. Note that all diagonal unitary symmetries fulfill $[S_m^{(i)}, G_i(k)] = 0$ for all $i \neq k$. Moreover, it is possible to satisfy Eq. (\theequation) by choosing $q_1=q_2=1/2$ and $\{S_m\}_m = \{\operatorname{diag}(1, e^{-i \phi}), \operatorname{diag}(1, e^{i \phi})\}$ with $\phi =  \arccos{\frac{x_{k-1}}{x_k}} =  \arccos{\sqrt{\frac{k-1}{k}}}$. This proves that the  transformation from $\ket{W}$ to $\sqrt{p'} \ket{0, \ldots, 0} + \sqrt{1-p'} \ket{W}$ is indeed possible via LOCC.
		Note that the corresponding measurement operators (for round $k$), which are given by
		\begin{align}
			\label{eq:wmeasops}
			M_m = h_k(k) S_m^{(k)} g_k(k-1)^{-1},
		\end{align}
		are upper triangular. Note further that in each round the state is transformed deterministically into another state in the $W$--class, i.e. the protocol is all--deterministic. The conditions for the existence of such a transformation are known explicitly (see \cite{turgutW, dVSp13,SpdV17,dVSp17}). 
		
		We use this example now to prove that in each non-derogatory ES SLOCC classes there exist pairs of symmetric states which can be converted into each other, as stated in the following theorem. 
		
		\begin{theorem}
			\label{thm:symtosym}
			In all non-derogatory ES SLOCC classes, there exist symmetric states $\ket{\psi}$ and $\ket{\phi}$ such that $\ket{\psi}$ can be converted into $\ket{\phi}$ via LOCC.
		\end{theorem}
		\begin{proof}
			
			Let us first prove the theorem for a single block, hence, for a $n$-partite representative $\ket{E_k}$. We have already presented an example for $k=1$. The proof for SLOCC classes represented by $\ket{E_k}$ for $k \geq 2$ is a simple generalization thereof, as seen as follows. 
			
			Consider the $2^n$-dimensional subspace spanned by computational basis states where each local state is either $\ket{0}$ or $\ket{k}$. Note that $\ket{E_k}$ projected onto that subspace gives a symmetric state which resembles the $\ket{W}$ state, however with $1$ replaced by $k$.
			
			Define the measurement operators $\{M_m\}_m$ as in  Eq. (\ref{eq:wmeasops}) where the operators  $h_k(k)$ and $g_k(k-1)$ act trivially on the $(k-1)$-dimensional subspace spanned by $\{\ket{1}, \ldots, \ket{k-1}\}$ and act as $h_k(k)$ and $g_k(k-1)$ respectively in the protocol explained above (see Eq. (\ref{eq:trafos})) on the subspace spanned by $\ket{0}$ and $\ket{k}$. The symmetries $S_m^{(k)}$ are diagonal with elements  $[S_m^{(k)}]_{jj}= \left(e^{i(-1)^m \phi/k}\right)^j$ where as before $\phi = \arccos{\frac{x_{k-1}}{x_k}} = \arccos{\sqrt{\frac{k-1}{k}}}$. Clearly, the constructed operators constitute a valid measurement. Considering the LOCC protocol described above with those measurement operators $\{M_m\}_m$ on the state $\ket{E_k}$ leads to states of the form
			\begin{align}\ldots \otimes \begin{pmatrix}
					1&  &  &       \\
					& \ddots &   &        \\
					&  & 1 &         \\
					&  &  & b        \\
				\end{pmatrix}  \otimes \begin{pmatrix}
					1&  &  &     x_n  \\
					& \ddots &   &        \\
					&  & 1 &         \\
					&  &  & b        \\
				\end{pmatrix} \ket{E_k}.
			\end{align}
			
			Hence, it is possible to deterministically transform the symmetric seed state $\ket{E_k}$ to those symmetric states. This proves the theorem for the case of SLOCC classes represented by single $\ket{E_k}$. The proof straight-forwardly generalizes to SLOCC classes represented by direct sums of $\ket{E_k^{(b)}}$ except for the case $k_b=0$ for all $b$. We will first present an example for $K=2$ and $k_b=0$, i.e., the qubit GHZ class. Our initial state of the transformation is the GHZ state, the final state is of the form  $g_x\otimes \ldots g_x\otimes [g_x Z(\pi/4)] \ket{GHZ}$, where $g_x \neq \identity$, $[g_x, \sigma_x]=0$ and $Z(\phi)=\operatorname{diag}(e^{i \phi}, e^{-i \phi})$. For convenience, we consider the normalization  $\tr(G_x)=1$. The following protocol (see also \cite{turgutGHZ}) allows to achieve the transformation. Each party measures once and the measurement is chosen to be the same for all parties. Its measurement operators are of the form 
			\bea \label{POVM2}
			M_1 =\frac{g_x Z(\pi/4)}{\sqrt{2}}, M_2 =\frac{g_x Z(\pi/4)\sigma_x}{\sqrt{2}}.
			\eea
			Wlog we assume that party $n$ is the last one to implement a non-trivial measurement. In the first $n-1$ rounds all parties except party $n$ (and the party implementing the measurement in this round) apply $\sigma_x$ if the second outcomes occurs. Party $n$ applies in this case $Z(-\pi/4)\sigma_x$ and in case of the first outcome $Z(-\pi/4)$. In the last round all parties apply $\sigma_x$ for the second outcome of the measurement of party $n$. Using that $[g_x, \sigma_x]=0$, that $\sigma_x^{\otimes n}$ is a symmetry and that further $g_x\otimes\ldots\otimes g_x\otimes g_xZ(\pi/4)\otimes \identity\ldots \identity\ket{GHZ}= g_x\otimes\ldots\otimes g_x\otimes g_x\otimes \identity\ldots \identity\otimes Z(\pi/4)\ket{GHZ}$ it is easy to see that this succeeds in implementing the transformation. This transformation can be straightforwardly generalized to higher-dimensional GHZ classes, as one can analogously to before embed it into a higher-dimensional space. 
		\end{proof}
		
		The results above show that e.g. $E(\ket{W})\geq E(\sqrt{p}\ket{0\ldots 0}+ \sqrt{1-p}\ket{W})$ for $p\geq 0$ for any entanglement measure $E$.

		Let us show now that also among non--derogatory ES classes weakly isolated states (even symmetric states) exist. Due to Lemma \ref{lemma:isolation} and the fact that we have determined important properties of all symmetries of non-derogatory ES classes, we are in the position to prove very general statements of the existence of weak isolation. However, due to the richness of the considered classes here and the fact that the statements hold for (almost all) system sizes, the proofs get rather technical. However, the main idea to prove the existence of weak isolation is as before: We identify explicitly operators $g_i$ such that the quasi--commutation relations presented in Lemma \ref{lemma:isolation} cannot hold for any symmetry of the state, which then implies that the corresponding state is weakly isolated. In this context, we will make use of the following lemma which can be straightforwardly proven (see Appendix \ref{app:commutationlemma}).
		
		\begin{lemma}
			\label{lemma:commutationlemma}
			Let $k \geq 1$. Given an $(k+1) \times (k+1)$ upper triangular matrix A for which $A_{l,l}$ = $x^l$ for some $x\in \mathbb{C}$, and  an $(k+1) \times (k+1)$ matrix
			\begin{align} \label{eq_G}
				g =\begin{pmatrix}
					1&   &   &     \\
					& \ddots &   &     \\
					&   &  1  & a   \\
					&   &   & 1   \\
				\end{pmatrix},
			\end{align}
			where $a \in \mathbb{C}$. Then it holds that $A^\dagger g^\dagger g A \propto g^\dagger g$ iff $|x|=1$ and A is of the form
			\begin{align}
				\label{eq:Aform}
				A \propto \begin{pmatrix}
					1&   &   & &     \\
					& x  &   &  &   \\
					&  & \ddots &   &     \\
					&  &   &  x^{k-1}  & a  \left( \frac{1}{{x^*}^{k-1}} - x^k  \right) \\
					&  &   &   & x^k   \\
				\end{pmatrix}.
			\end{align}
		\end{lemma}
		
		All matrix elements which are not displayed here are zeros. We use this lemma now to show that weak isolation exists in SLOCC classes represented by some $\ket{E_k}$ with $k \geq 2$. The proof of this statement implies then also that one can even find symmetric states which are weakly isolated in case $k \geq 3$. Considering then SLOCC classes which correspond to states of the form $\oplus_b \ket{E_{k_b}}$ for more than one $b$, gets more involved as additional symmetries exist in case some blocks have the same size (see Theorem {\ref{lemma:eksumsyms}). An example of such a class is the one corresponding to the generalized GHZ state, which might be expected to have no isolation due to the fact that there is no isolation for the qubit case. However, as we show below, this is not the case and in fact for all these classes, weak isolation exists. 
			
			Let us now start out by showing that weak isolation exists in SLOCC classes represented by some $\ket{E_k}$, 
			
			\begin{theorem}
				\label{theo:isolation}
				In any SLOCC classes represented by some $\ket{E_k}$ with $k \geq 2$, weakly isolated states exist. 
			\end{theorem}
			\begin{proof}
				We will proof the theorem by constructing a family of states $g_1 \otimes \ldots \otimes g_n \ket{E_k}$ with the following property. For all symmetries  $S^{(1)} \otimes S^{(2)} \otimes \ldots \otimes S^{(n)}$ of the seed state $\ket{E_k}$, it holds that either $(S^{(i)})^\dagger g_i^\dagger g_i S^{(i)} \propto g_i^\dagger g_i$ is true for no more than $n-2$ parties, or $S\propto\identity^{\otimes n}$. Due to Lemma \ref{lemma:isolation}, states with this property are neither reachable via $LOCC_{\mathbb{N}}$, nor convertible via $LOCC_1$.
				
				Consider to this end $g_i$ as in Lemma \ref{lemma:commutationlemma} with $a$ being replaced by $a_i$ and pairwise different $a_i \neq 0$.  Assuming that $S_i^\dagger g_i^\dagger g_i S_i \propto g_i^\dagger g_i$ holds for more than $n-2$ parties (let us assume wlog that it holds for the first $n-1$ parties), let us now show that $S \propto \identity$. Recall first that it follows from Theorem \ref{lemma:eksyms} that $S^{(l)}$ for $l\in \{1,...,n\}$ is upper triangular and  its  diagonal  elements  are up to some proportionality factor powers of $x$. Lemma \ref{lemma:commutationlemma} further implies that 
				\begin{align}
					\label{eq:symform}
					S^{(i)} \propto \begin{pmatrix}
						1&   &   & &     \\
						& x  &   &  &   \\
						&  & \ddots &   &     \\
						&  &   &  x^{k-1}  & a_i  \left( \frac{1}{{x^*}^{k-1}} - x^k  \right) \\
						&  &   &   & x^k   \\
					\end{pmatrix}
				\end{align} for $i \in \{1, \ldots, n-1\}$ and, moreover, $|x|=1$. We will now continue to restrict the form of $S^{(i)}$ by making  use  of  properties  of symmetries of $\ket{E_k}$. Whereas Theorem \ref{lemma:eksyms} provides an implicit characterization of the symmetry group -more precisely, of its generators- we are interested here in arbitrary symmetries and instead of considering products of the generators, we will rather focus on specific properties of the symmetries.  
				
				Considering the defining equation for symmetries (we introduce here a proportionality factor to avoid the need to take any normalization condition into account) we have that,
				\begin{align}
					S^{(1)} \otimes S^{(2)} \otimes \ldots \otimes S^{(n)} \ket{E_k}= \alpha \ket{E_k}
				\end{align}
				which is equivalent to
				\begin{align}
					\label{eq:symconditions}
					\sum_{i_1 + \ldots + i_n = k}  [S^{(1)}]_{m_1,i_1} \ldots  [S^{(n)}]_{m_n,i_n} = \alpha \ \delta_{m_1 + \ldots + m_n, k},
				\end{align}
				for all $m_1, \ldots, m_n$ and for some $\alpha \in \mathbb{C}$.
				
				By choosing $m = (m_1, \ldots, m_n) = (0, \ldots, 0, k-1)$, only one term in the sum of Eq. (\ref{eq:symconditions}) survives, namely, $[S^{(n)}]_{k-1,k} = 0$, due to the form of $S^{(1)}, \ldots, S^{(n-1)}$.
				
				Although the theorem holds for $k \geq 2$, let us for now continue the proof for the case $k \geq 3$ and discuss the remaining case at the end. Above, we have shown that in row $k-1$ of the matrix $S^{(n)}$, only the diagonal entry is non-vanishing. In case $k \geq 3$, we can extend this result to all rows of $S^{(n)}$. In other words, for all rows $s \in \{0, \ldots, k-1\}$ it holds that $[S^{(n)}]_{s,s+1} = \ldots = [S^{(n)}]_{s,k} =  0$. Let us now show that the statement holds for row $s$. To this end, consider for all $t \in \{s+1, \ldots, k\}$ some instances of Eq. (\ref{eq:symconditions}), where $m =  (m_1, \ldots, m_{n-1}, s)$ for some choices of $m_1, \ldots, m_{n-1}$ such that $m_1 + \ldots + m_{n-1} = k -t$ and $m_i \leq k-2$ for all $i \in \{1, \ldots, n-1\}$. In other words, we want to consider Eq. (\ref{eq:symconditions}) for cases where the first $n-1$ of the $m_i$ sum up to specific values, $k-t$, but $S^{(1)}, \ldots, S^{(n-1)}$ are only accessed in rows in which all but the diagonal entries are vanishing. This choice of $m$ leads to the fact that in the sum in Eq. (\ref{eq:symconditions}), only terms where $i_1 = m_1, i_2 = m_2, \ldots, i_{n-1} = m_{n-1}$ contribute. Hence we obtain
				\begin{align}
					x^{k-t} [S^{(n)}]_{s,t} = \alpha \delta_{k-t+s,k} = 0.
				\end{align}
				for all $t \in \{s+1, \ldots, k\}$, which proves that $S^{(n)}$ is diagonal. It immediately follows that all $S^{(i)}$ are diagonal with the same technique. Recall that in order to prove these statements, we were using the assumption $k \geq 3$. This assumption was needed, as in this case the desired choices of $m_i$ are possible and hence the considered instances of Eq.\;(\ref{eq:symconditions}) exist, which is not the case if $k=2$. Note that although the case $k=2$ has to be treated with additional care,it is still possible to show that all $S^{(i)}$ must be diagonal using Eq.\;(\ref{eq:symconditions}) and the additional assumption that $a_i$ are pairwise different, which is not needed in case $k\geq 3$.
				
				In any case, as $S_i$ must have the form in Eq. (\ref{eq:symform}), but at the same time be diagonal, it follows that $x=1$ (as $a_i \neq 0$) and thus $S \propto \identity$, which completes the proof of the theorem.
			\end{proof}

			As already noted in the proof of Theorem \ref{theo:isolation}, choosing the $a_i$ pairwise different in the construction of isolated states in the SLOCC class of $\ket{E_k}$ is not necessary in case $k \geq 3$. Thus, the following corollary follows by choosing $g_i = g$ for all $i$.
			\begin{corollary}
				In SLOCC classes represented by $\ket{E_k}$, where $k \geq 3$, there exist symmetric weakly isolated states. 
			\end{corollary}
			
			This adds to the example of the many-qubit symmetric states (non-exceptionally symmetric states) that are isolated (in that case not only weakly isolated). 
			
			As a first instance of states of the form $\bigoplus_i \ket{E_{k_i}}$ we consider now the  higher-dimensional GHZ classes symmetric states corresponding to $k_i=0$ for all $i$. We show that also there, weakly isolated states can be found, which is in contrast to the qubit case, where no isolation exists. 
			
			\begin{observation}\label{obs_GHZd}
				There is weak isolation in higher-dimensional GHZ classes. Moreover, there exist symmetric weakly isolated states.
			\end{observation}
			\begin{proof}
				Let us consider an $n$-partite, $d$-dimensional GHZ state, $\ket{GHZ_n^d} = \sum_{i=0}^{d-1} \ket{\underbrace{i \ldots i}_{n}}$, and  $G =  \identity + \sum_{i=1}^{d-1} (c_{i} \ket{i-1}\bra{i}+ h.c.)$ with  $c_{i} \neq 0$ small enough such that $G>0$ and with pairwise different absolute values, i.e. $|c_{i}|\neq |c_{k}|$ for $i\neq k$. We will see that states $g^{\otimes n}\ket{GHZ_n^d}$ with $g=\sqrt{G}$ are weakly isolated. To this end, we will, as before, show that  $(S^{(i)})^\dagger G S^{(i)} \propto G$ for $n-1$ parties implies $S\propto \identity$. Using $S=  \left( D(\vec{\gamma}_1) \otimes \ldots \otimes D(\vec{\gamma}_n) \right) \left( X_\sigma^{\otimes n} \right)$ (see Theorem \ref{lemma:eksumsyms} and \cite{multicopyLOCC}), we have that
				\begin{align}
					D(\vec{\gamma})^\dagger G  D(\vec{\gamma}) \propto X_\sigma G X_\sigma^\dagger.
				\end{align}
				Considering the diagonal entries of this equation it is straightforward to see that up to a (irrelevant) global proportionality factor the entries in $\vec{\gamma}$ are solely phases. Hence, we either have that $c_{j} e^{i(\phi_{j-1}-\phi_j)}=0$ or $c_{j} e^{i(\phi_{j-1}-\phi_j)}=c_{k}$. Which of the two cases one observes and the relation among $j$ and $k$ is determined by the permutation $X_\sigma$. As the absolute values of $c_{j}$ and $c_{k}$ are non-zero and differ for $j\neq k$, we have that the permutation has to be trivial. From this and $c_{j} \neq 0$, we obtain that $e^{i(\phi_{j-1}-\phi_j)}=1$ and hence $S^{(i)}\propto \identity$ for all parties $i$, which proves that there is weak isolation in the GHZ classes beyond qubits. As the examples that we discussed here are symmetric this proves the observation.

			\end{proof}

			In contrast to the positive result on LOCC transformations from symmetric states to symmetric states, the isolation result in Theorem \ref{theo:isolation} does not immediately generalize to (non-derogatory) direct sums of $\ket{E_k}$. 
			However, based on the studies of the symmetries of $\bigoplus_i \ket{E_{k_i}}$ in Theorem \ref{lemma:eksumsyms}, we will see in the following theorem, that it is nevertheless possible to construct weakly isolated states.

			\begin{theorem}
				\label{theo:sumsisolation}
				In the SLOCC classes represented by $\bigoplus_{b=1}^K \ket{E_{k_b}}$ with $K\geq 2$ and $k_b\neq 0$ for at least one $b\in\{1,\ldots, K\}$, there are weakly isolated states present. If further $k_b\neq 2$ for all $b\in\{1,\ldots, K\}$, then there exist symmetric weakly isolated states.
			\end{theorem}

			Whereas the idea of the proof is the same as in the proof of Theorem \ref{theo:isolation}, it is much more involved due to the additional symmetries occurring in this case. For that reason we postpone to proof to Appendix \ref{App:proof_iso_NDES}.

			\section{Derogatory exceptionally symmetric classes\label{sec:derog}}
			Most things said about non-derogatory ES SLOCC classes in the previous section are not true any more for derogatory ES classes. In this subsection, we will first argue that a full characterization of the symmetries of all derogatory ES classes seems infeasible. Then, we will show that some of the derogatory ES classes resemble the multi-copy scenario and inherit some of the properties of non-derogatory ES states. Despite the fact that those states necessarily have a non--trivial stabilizer, we show that there exist derogatory ES classes in which all states are isolated with respect to $LOCC_{\mathbb{N}}$. To illustrate the properties of derogatory ES classes we finally study the three- and four-qutrit cases in detail.

			As mentioned in Section \ref{sec:preliminaries}, the first difficulty that arises when studying derogatory ES classes, is the fact that representatives of the SLOCC classes are not available. As explained there, superpositions of the states $\{\ket{E_k^{n_1, \ldots, n_K}}\}_{k, (n_1 \ldots, n_K)}$ must be taken into account. Moreover, not all superpositions correspond to derogatory classes as some turn out to be equivalent to non-derogatory ES states, instead. For these reasons, a full characterization of the symmetries of derogatory ES classes is beyond the scope of the present work.
			
			Let us remark here, that some derogatory ES classes have representatives that decompose into two copies of (possibly different) entangled states in lower dimensions. Examples are the classes represented by $\{\ket{E_k^{n_1, n_2}}\}_{j, (n_1, n_2)}$ for $B = J_2 \oplus J_2$ presented in Section \ref{sec:preliminaries} (see Eqs.\;(\ref{eq:derogEG1})). This happens in general, when considering an individual $\ket{E_k^{n_1, \ldots, n_K}}$. 
			To see this, note that without loss of generality we can consider all blocks to have the same size ($k_b+1=k$ $\forall b$ \cite{footnote10}), i.e.\ $B=\oplus_{i=1}^K J_k$. This is because either the block is big enough to support the excitation $k-1$ or the seed state is not supported there. In the same way, if $k_b+1>k$ the seed state is not supported on the corresponding levels (thus, it suffices to consider $k_b+1=k$). The aforementioned property follows then as simple consequence of the isometry $\oplus_{i=1}^K\mathbb{C}^k\simeq\mathbb{C}^K\otimes\mathbb{C}^k$ with $|i^{(b)}\rangle\to|b-1\rangle\otimes|i\rangle$. Hence, a seed state can be written as $\ket{E_k^{\{n_b\}}} =|\Psi_{\{n_b\}}\rangle\otimes|E_k\rangle$, where $|\Psi_{\{n_b\}}\rangle=\sum_{\vec{b}:\#b = n_b}\ket{b_1-1\cdots b_n-1}\in\operatorname{Sym}^n(\mathbb{C}^K)$ and where we use the same notation as in Eq.(\ref{eq:ekderog}). Note that these derogatory classes hence inherit some of the properties of non-derogatory ES classes discussed earlier in this section.
			
			However, it seems that derogatory ES classes allow for fewer transformations compared to their non-derogatory counterparts. The reason for that is that the symmetries of the form $B\otimes B^{-1}$ of non-derogatory ES classes allow for more freedom, as their eigenvalues can be freely chosen. In fact, we identify here a 5-qutrit derogatory ES class in which all states are $LOCC_{\mathbb{N}}$-isolated, i.e., they can neither be reached nor converted to another pure state with a finite-round LOCC protocol. Note that this is the first example of such a $LOCC_{\mathbb{N}}$-isolated class with non-trivial stabilizer. 
			
			\begin{observation}
				\label{obs:derogatoryisolation}
				There exist derogatory classes in which all states are $LOCC_{\mathbb{N}}$-isolated.
			\end{observation}
			\begin{proof}
				In order to prove the observation, we explicitly consider one derogatory SLOCC class and show that all states within the class are isolated. Let us consider the SLOCC class represented by the five-qutrit state
				\begin{align}
					\ket{\psi_{\text{derogatroy}}} = \ket{E_0^{5,0}} + \ket{E_0^{3,2}} + \ket{E_0^{2,3}}  + \ket{E_1^{0,5}}. 
				\end{align}
				Note that this state belongs to a derogatory class described by one block of size one and a second block of size two. It can be easily verified that indeed, $B_{(i)} \otimes B^{-1}_{(j)}$, are symmetries of $\ket{\psi_{\text{derogatroy}}}$, for any parties $i,j \in \{1, \ldots, 5\}$, where $B = \begin{pmatrix} 1 & 0 & 0\\ 0 & 1 & 1 \\ 0 & 0 & 1 \end{pmatrix}$.
				
				In Appendix \ref{sec:symmetries} we show that, in fact,  $B_{(i)} \otimes B^{-1}_{(j)}$ and analytic functions thereof constitute the full symmetry group of $\ket{\psi_{\text{derogatroy}}}$. Hence, all symmetries $S^{(1)} \otimes \ldots \otimes S^{(5)}$ differing from $\identity$ have the property that at least two of the $S^{(i)}$ are  not diagonalizable. Corollary \ref{lemma:quasicommutation} shows that with such symmetries it is not possible to satisfy condition (i) of Theorem \ref{thm:reachability}. Hence, no state in the considered SLOCC class is $LOCC_\mathbb{N}$-reachable and therefore, no state is convertible either. Thus, all states in this class are isolated (not only weakly isolated).
			\end{proof}
			
			\subsection{3 and 4 qutrits}
			In this subsection, we consider the smallest local dimension ($d=3$) in which non-trivial derogatory matrices emerge and we focus on 3- and 4-qutrit derogatory ES SLOCC classes. In 3 dimensions, the Jordan normal form of derogatory matrices can only take 2 distinct forms, namely
			\begin{equation}
			B_1 = \begin{pmatrix} x & 0 & 0\\ 0 & x & 0 \\ 0 & 0 & y \end{pmatrix} \text{ and } B_2 = \begin{pmatrix} x & 0 & 1\\ 0 & x & 0 \\ 0 & 0 & x \end{pmatrix} \label{eq:B_type12}
			\end{equation}
			with $x,y\in\mathbb{C}$, which we call ``type 1" and ``type 2". Notice that we swap the computational basis vectors $|1\rangle$ and $|2\rangle$ when writing $B_2$ as compared to the normal form $\begin{pmatrix} x & 1 \\ 0 & x \end{pmatrix}\oplus x$. This is to express the state $|\Psi_2\rangle$ in Eq.\;(\ref{eq:psi2derog}) as an $n$-qubit symmetric state in the subspace $\text{span}\{\ket{0},\ket{1}\}$ plus an orthogonal $n$-qutrit part. In general, the matrix $B$ of the $B\otimes B^{-1}\otimes\identity^{\otimes2}$ symmetry need not be in Jordan normal form. Recall from Sec.\;\ref{subsec:Prelim_ESandNES} that the Jordan normal form of B is SLOCC-invariant and one can always find a state within the same SLOCC class for which B is in Jordan normal form \cite{MiRo13}. Hence, in order to determine the representatives of derogatory ES classes, it suffices to consider the states stabilized by the symmetry with $B_1$ or $B_2$ and identify those that do not have additional non-derogatory symmetries.
			
			Any potential $n$-qutrit derogatory ES state $|\Psi_i\rangle$ of type $i=1,2$ satisfies $B_i \otimes {B_i}^{-1} \otimes \identity^{\otimes n-2}|\Psi_i\rangle=|\Psi_i\rangle$ where the use of an equality instead of a proportionality is explained in Sec.\;\ref{subsec:Prelim_ESandNES}. Since an ES state stabilized by $B\otimes B^{-1}$ is also stabilized by $f(B)\otimes f(B)^{-1}$ for any analytic function $f$ (see Sec.\;\ref{subsec:Prelim_ESandNES}), non-derogatory states can also be stabilized by the derogatory $B_i \otimes {B_i}^{-1}$ in Eq.\;(\ref{eq:B_type12}). Therefore, the potential derogatory ES states also include non-derogatory ES states.
			
			As mentioned in Sec.\;\ref{subsec:Prelim_DandND}, any derogatory ES state has to be of the form in Eq.\;(\ref{eq:sumek2}). For the ease of presentation in this section, we will use instead of $\ket{E_k^{n_1, \ldots, n_K}}$ in Eq.\;(\ref{eq:ekderog}) 
			\begin{flalign}
				\ket{S^n_k} &= \frac{1}{\sqrt{n!}}\sum_{\pi} P_{\pi}(\ket{0^{n-k}1^k}) = \frac{1}{\sqrt{C^n_k}}\ket{E^{n-k,k}_0}, \\
				\ket{F^n_k} &= \frac{1}{\sqrt{n!}}\sum_{\pi} P_{\pi}(\ket{0^{n-k}2^k}) = \frac{1}{\sqrt{C^n_k}} \ket{E^{n,0}_k}
			\end{flalign}

			as the orthonormal basis, where the summations are over all permutations in the symmetric group $\text{S}_n$ and $C^n_k=\frac{n!}{(n-k)!k!}$ is the binomial coefficient. The complete sets of potential type-1 and type-2 $n$-qutrit derogatory ES states can be obtained from Eqs.\;(\ref{eq:ek})--(\ref{eq:ekderogsuperpos}). Since the matrix $B_1$ has a degeneracy only in the first two 1-dimensional Jordan blocks, $\ket{\Psi_1}$ can be any superposition of states with 0 excitation distributed over the two blocks, $\ket{E^{n-k,k}_0}$ ($k=0,\ldots,n$), plus a zero-excitation state $\ket{E_0}$ in the third dimension. For $B_2$, the degeneracy is between the 2-dimensional Jordan block in the subspace $\text{span}\{\ket{0},\ket{2}\}$ and the middle 1-dimensional block. Hence, $\ket{\Psi_2}$ can be any superposition of states with 0 excitation distributed over the two blocks plus a 1-excitation state, $|E^{n,0}_1\rangle$, that has support only from the 2-dimensional block in $B_2$. In the new notation, these two sets of states are 
			\begin{flalign}
				|\Psi_1\rangle &= \sum_{k=0}^{n} \frac{a_k}{\sqrt{C^n_k}} |E^{n-k,k}_0\rangle \oplus a_{n+1}|E_0\rangle \nonumber\\
				&= a_0|0^n\rangle + \sum_{k=1}^{n-1} a_k|S^n_k\rangle + a_n|1^n\rangle + a_{n+1}|2^n\rangle, \label{eq:psi1derog}\\
				|\Psi_2\rangle &= \sum_{k=0}^{n} \frac{b_k}{\sqrt{C^n_k}}|E^{n-k,k}_0\rangle + \frac{b_{n+1}}{\sqrt{n}}|E^{n,0}_1\rangle \nonumber\\
				&= b_0|0^n\rangle + \sum_{k=1}^{n-1} b_k|S^n_k\rangle + b_n|1^n\rangle +b_{n+1}|F^n_1\rangle, \label{eq:psi2derog}
			\end{flalign}
			where $a_k,b_k\in\mathbb{C}$ satsify $\sum_{k=0}^{n+1}|a_k|^2=1=\sum_{k=0}^{n+1}|b_k|^2$. 
			
			\subsubsection{3 qutrits}
			We will show that there exist no derogatory ES states in 3 qutrits. First, we define an $n$-qubit symmetric state (embedded in the Hilbert space of $n$ qutrits) to be a state with local rank $\leq 2$. We then consider all potential 3-qutrit derogatory ES states and determine all the corresponding SLOCC classes. We show that every symmetric SLOCC representative stabilized by $B_i \otimes {B_i}^{-1} \otimes \identity$ ($i=1,2$) is either a 3-qubit state or coinciding with the 3-qutrit non-derogatory ES SLOCC representatives $\ket{E_0}\oplus\ket{E_0}\oplus\ket{E_0}=\ket{0^3}+\ket{1^3}+\ket{2^3}$, $\frac{1}{\sqrt{3}}\ket{E_1}\oplus\ket{E_0}=\ket{S^3_1}+\ket{2^3}$, and $\ket{E_2}\propto\ket{S_2^3}+\ket{F^3_1}$ (see Sec.\;\ref{subsec:Prelim_ESandNES}).
			
			\begin{observation}
				There are no derogatory ES states in 3 qutrits since any potential 3-qutrit derogatory ES state belongs to one of the 5 SLOCC classes represented by
				\begin{align*}
					&(a) \;\ket{0^3}+\ket{2^3}, (b)\; \ket{0^3}+\ket{1^3}+\ket{2^3}, (c)\; \ket{S^3_1}+\ket{2^3},\\ 
					&(d)\; \ket{F^3_1}, \;\text{and }(e)\; \ket{S_2^3}+\ket{F^3_1},
				\end{align*}
				which are all non-derogatory. \label{Obs:3Qutrit}
			\end{observation}
			
			\begin{proof}
				The states corresponding to $a_4=0=b_4$ in Eqs.\;(\ref{eq:psi1derog}) and (\ref{eq:psi2derog}) are 3-qubit states which are non-derogatory \cite{footnotefullrank}. We will show that all 3-qutrit ES states in Eqs.\;(\ref{eq:psi1derog}) and (\ref{eq:psi2derog}) with $a_4,b_4\neq0$ can be reached from 3-qubit $\ket{\text{GHZ}_2}\propto\ket{0^3}+\ket{2^3}$, $\ket{\text{W}}=\ket{F^3_1}$, and the representatives of three 3-qutrit non-derogatory ES SLOCC classes using SLOCC transformation $A^{\otimes 3}$.
				
				As shown in Eqs.\;(\ref{eq:psi1derog}) and (\ref{eq:psi2derog}), potential type-1 and type-2 derogatory ES states are given by the superposition of any 3-qubit symmetric states and an orthogonal state $\ket{2^3}$ and $\ket{F^3_1}$, respectively. It was proven in Ref.\;\cite{slocc} that any 3-qubit symmetric states falls into one of the three SLOCC classes represented by the states $\ket{0^3}, \ket{0^3}+\ket{1^3}$, and $\ket{S^3_1}$. The representatives of all potential type-1 derogatory ES SLOCC classes are then simply superpositions of the three 3-qubit representatives and the state $\ket{2^3}$ as SLOCC transformations $A^{\otimes 3}$ with block-diagonal matrices $A=\begin{pmatrix} a & c\\ b & d \end{pmatrix}\oplus e \in GL(3,\mathbb{C})$ are sufficient to reach any state $\ket{\Psi_1}$ in Eq.\;(\ref{eq:psi1derog}) from
				\begin{align*}
					&\text{(a)} \;\ket{0^3}+\ket{2^3}, \;\text{(b)}\; \ket{0^3}+\ket{1^3}+\ket{2^3}, \;\text{or (c)}\; \ket{S^3_1}+\ket{2^3}.
				\end{align*}
				The state (a) is the 3-qubit GHZ state \cite{footnote11}, whereas states (b) and (c) are stabilized by $B\otimes B^{-1}\otimes\identity$ with non-derogatory matrices $B=x\oplus y\oplus z$ and $B=\begin{pmatrix} x & 0 & 1\\ 0 & y & 0 \\ 0 & 0 & x \end{pmatrix}$, respectively.
				
				We now show that any potential type-2 derogatory ES state belongs to one of the three following SLOCC classes represented by 
				\begin{align*}
					&\text{(d)} \;\ket{F^3_1}, \;\text{(e)}\; \ket{S_2^3}+\ket{F^3_1}, \;\text{and (f)}\; \ket{1^3}+\ket{F^3_1}
				\end{align*}
				where (d) is the 3-qubit W state, (e) is stabilized by $B\otimes B^{-1}\otimes\identity$ with non-derogatory matrix $B=\begin{pmatrix} x & 1 & 0\\ 0 & x & 1 \\ 0 & 0 & x \end{pmatrix}$, and (f) is SLOCC equivalent to the non-derogatory representative (c). It can be easily seen that these are the complete list of type-2 ES SLOCC representatives from explicit constructions of SLOCC transformations $A^{\otimes 3}$ with
				\begin{align*}
					&\text{(d)}\; A=\begin{pmatrix} 1 & 0 & \frac{b_0}{\sqrt{3}}\\ 0 & 1 & b_1 \\ 0 & 0 & b_4 \end{pmatrix}, \;\text{(e)}\; A=\begin{pmatrix} 1 & 0 & \frac{b_0}{\sqrt{3}}\\ 0 & \sqrt{b_2} & b_1 \\ 0 & 0 & b_4 \end{pmatrix},\\
					&\text{(f)}\; A=\begin{pmatrix} 1 & \frac{b_2}{\sqrt{3} b^{2/3}_3} & \frac{b_0}{\sqrt{3}}-\frac{b_2^3}{9 b^2_3}\\ 0 & b^{1/3}_3 & b_1-\frac{b_2^2}{\sqrt{3} b_3} \\ 0 & 0 & b_4 \end{pmatrix},
				\end{align*}
				which transform the state (d) to any 3-qutrit state $\ket{\Psi_2}$ in Eq.\;(\ref{eq:psi2derog}) with $b_2=b_3=0$, the state (e) to the ones with $b_2,b_4\neq0$ and $b_3=0$, and the state (f) to those with $b_3,b_4\neq0$. Recall that $b_4\neq0$ ensures the state $\ket{\Psi_2}$ to have local rank 3. This shows in particular that all 3-qutrit states are non-derogatory.
			\end{proof}

			We already know all the representatives of $n$-qutrit non-derogatory ES SLOCC classes from Sec.\;\ref{subsec:Prelim_ESandNES} to be
			\begin{align}
				&\ket{E_0}\oplus\ket{E_0}\oplus\ket{E_0}=\ket{0^n}+\ket{1^n}+\ket{2^n},\label{eq:nQutritNonderog1}\\
				&\ket{E_1}\oplus\ket{E_0}=\sqrt{n}\ket{S^n_1}+\ket{2^n}, \;\text{and}\label{eq:nQutritNonderog2}\\ &\ket{E_2}=\sqrt{C^n_2}\ket{S_2^n}+\sqrt{n}\ket{F^n_1}, \label{eq:nQutritNonderog3}
			\end{align}
			where the coefficients in Eqs.\;(\ref{eq:nQutritNonderog2}) and (\ref{eq:nQutritNonderog3}) can be set to 1 using SLOCC transformations. If we consider potential candidates for derogatory ES SLOCC representatives, they would be, for instance, those states of the form $\ket{S^n_{k>1}}+\ket{2^n}$ that are SLOCC inequivalent to $\ket{S^n_1}+\ket{2^n}$. However, for 3 qutrits, the state $\ket{S^3_2}+\ket{2^3}$ is SLOCC equivalent to $\ket{S^3_1}+\ket{2^3}$ and the state $\ket{1^3}+\ket{2^3}$ is just the 3-qubit GHZ state. Thus, these states are not derogatory. The situation changes in 4 qutrits because there exist states such as $\ket{S^4_2}+\ket{2^4}$ that are not SLOCC equivalent to any of the three non-derogatory representatives as we will see in the next subsection. This guarantees the existence of derogatory ES representatives for 4 qutrits.
			
			\subsubsection{4 qutrits\label{subsec:4qutritDerog}}
			We move on to study 4-qutrit ES states that are stabilized by $B_i \otimes {B_i}^{-1} \otimes \identity^{\otimes2}$ ($i=1,2$) since our goal is to investigate the properties of derogatory ES states which, as explained in the previous subsection, is guaranteed to exist in 4 qutrits. We first list the representatives for all 4-qutrit derogatory ES SLOCC classes. Then, we characterize all the local invertible symmetries for each of these representatives. Finally, we show that every 4-qutrit derogatory ES SLOCC class contains weakly isolated states, $\text{LOCC}_{\mathbb{N}}$-reachable states, and also $\text{LOCC}_1$-convertible states.
			\\
			
			\paragraph{\textbf{Representatives of derogatory ES SLOCC classes}}
			
			To obtain representatives of all 4-qutrit derogatory ES SLOCC classes, we start by   identifying the representatives of SLOCC classes with symmetric states stabilized by $B_i \otimes {B_i}^{-1} \otimes \identity^{\otimes2}$ ($i=1,2$) and then we eliminate the ones that are SLOCC equivalent to the three non-derogatory representatives in Eqs.\;(\ref{eq:nQutritNonderog1})--(\ref{eq:nQutritNonderog3}).
			
			Consider now the symmetric states in Eqs.\;(\ref{eq:psi1derog}) and (\ref{eq:psi2derog}) and set $n=4$. Note that we need not consider the states with $a_5=0=b_5$ as they are 4-qubit states and therefore must be non-derogatory \cite{footnotefullrank}. The first 4 terms of both states describe any 4-qubit symmetric state. We can use the result from Ref.\;\cite{BaKr09} that there are only 5 families of 4-qubit symmetric SLOCC classes represented by
			\begin{align}
				&\ket{1^4},\; \ket{S^4_1},\;\ket{S^4_2},\;\ket{0^4}+\ket{S^4_2}, \nonumber\\ &\{\ket{0^4}+\ket{1^4}+\mu\ket{S^4_2}:\mu\in\mathbb{C} \text{ and } \mu\neq\pm\sqrt{\frac{2}{3}}\}.
			\end{align}
			We define the set encompassing all these 4-qubit states to be $\mathcal{S}^4_{\text{qubit}}$. The last family contains infinitely many SLOCC classes for each class corresponding to a finite number of $\mu\in\mathbb{C}$. For $\ket{\Psi_1}$ in Eq.\;(\ref{eq:psi1derog}), superpositions of each 4-qubit representative and the orthogonal state $\ket{2^4}$ give an overcomplete list of SLOCC representatives $\{\ket{\psi}+\ket{2^4}:\ket{\psi}\in\mathcal{S}^4_{\text{qubit}}\}$. This list is overcomplete because
			\begin{align}
				\ket{\psi(\mu)}&\coloneqq \ket{0^4}+\mu\ket{S^4_2}+\ket{1^4}+\ket{2^4}\label{eq:PsiMu}\\
				&=(\ket{E^{4,0}_0}+\frac{\mu}{\sqrt{6}}\ket{E^{2,2}_0}+\ket{E^{0,4}_0})\oplus\ket{E_0} 
			\end{align}
			is SLOCC equivalent to, for instance, $\ket{\psi(-\mu)}$. Among the list, we recognize that $\ket{1^4}+\ket{2^4}$ is just a 4-qubit $\ket{\text{GHZ}_2}$, and the states $\ket{S^4_1}+\ket{2^4}$ and $\ket{\psi(\mu=0)}$ are non-derogatory (see Eqs.\;(\ref{eq:nQutritNonderog1}) and (\ref{eq:nQutritNonderog2})). We should also point out that (i) $\ket{\psi(\mu=\pm\sqrt{6})}$ is SLOCC equivalent to $\ket{\psi(\mu=0)}$ and (ii) $\ket{\psi(\mu=\pm\sqrt{\frac{2}{3}})}$ is SLOCC equivalent to $\ket{S^4_2}+\ket{2^4}$. The remaining representatives $\ket{S^4_2}+\ket{2^4}$, $\ket{0^4}+\ket{S^4_2}+\ket{2^4}$ and $\{\ket{\psi(\mu)}:\mu\in\mathbb{C} \text{ and } \mu\neq0, \pm\sqrt{\frac{2}{3}}, \pm\sqrt{6}\}$ are indeed derogatory which we will prove in Sec.\;\ref{par:4qutritSymm} when we characterize all local invertible symmetries of these representatives.

			For $\ket{\Psi_2}$ in Eq.\;(\ref{eq:psi2derog}), we have an overcomplete list of SLOCC representatives with $\{\ket{\psi}+\ket{F^4_1}:\ket{\psi}\in\mathcal{S}^4_{\text{qubit}}\}$. In contrast to the type-1 derogatory case, the number of potential type-2 derogatory ES SLOCC classes is finite and can be represented by 5 states
			\begin{flalign}
				&\ket{F^4_1},\;\ket{S^4_2}+\ket{F^4_1},\;\ket{1^4}+\ket{F^4_1},\nonumber\\
				&\ket{1^4}+\ket{S^4_2}+\ket{F^4_1}\;\text{and }\ket{S^4_3}+\ket{F^4_1}. \label{eq:4qutritType2Reps}
			\end{flalign}
			By explicit construction of SLOCC transformations $A^{\otimes4}$, we show that all states $\ket{\Psi_2}$ in Eq.\;(\ref{eq:psi2derog}) with $b_5\neq0$ can be reached from these 5 representatives. The constructions of $A^{\otimes4}$ are shown in Appendix \ref{subsec:4qutritsAx4}. Again, we identify the representatives that are non-derogatory: $\ket{S^4_2}+\ket{F^4_1}$ (SLOCC equivalent to $\sqrt{6}\ket{S^4_2}+2\ket{F^4_1}=\ket{E_2}$) and $\ket{1^4}+\ket{F^4_1}=\ket{E_0}\oplus\frac{1}{2}\ket{E_1}$ which are stabilized by $B\otimes B^{-1}\otimes\identity^{\otimes2}$ with non-derogatory $B=\begin{pmatrix} x & 1 & 0 \\ 0 & x & 1 \\ 0 & 0 & x \end{pmatrix}$ and $B=\begin{pmatrix} x & 0 & 1 \\ 0 & y & 0 \\ 0 & 0 & x \end{pmatrix}$, respectively. Restricting to 4-qutrit states with full local rank, the remaining two representatives $\ket{1^4}+\ket{S^4_2}+\ket{F^4_1}=\ket{E^{0,4}_0}+\frac{1}{\sqrt{6}}\ket{E^{2,2}_0}+\frac{1}{2}\ket{E^{4,0}_1}$ and $\ket{S^4_3}+\ket{F^4_1}=\frac{1}{2}(\ket{E^{1,3}_0}+\ket{E^{4,0}_1})$ are derogatory which will be verified in Sec.\;\ref{par:4qutritSymm}. Thus, we obtain the following representatives of all 4-qutrit derogatory ES SLOCC classes:
			\begin{align}
				&\ket{S^4_2}+\ket{2^4},\;\ket{0^4}+\ket{S^4_2}+\ket{2^4},\;\ket{1^4}+\ket{S^4_2}+\ket{F^4_1},\label{eq:4qutritDerogRep}\\
				&\ket{S^4_3}+\ket{F^4_1},\;\{\ket{\psi(\mu)}:\mu\in\mathbb{C} \text{ and } \mu\neq0, \pm\sqrt{\frac{2}{3}}, \pm\sqrt{6}\}.\nonumber
			\end{align}
			As already pointed out, different values of $\mu$ may lead to states that are SLOCC equivalent. However, all other states that are not of the form $\ket{\psi(\mu)}$ are not within the SLOCC class of another state in Eq.\;(\ref{eq:4qutritDerogRep}). In order to see that, recall that the Jordan normal form is SLOCC invariant. As we will show in the next section for all symmetries $B_{(i)}\otimes B^{-1}_{(j)}$, it holds that $B=B_1$ ($B$ generated by $B_2$) for type-1 (type-2) representatives respectively. Hence, type-1 and type-2 representatives cannot be within the same SLOCC class. 
			The cardinality of the group that consist of all $A$ for which $A^{\otimes4}\ket{\psi}=\ket{\psi}$  has to be the same for any symmetric state in the SLOCC class of a symmetric state $\ket{\psi}$. Using this, as well as the results on the symmetries $A^{\otimes4}$ given in the next subsection it is easy to see that the two type-2 representatives, as well as $\ket{S^4_2}+\ket{2^4}$ and $\ket{0^4}+\ket{S^4_2}+\ket{2^4}$ are SLOCC inequivalent to all other states in Eq.\;(\ref{eq:4qutritDerogRep}).
			\paragraph{\textbf{Characterization of all local symmetries}}\label{par:4qutritSymm}
			
			For each derogatory representative $\ket{\psi}$, we find all invertible symmetries $\bigotimes^4_{j=1} S^{(j)}$ (i.e., $\bigotimes^4_{j=1} S^{(j)}\ket{\psi}=\ket{\psi}$). We use the fact that the operators $S^{(j)}$ are invertible to obtain
			\begin{align}
				\prescript{}{i,j}{\bra{2,k}}S^{(i)}\otimes S^{(j)}\ket{\Psi_1}=\prescript{}{i,j}{\bra{2,k}}\Psi_1\rangle=0 , \\
				\prescript{}{i,j}{\bra{2,k'}}S^{(i)}\otimes S^{(j)}\ket{\Psi_2}=\prescript{}{i,j}{\bra{2,k'}}\Psi_2\rangle=0,
			\end{align}
			for $k=0,1$, $k'=1,2$ and for any states $\ket{\Psi_1}$ and $\ket{\Psi_2}$ in Eqs.\;(\ref{eq:psi1derog}) and (\ref{eq:psi2derog}). From these equations, we show that the local symmetries of the type-1 and type-2 derogatory representatives must be of the forms 
			\begin{align}
				S^{(j)}=\begin{pmatrix} a & d & 0 \\ b & e & 0 \\ 0 & 0 & p \end{pmatrix} \text{(for type 1)},\; S^{(j)}=\begin{pmatrix} a & d & g \\ 0 & e & h \\ 0 & 0 & p \end{pmatrix} \text{(for type 2)} \label{eq:symmPsi12}
			\end{align}
			for all $j=1,2,3,4$ (see Appendix \ref{subsec:structureOfS_j} for the proof). Note that the local symmetries of each derogatory representative $\ket{\psi}$ form a group $S^{(j)}_{\psi}$ for any site $j$. As mentioned in Sec.\;\ref{subsec:Prelim_ESandNES}, any symmetry $\bigotimes^4_{j=1} S^{(j)}$ can be generated by the symmetries of the two forms $A^{\otimes4}$ and $B_{(i)}\otimes B_{(j)}^{-1}$ for any ES state \cite{MiRo13}. Due to Eq.\;(\ref{eq:symmPsi12}), we can restrict the invertible matrices $A$ and $B$ to block-diagonal (upper triangular) form for type 1 (2) derogatory symmetries.

			As discussed in Sec.\;\ref{subsec:Prelim_ESandNES}, we can use equality instead of proportionality in solving for the symmetries. We identify all symmetries of the form $A^{\otimes4}$ of each derogatory representative $\ket{\psi}$ in Eq.\;(\ref{eq:4qutritDerogRep}) by solving the degree-2 polynomial equations
			\begin{align}
				(A^{\otimes2}\otimes\identity^{\otimes2}- \identity^{\otimes2}\otimes \widetilde{A}^{\otimes2})\ket{\psi}=0
			\end{align}
			and imposing that $\widetilde{A}=A^{-1}$ once we fix enough entries in matrices $A$ and $\widetilde{A}$.
			
			The full set of $B\otimes B^{-1}\otimes\identity^{\otimes2}$ symmetries of each $\ket{\psi}$ can be found easily by solving the equation
			\begin{align}
				(B\otimes\identity^{\otimes3}- \identity\otimes B\otimes\identity^{\otimes2})\ket{\psi}=0
			\end{align}
			for arbitrary invertible matrix $B$. We find that all these symmetries of type-1 and -2 derogatory representatives can be generated entirely by applying analytic functions $f$ on the matrices $B_1$ and $B_2$ in Eq.\;(\ref{eq:B_type12}), respectively. This is not necessarily true in general.

			As mentioned before, the general expression of the local symmetry $S^{(j)}$ for any site $j$ of each derogatory representative can be deduced from arbitrary products of matrices $A$ and $B$ since $A^{\otimes4}$ and $B_{(i)}\otimes B_{(j)}^{-1}$ on any sites $i,j$ generate the full local symmetry. This gives us the group of local invertible symmetries of each representative $\ket{\psi}$.
			
			\subparagraph{Type-1 symmetries}
			We now list all the local symmetries of type-1 derogatory representatives $\ket{S^4_2}+\ket{2^4},\;\ket{0^4}+\ket{S^4_2}+\ket{2^4}$, and $\{\ket{\psi(\mu)}:\mu\in\mathbb{C} \text{ and } \mu\neq0, \pm\sqrt{\frac{2}{3}}, \pm\sqrt{6}\}$. Any matrix $B$ that corresponds to the $B_{(i)}\otimes B_{(j)}^{-1}$ symmetries of these representatives must be $B=x\oplus x\oplus y=B_1$ in Eq.\;(\ref{eq:B_type12}) for non-zero $x,y\in\mathbb{C}$. Since $B$ does not include any non-derogatory matrices, these representatives are indeed derogatory. 
			
			As for the $A^{\otimes4}$ symmetries, there are only 3 possibilities for the matrix $A$ which are
			\begin{align}
				&A_1=\begin{pmatrix} a & 0 \\ 0 & \pm\frac{1}{a} \end{pmatrix}\oplus p,\;\;A_2=\begin{pmatrix} 0 & \pm\frac{1}{b} \\ b & 0 \end{pmatrix}\oplus p,\nonumber\\
				&A_3= c\begin{pmatrix} 1 & e^{i(\frac{\delta+\varphi}{2}+\beta)} \\ ie^{i(\frac{\delta+\varphi}{2}+\alpha)} & e^{i\delta} \end{pmatrix}\oplus p, \label{eq:A1A2A3}
			\end{align}
			where the allowed values for each parameter are different for each representative: 
			\begin{enumerate}[(i)]
				\item $\ket{S^4_2}+\ket{2^4}$ only has this symmetry with $A_1$ and $A_2$ for non-zero $a,b\in\mathbb{C}$ and $p=\pm1$ or $\pm i$. From arbitrary products of $A$ and $B$ of this representative, we deduce that its local symmetries are given by $S^{(j)}= \begin{pmatrix} x_j a & 0 \\ 0 & \pm\frac{x_j}{a} \end{pmatrix}\oplus y_j,\;S^{(j)}=\begin{pmatrix} 0 & \pm\frac{x_j}{b} \\ x_j b & 0 \end{pmatrix}\oplus y_j$, where $a,b$ must be the same for all sites (i.e., $j$-independent).
				\item $\ket{0^4}+\ket{S^4_2}+\ket{2^4}$ only has this symmetry with $A_1$ for $a,p=\pm1$ or $\pm i$. Similarly, its local symmetries can only take the form $S^{(j)}= x_j \oplus \pm x_j \oplus y_j$ (with $a$ absorbed in $x_j$).
				\item For any $\mu\neq0,\pm\sqrt{\frac{2}{3}},\pm\sqrt{6}$, $\ket{\psi(\mu)}\coloneqq \ket{0^4}+\ket{1^4}+\ket{2^4}+\mu\ket{S^4_2}$ has this symmetry with $A_1$ and $A_2$ for $a,b,p=\pm1$ or $\pm i$. Thus, its local symmetries can only take the two forms $S^{(j)}= x_j \oplus \pm x_j \oplus y_j$ and $S^{(j)}= \begin{pmatrix} 0 & \pm x_j \\ x_j & 0 \end{pmatrix}\oplus y_j$. Only for $\mu=\sqrt{2}i$ (and also for $\mu=-\sqrt{2}i$ which we do not consider separately here as it corresponds to the same SLOCC class), $\ket{\psi(\mu)}$ also has the $A^{\otimes4}$ symmetry with $A_3$ for $c=\begin{cases}\frac{1}{\sqrt{2}}e^{i\pi(\frac{m}{2}-\frac{1}{12})}, \text{\; if } e^{i(\delta+\varphi)}=1,\\
				\frac{1}{\sqrt{2}}e^{i\pi(\frac{m}{2}+\frac{1}{12})}, \text{\; if } e^{i(\delta+\varphi)}=-1,\end{cases}$ $(m=0,1,2,3)$, $\alpha,\beta\in\{0,\pi\}$, $\delta,\varphi\in\{\frac{\pi}{2},\frac{3\pi}{2}\}$ fulfilling $e^{i(\alpha+\beta+\varphi)}=i$ for $A_3$ being invertible, and $p=\pm1$ or $\pm i$. Hence, the local symmetry group of $\ket{\psi(\mu=\pm\sqrt{2}i)}$ composes of all the local symmetries for other $\mu$ and, additionally, the ones of $S^{(j)} = x'_j\begin{pmatrix} 1 & e^{i(\frac{\delta+\varphi}{2}+\beta)} \\ ie^{i(\frac{\delta+\varphi}{2}+\alpha)} & e^{i\delta} \end{pmatrix}\oplus y'_j$, where $\alpha,\beta,\delta,\varphi$ must be the same for all sites (i.e., $j$-independent).
			\end{enumerate}
			In addition, for $\bigotimes_{j=1}^4 S^{(j)}$ to stabilize the above representatives, all $x_j,y_j,x'_j,y'_j\in\mathbb{C}$ above must satisfy $\prod_{j=1}^4 x_j=1=\prod_{j=1}^4 y_j=\prod_{j=1}^4 y'_j$ and for $\mu=\sqrt{2}i$, $\prod_{j=1}^4 x'_j=\begin{cases} \frac{1}{4}e^{-i\frac{\pi}{3}}=\frac{1-\sqrt{3}i}{8}, \text{\;\; if } e^{i(\delta+\varphi)}=1,\\
			\frac{1}{4}e^{i\frac{\pi}{3}}=\frac{1+\sqrt{3}i}{8}, \text{\;\; if } e^{i(\delta+\varphi)}=-1.\end{cases}$
			
			\subparagraph{Type-2 symmetries}
			We give the full list of local symmetries of type-2 derogatory representatives $\ket{1^4}+\ket{S^4_2}+\ket{F^4_1}$ and $\ket{S^4_3}+\ket{F^4_1}$. These representatives can only have the $B_{(i)}\otimes B_{(j)}^{-1}$ symmetries with $B=\begin{pmatrix} x & 0 & y \\ 0 & x & 0 \\ 0 & 0 & x \end{pmatrix}$ for $x,y\in\mathbb{C}$ and $x\neq0$, which are generated by $B_2$ in Eq.\;(\ref{eq:B_type12}). Since $B$ does not include any non-derogatory matrices, these representatives must be derogatory. 
			
			The conditions for each representative to be stabilized by $A^{\otimes4}$ are as follows:
			\begin{enumerate}[(i)]
				\item $\ket{1^4}+\ket{S^4_2}+\ket{F^4_1}$ only has this symmetry with $A=a\oplus \pm a\oplus a$ where $a=\pm1$ or $\pm i$. 
				\item $\ket{S^4_3}+\ket{F^4_1}$ only has this symmetry with $A=a\oplus \frac{1}{a^{1/3}}e^{i\frac{2m\pi}{3}}\oplus \frac{1}{a^3}$ for non-zero $a\in\mathbb{C}$ and $m=0,1,2$. 
			\end{enumerate}
			Arbitrary products of the matrices $A$ and $B$ of each representative lead us to the general form of local symmetries
			\begin{align}
				\text{(i)}\;S^{(j)}=\begin{pmatrix} x_j & 0 & y_j \\ 0 & \pm x_j & 0 \\ 0 & 0 & x_j \end{pmatrix},  \text{(ii)}\;S^{(j)}=\begin{pmatrix} x_j a & 0 & y_j \\ 0 & \frac{x_j}{a^{1/3}}e^{i\frac{2m\pi}{3}} & 0 \\ 0 & 0 & \frac{x_j}{a^3} \end{pmatrix},\nonumber
			\end{align}
			where $x_j, y_j,a\in\mathbb{C}$, $x_j,a\neq0$ and $a$ is $j$-independent, for representatives (i) and (ii), respectively. In addition, for $\bigotimes_{j=1}^4 S^{(j)}$ to be a stabilizer of the above representatives, all $x_j,y_j\in\mathbb{C}$ must satisfy $\prod_{j=1}^4 x_j=1$ and $\sum_{j=1}^4 \frac{y_j}{x_j}=0$. We show more details on how we derive all the possible matrices $A$ in Appendix \ref{subsec:AllA4}.
			\\

			\paragraph{\textbf{Weak isolation, reachability via $\text{LOCC}_\mathbb{N}$ and convertibility via $\text{LOCC}_1$}}
			
			We will next prove that weakly isolated states, $\text{LOCC}_\mathbb{N}$-reachable states, and $\text{LOCC}_1$-convertible states exist in every derogatory ES SLOCC class. To show the existence of weakly isolated states (see Lemma \ref{lemma:isolation}), we construct a positive definite $3 \times 3$ complex matrix $G_j=g_j^\dagger g_j$ for any site $j$ such that $(S^{(j)})^\dagger G_j S^{(j)}\propto G_j$ cannot be satisfied for any non-trivial $S^{(j)}\in S^{(j)}_{\psi}$ (i.e., $S^{(j)}\not\propto\identity$). In fact, given that $G_j>0$,  $(S^{(j)})^\dagger G_j S^{(j)}\propto G_j$ holds only if $S^{(j)}$ is diagonalizable (see Corollary \ref{lemma:quasicommutation}). Hence, we only have to consider diagonalizable $S^{(j)}$ in the remaining section.
			
			We find that there always exist some matrices $G_j$ that violates  $(S^{(j)})^\dagger G_j S^{(j)}\propto G_j$ locally for every 4-qutrit derogatory ES SLOCC class. Two explicit examples are 
			\begin{align}
				G_j= \begin{pmatrix} \alpha & \beta^* & 0 \\ \beta & \delta & \varepsilon^* \\ 0 & \varepsilon & \nu \end{pmatrix},\; G_j= \begin{pmatrix} \alpha & \beta^* & \gamma^* \\ \beta & \delta & 0 \\ \gamma & 0 & \nu \end{pmatrix},\label{eq:GjViolatePropto}
			\end{align}
			with $\alpha,\delta,\nu>0$ and non-zero $\beta,\gamma,\varepsilon\in\mathbb{C}$ such that $G_j>0$. It is straightforward to check that $(S^{(j)})^\dagger G_j S^{(j)} - \lambda G_j=0$ cannot hold for any $\lambda$ and for any non-trivial local symmetry $S^{(j)}\in S^{(j)}_{\psi}$ of each representative $\ket{\psi}$ in Eq.\;\eqref{eq:4qutritDerogRep}. If we pick these matrices $G_j$ for more than two sites $j$ and any $3\times3$ positive definite matrices $G'_k=\begin{pmatrix} \alpha & \beta^* & \gamma^* \\ \beta & \delta & \epsilon^* \\ \gamma & \epsilon & \nu \end{pmatrix}$ with $\gamma\neq0$ or $\epsilon\neq0$ for the remaining sites $k$, then this will fulfill the weak isolation criterion in Lemma \ref{lemma:isolation}. Thus, we conclude that all 4-qutrit derogatory ES SLOCC classes contain weakly isolated states. 
			
			We also find that all 4-qutrit derogatory ES SLOCC classes contain states that are both reachable under $\text{LOCC}_\mathbb{N}$ and convertible under $\text{LOCC}_1$. We first construct such examples for all 4-qutrit derogatory ES SLOCC classes except for the class of the type-2 representative $\ket{S^4_3}+\ket{F^4_1}$, then we provide an example for the class represented by $\ket{S^4_3}+\ket{F^4_1}$. Let us consider  $\ket{\widetilde{\Psi}}\propto\bigotimes_{i=1}^4\widetilde{h}_i \ket{\psi}$ where $\ket{\psi}$ is of the form in Eq. \eqref{eq:4qutritDerogRep}, $\ket{\psi}\neq \ket{S^4_3}+\ket{F^4_1}$ and $\widetilde{h}_1^\dagger \widetilde{h}_1 =\widetilde{H}_1=\begin{pmatrix} \alpha & \beta^* & 0 \\ \beta & \delta & \varepsilon^* \\ 0 & \varepsilon & \nu \end{pmatrix}>0$ with $\beta,\varepsilon\neq0$. Note that $\widetilde{H}_1$ satisfies $S^{\dagger}\widetilde{H}_1 S\not\propto\widetilde{H}_1$ for $S=1\oplus-1\oplus1\in S^{(1)}_{\psi}$ \cite{footnote12} (i.e., satisfying condition (ii) in Theorem \ref{thm:reachability}), and $\widetilde{h}_j^\dagger \widetilde{h}_j =\widetilde{H}_j=\alpha_j \oplus\delta_j \oplus\nu_j>0$ commuting with $S$ for $j=2,3,4$ (i.e., fulfilling condition (i) in both Theorem \ref{thm:reachability} and Lemma \ref{lemma:convertibility}). A state $\bigotimes_{i=1}^4 g_i \ket{\psi}$ with $g_i=\widetilde{h}_i$ for $i=2,3,4$ and $g_1^\dagger g_1=G_1=\frac{1}{3}\widetilde{H}_1+\frac{2}{3}S^{\dagger}\widetilde{H}_1 S=\begin{pmatrix} \alpha & -\frac{\beta^*}{3} & 0 \\ -\frac{\beta}{3} & \delta & -\frac{\varepsilon^*}{3} \\ 0 & -\frac{\varepsilon}{3} & \nu \end{pmatrix}$ satisfying $\widetilde{H}_1\not\propto S'^{\dagger}G_1 S'$ for all $S'\in S^{(1)}_{\psi}$ can reach the state $\ket{\widetilde{\Psi}}$ via $\text{LOCC}_1$. The state $\ket{\widetilde{\Psi}}$ can also be converted to another state $\bigotimes_{i=1}^4 h_i \ket{\psi}$ with $h_1^\dagger h_1=H_1=\begin{pmatrix} \alpha & -2\beta^* & 0 \\ -2\beta & \delta & -2\varepsilon^* \\ 0 & -2\varepsilon & \nu \end{pmatrix}>0$ and $h_i=\widetilde{h}_i$ for $i=2,3,4$ via $\text{LOCC}_1$ because $\widetilde{H}_1=\frac{1}{4}H_1+\frac{3}{4}S^{\dagger}H_1 S$, $H_1\not\propto \widetilde{S}^{\dagger}\widetilde{H}_1 \widetilde{S}$ for all $\widetilde{S}\in S^{(1)}_{\psi}$, and $[S,H_{i\geq2}]=0$ together satisfy Lemma \ref{lemma:convertibility}. Therefore, $\ket{\widetilde{\Psi}}$ is both $\text{LOCC}_\mathbb{N}$-reachable and $\text{LOCC}_1$-convertible.
			
			We now construct an $\text{LOCC}_{\mathbb{N}}$-reachable and $\text{LOCC}_1$-convertible example for the SLOCC class of $\ket{S^4_3}+\ket{F^4_1}$. We consider the same state $\ket{\widetilde{\Psi}}$ as before (but now $\ket{\psi}=\ket{S^4_3}+\ket{F^4_1}$), which satisfies $S'^{\dagger}\widetilde{H}_1 S'\not\propto\widetilde{H}_1$ for $S'=1\oplus-1\oplus-1\in S^{(1)}_{\psi}$. The state $\bigotimes_{i=1}^4 g'_i \ket{\psi}$ with $g_1'^\dagger g'_1 =G'_1 =\frac{1}{3}\widetilde{H}_1+\frac{2}{3}S'^{\dagger}\widetilde{H}_1 S'=\begin{pmatrix} \alpha & -\frac{\beta^*}{3} & 0 \\ -\frac{\beta}{3} & \delta & \varepsilon^* \\ 0 & \varepsilon & \nu \end{pmatrix}$ satisfying $\widetilde{H}_1\not\propto S''^{\dagger}G'_1 S''$ for all $S''\in S^{(1)}_{\psi}$ and $g'_i=\widetilde{h}_i$ for $i=2,3,4$ can reach the state $\ket{\widetilde{\Psi}}$ via $\text{LOCC}_1$. Similarly, the state $\ket{\widetilde{\Psi}}$ can also be converted to a different state $\bigotimes_{i=1}^4 h'_i \ket{\psi}$ with $h_1'^\dagger h'_1=H'_1=\begin{pmatrix} \alpha & -2\beta^* & 0 \\ -2\beta & \delta & \varepsilon^* \\ 0 & \varepsilon & \nu \end{pmatrix}>0$ and $h'_i=\widetilde{h}_i$ for $i=2,3,4$ via $\text{LOCC}_1$ since $\widetilde{H}_1=\frac{1}{4}H'_1+\frac{3}{4}S'^{\dagger}H'_1 S'$, $H'_1\not\propto \widetilde{S}'^{\dagger}\widetilde{H}_1 \widetilde{S}'$ for all $\widetilde{S}'\in S^{(1)}_{\psi}$, and $[S',H'_{i\geq2}]=0$ satisfy Lemma \ref{lemma:convertibility}. Thus, $\ket{\widetilde{\Psi}}$ is again both $\text{LOCC}_\mathbb{N}$-reachable and $\text{LOCC}_1$-convertible.
			
			There are also states that are $\text{LOCC}_1$-convertible but not $\text{LOCC}_\mathbb{N}$-reachable in some 4-qutrit derogatory ES SLOCC classes. For example, in the SLOCC class of the type-1 representative $\ket{\psi}=\ket{0^4}+\ket{S^4_2}+\ket{2^4}$, a state $\ket{\Psi}\propto\bigotimes_{i=1}^4 g_i \ket{\psi}$ with $g_i^\dagger g_i=G_i=\alpha_i \oplus\delta_i \oplus\nu_i$ for $i=1,\ldots,4$ is not $\text{LOCC}_\mathbb{N}$-reachable because $[S,G_i]=0$ for all $S\in S^{(i)}_{\psi}$ which are all diagonal (see Sec.\;\ref{par:4qutritSymm}) and for all $i$, thereby violating condition (ii) in Theorem \ref{thm:reachability}. However, there exists a state $\ket{\Phi}\propto h_1\bigotimes_{i=2}^4 g_i \ket{\psi}$ with $h_1^\dagger h_1 =H_1=\begin{pmatrix} \alpha_1 & \beta^* & 0 \\ \beta & \delta_1 & \varepsilon^* \\ 0 & \varepsilon & \nu_1 \end{pmatrix}>0$ with $\beta,\varepsilon\neq0$ and an $S=1\oplus-1\oplus1\in S^{(1)}_{\psi}$ such that $G_1=\frac{1}{2}(H_1+S^\dagger H_1 S)$ and $H_1\not\propto \widetilde{S}^{\dagger}G_1 \widetilde{S}$ for all $\widetilde{S}\in S^{(1)}_{\psi}$. This satisfies both conditions in Lemma \ref{lemma:convertibility}, so the initial state $\ket{\Psi}$ can be converted into $\ket{\Phi}$ with $\text{LOCC}_1$.

			As a remark, there are more weakly isolated states, $\text{LOCC}_\mathbb{N}$-reachable states, and $\text{LOCC}_1$-convertible states than what we have presented here. We summarize our findings regarding 4-qutrit derogatory ES SLOCC classes in the following observation.
			
			\begin{observation}
				The representatives of all 4-qutrit derogatory ES SLOCC classes are
				\begin{align}
					&\ket{S^4_2}+\ket{2^4},\;\ket{0^4}+\ket{S^4_2}+\ket{2^4},\;\ket{1^4}+\ket{S^4_2}+\ket{F^4_1},\nonumber\\
					&\ket{S^4_3}+\ket{F^4_1},\;\{\ket{\psi(\mu)}:\mu\in\mathbb{C} \text{ and } \mu\neq0, \pm\sqrt{\frac{2}{3}}, \pm\sqrt{6}\}.\nonumber
				\end{align}
				All these classes contain both weakly isolated states, $\text{LOCC}_\mathbb{N}$-reachable states, and $\text{LOCC}_1$-convertible states. \label{Obs:4Qutrit}
			\end{observation}
			
			\paragraph{\textbf{Convertible via LOCC with probabilistic steps}}\label{sec:probStepProtocol}
			We have already shown that all 4-qutrit derogatory ES SLOCC classes possess rich entanglement structures in the sense that each class contains examples that are either weakly isolated, $\text{LOCC}_\mathbb{N}$-reachable, or $\text{LOCC}_1$-convertible. In fact, some of these classes contain states that are convertible only with more intricate LOCC protocols. In the following, we will give an example of a state in the SLOCC class represented by the state $|\psi(\mu=\sqrt{2}i)\rangle$ in Eq.\;(\ref{eq:PsiMu}), which can only be converted to another LU-inequivalent state via an LOCC protocol with intermediate probabilistic steps \cite{SpdV17,dVSp17}. This kind of protocol converts the initial state into some LU-inequivalent states probabilistically via local measurements in the intermediate steps before reaching the same target state in the final step. In contrast to most previously studied LOCC transformations, there exist examples of this type of transformations (such as the one presented below) which are not achievable with an \textit{all-det-$\text{LOCC}_\mathbb{N}$} protocol which consists only of deterministic rounds.
			
			In this example, we consider the SLOCC class of $\ket{\Psi_s}\coloneqq|\psi(\mu=\sqrt{2}i)\rangle$ and the goal is to convert the initial state $\ket{\Psi}\propto g_1\otimes g_2\otimes\identity\otimes\identity\ket{\Psi_s}$ into the target state $\ket{\Phi}\propto h_1\otimes h_2\otimes\identity\otimes\identity\ket{\Psi_s}$ where $g_i=\sqrt{G_i}$ and $h_i=\sqrt{H_i}$ ($i=1,2$) are defined by the positive-definite matrices
			\begin{align}
				&G_1=\begin{pmatrix}
					2p+2 & 1-(2p-1)i & p(1-\sqrt{3}i)\\
					1+(2p-1)i & 4-2p & 0\\
					p(1+\sqrt{3}i) & 0 & 2
				\end{pmatrix},\label{eq:G1_def}\\
				&G_2=\begin{pmatrix}
					3 & 3-\frac{2}{1-q}-i & \frac{3+\sqrt{2}+i}{\sqrt{2}}q\\
					3-\frac{2}{1-q}+i & 5 & 0\\
					\frac{3+\sqrt{2}-i}{\sqrt{2}}q & 0 & 4
				\end{pmatrix},\label{eq:G2_def}\\
				&H_1=\begin{pmatrix}
					4 & 1-i & \frac{1+\sqrt{3}}{2}(1-i)\\
					1+i & 2 & 1\\
					\frac{1+\sqrt{3}}{2}(1+i) & 1 & 2
				\end{pmatrix},\\
				&H_2=\begin{pmatrix}
					5 & -i-\frac{1+q}{1-q} & \frac{1+i}{\sqrt{2}}\\
					i-\frac{1+q}{1-q} & 3 & 1\\
					\frac{1-i}{\sqrt{2}} & 1 & 4
				\end{pmatrix},\label{eq:H2_def}
			\end{align}
			and the parameters $p=\frac{1-\sqrt{2}+\sqrt{3}}{2}$ and $q=\frac{\sqrt{2+\sqrt{2}}-1}{1+\sqrt{2}}$. In Appendix \ref{app:NoAllDetLOCC}, we show that no all-det-$\text{LOCC}_\mathbb{N}$ protocol can convert $\ket{\Psi}$ to $\ket{\Phi}$.

			The LOCC protocol that transforms $\ket{\Psi}$ to $\ket{\Phi}$ with a probabilistic intermediate step consists of two steps. The first step involves party 1 measuring with $M_1=\sqrt{p}h_1g_1^{-1}$ and $M_2=\sqrt{1-p}h_1S^{(1)}g_1^{-1}$ where $p=\frac{1-\sqrt{2}+\sqrt{3}}{2}$ and $S^{(1)}=\begin{pmatrix}
			0 & e^{i\frac{3\pi}{4}}\\
			e^{i\frac{3\pi}{4}} & 0
			\end{pmatrix}\oplus1$, which fulfill $\sum_{i=1}^2 M_i^\dagger M_i=\identity$. This results in two (unnormalized) intermediate states $h_1\otimes g_2\otimes\identity\otimes\identity\ket{\Psi_s}$ and $h_1S^{(1)}\otimes g_2\otimes\identity\otimes\identity\ket{\Psi_s}$ which are LU-inequivalent due to $S^{(1)\dagger} H_1 S^{(1)}\otimes G_2\otimes\identity^{\otimes2}\not\propto \widehat{S}^\dagger (H_1\otimes G_2\otimes\identity^{\otimes2})\widehat{S}$ for all $\widehat{S}=\bigotimes_{j=1}^4 \widehat{S}^{(j)}\in S_{\Psi_s}$, and therefore, the protocol is not all-deterministic. As we will see later, there exist local unitary symmetries $S,S'\not\propto\identity$ satisfying $S^{\dagger} H_1 S\propto H_1$ and $S'^{\dagger}S^{(1)\dagger} H_1 S^{(1)}S'\propto H_1$ (with $S'=(S^{(1)})^{-1}S$), such that both states satisfy condition (i) in Lemma \ref{lemma:convertibility} for deterministic $\text{LOCC}_1$ conversions in the next step. 
			
			In the first branch, party 2 measures with $\overline{M}_1=\sqrt{q}h_2g_2^{-1}$ and $\overline{M}_2=\sqrt{1-q}h_2\overline{S}^{(2)}g_2^{-1}$ where $q=\frac{\sqrt{2+\sqrt{2}}-1}{1+\sqrt{2}}$, $\overline{S}^{(2)}=\overline{x}_2\begin{pmatrix}
			1 & -i\\
			1 & i
			\end{pmatrix}\oplus1$, and $\overline{x}_2=\frac{q(1+\sqrt{2}+i)}{2(1-q)}$, which satisfy $\sum_{i=1}^2 \overline{M}_i^\dagger \overline{M}_i=\identity$ and result in two (unnormalized) states $h_1\otimes h_2\otimes\identity\otimes\identity\ket{\Psi_s}$ and $h_1\otimes h_2\overline{S}^{(2)}\otimes\identity\otimes\identity\ket{\Psi_s}$. If party 2 obtains the second measurement outcome, party 1 will apply $\overline{U}=h_1 \overline{S}^{(1)} h_1^{-1}$ with $\overline{S}^{(1)}=\overline{x}_1\begin{pmatrix}
			1 & -i\\
			1 & i
			\end{pmatrix}\oplus\overline{y}_1$, $\overline{x}_1=\frac{1+\sqrt{3}i}{4\overline{x}_2}$, and $\overline{y}_1=\frac{1+i-\sqrt{6}e^{i\frac{3\pi}{4}}}{2}\overline{x}_1$ (where $\overline{U}$ is unitary due to $[H_1, \overline{S}^{(1)}]=0$ and $\overline{S}^{(1)}$ being unitary). Parties 3 and 4 will apply unitaries $\overline{S}^{(3)}=\frac{1}{\sqrt{2}}\begin{pmatrix}
			1 & -i\\
			1 & i
			\end{pmatrix}\oplus1$ and $\overline{S}^{(4)}=\frac{1}{\sqrt{2}}\begin{pmatrix}
			1 & -i\\
			1 & i
			\end{pmatrix}\oplus\frac{1}{\overline{y}_1}$, respectively. With $\bigotimes_{j=1}^4 \overline{S}^{(j)}\ket{\Psi_s}=\ket{\Psi_s}$, we obtain the target state $\ket{\Phi}$.

			In the second branch, party 2 measures with $\widetilde{M}_1=\sqrt{q}h_2S^{(2)}g_2^{-1}$ and $\widetilde{M}_2=\sqrt{1-q}h_2\widetilde{S}^{(2)}S^{(2)}g_2^{-1}$ where $q=\frac{\sqrt{2+\sqrt{2}}-1}{1+\sqrt{2}}$, $S^{(2)}=\begin{pmatrix}
			0 & 1\\
			1 & 0
			\end{pmatrix}\oplus\frac{1+i}{\sqrt{2}}$, $\widetilde{S}^{(2)}=\widetilde{x}_2\begin{pmatrix}
			1 & 1\\
			i & -i
			\end{pmatrix}\oplus1$, and $\widetilde{x}_2=i\frac{q(1+\sqrt{2}+i)}{2(1-q)}$, which fulfill $\sum_{i=1}^2 \widetilde{M}_i^\dagger \widetilde{M}_i=\identity$ and result in two (unnormalized) states $h_1S^{(1)}\otimes h_2S^{(2)}\otimes\identity\otimes\identity\ket{\Psi_s}$ and $h_1S^{(1)}\otimes h_2\widetilde{S}^{(2)}S^{(2)}\otimes\identity\otimes\identity\ket{\Psi_s}$. For the first measurement outcome, parties 3 and 4 will apply unitaries $S^{(3)}=\begin{pmatrix}
			0 & 1\\
			1 & 0
			\end{pmatrix}\oplus1$ and $S^{(4)}=\begin{pmatrix}
			0 & e^{-i\frac{3\pi}{4}}\\
			e^{-i\frac{3\pi}{4}} & 0
			\end{pmatrix}\oplus\frac{1-i}{\sqrt{2}}$, respectively. If the second measurement outcome occurs, party 1 applies $\widetilde{U}=h_1 \widetilde{S}^{(1)} h_1^{-1}$ with $\widetilde{S}^{(1)}=\widetilde{x}_1\begin{pmatrix}
			1 & 1\\
			i & -i
			\end{pmatrix}\oplus\widetilde{y}_1$, $\widetilde{x}_1=\frac{1-\sqrt{3}i}{4\widetilde{x}_2}$, and $\widetilde{y}_1=\frac{1-i+\sqrt{6}e^{i\frac{\pi}{4}}}{2}\widetilde{x}_1$ (where $\widetilde{U}$ is unitary since $[H_1, \widetilde{S}^{(1)}]=0$ and $\widetilde{S}^{(1)}$ is unitary), whereas parties 3 and 4 apply unitaries $\widetilde{S}^{(3)}S^{(3)}$ and $\widetilde{S}^{(4)}S^{(4)}$ where $\widetilde{S}^{(3)}=\frac{1}{\sqrt{2}}\begin{pmatrix}
			1 & 1\\
			i & -i
			\end{pmatrix}\oplus1$ and $\widetilde{S}^{(4)}=\frac{1}{\sqrt{2}}\begin{pmatrix}
			1 & 1\\
			i & -i
			\end{pmatrix}\oplus\frac{1}{\widetilde{y}_1}$. Given that $\bigotimes_{j=1}^4 S^{(j)}\ket{\Psi_s}=\ket{\Psi_s}$ and $\bigotimes_{j=1}^4 \widetilde{S}^{(j)}\ket{\Psi_s}=\ket{\Psi_s}$, the intermediate states in the second branch can also be converted into the target state $\ket{\Phi}$ using a single round of LOCC.

			To summarise the example, the first step is probabilistic, which gives rise to two LU-inequivalent branches, but the second step is deterministic in each branch.
			
			Note that one can construct another 4-qutrit example that uses an intermediate probabilistic step by simply embedding the $2\times2$ matrices from the 4-qubit example in Refs.\;\cite{SpdV17,dVSp17} in the subspace $\text{span}\{\ket{0},\ket{1}\}$ and adding an orthogonal component to the third dimension via a direct sum. However, the argument in Refs.\;\cite{SpdV17,dVSp17} for why no all-det-$\text{LOCC}_\mathbb{N}$ protocol can achieve such transformation cannot be applied here directly because there are more non-trivial symmetries in 4 qutrits than in 4 qubits. For instance, the trivial symmetry $S=\identity$ for qubits can become non-trivial after extending it to 3 dimensions (e.g., $S'=\identity\oplus e^{i\varphi}$ with $\varphi\neq0$).
			
			\section{Conclusion}
			
			In this work we have exhaustively analyzed deterministic LOCC transformations in SLOCC classes that contain a symmetric state. Since these conversions are almost never possible for generic states, only zero-measure subsets of states can display a rich LOCC structure. Symmetric states (together with their SLOCC classes) indeed form such a zero-measure subset which, moreover, has a clear physical and mathematical relevance. In fact, the only SLOCC classes known so far that are free of isolation correspond to the ones of the symmetric $n$-qubit W and GHZ states. Our results indicate that nevertheless symmetric SLOCC classes are in general equally limited for LOCC manipulation as the full state space. We have shown that, as in the general case, almost every symmetric $n$-qubit state with $n\geq 5$ has a trivial local stabilizer. Thus, generic states of this form are isolated. Turning to symmetric states of arbitrary local dimension, we have proven that transformations among non-ES states are never possible and that weak isolation exists in the corresponding classes. This motivates the study in greater detail of ES classes, which in particular have a richer stabilizer (and which we have moreover proven to be non-generic within symmetric classes at least for $n$-qubit states). In fact, we have characterized the local stabilizer for these classes in the non-derogatory case and we have found that LOCC transformations among symmetric states are in this case possible. However, we have provided a plethora of results that show that these families are still plagued with (weak) isolation. Nevertheless, the identified possible transformations in this case might turn out to be helpful in finding applications of entangled ES states. Finally, we have considered derogatory ES classes, where a full characterization of the SLOCC classes and their representatives remains open. Notwithstanding, we have observed that they admit classes in which all states are isolated under LOCC$_{\mathbb{N}}$. In addition to this, we have characterized the local stabilizer and study in detail the LOCC convertibility properties of derogatory ES classes for 4-qutrit states.
			
			On a technical level, we have obtained other results that might be of future use. First, in order to study LOCC$_{\mathbb{N}}$ transformability we have generalized the results of \cite{SpdV17,dVSp17} from SLOCC classes with a finite stabilizer to an arbitrary one, where the commutation conditions of the former case translate to quasi-commutation relations. We expect that these results will be helpful when studying LOCC transformations for other classes of states. Second, we have studied properties of the local stabilizer for several SLOCC families of symmetric states and, in particular, we have determined it for non--derogatory ES SLOCC classes, a result that can be of use for other applications of these states in quantum information theory and quantum many-body physics.
			
			A question that remains open is the existence of SLOCC families with no isolated states beyond the $n$-qubit W and GHZ classes. Further studies of the structure of SLOCC classes generated by states displaying different forms of symmetry than permutation symmetry and translational symmetry, as studied in \cite{SaMo19,He21} might prove fruitful in this direction. The results presented here can also be used to study transformations from pure states to ensembles of pure states. 
			
			\begin{acknowledgements}
				MH, CS, NKHL and BK acknowledge financial support from the Austrian Science Fund (FWF): W1259-N27 (DK-ALM), F 7107-N38 (SFB BeyondC) and P 32273-N27 (Stand-Alone Project).  JIdV acknowledges the following institutions for financial support: Spanish Ministerio de Ciencia e Innovaci\'{o}n (grant PID2020-113523GB-I00) and Comunidad de Madrid (grant QUITEMAD-CMS2018/TCS-4342 and the Multiannual Agreement with UC3M in the line of Excellence of University Professors EPUC3M23 in the context of the V PRICIT). 
				
				We thank J. Ignacio Cirac for helpful discussions regarding the quasi-commutation relation.
			\end{acknowledgements}

			\appendix
			\section{Proofs of Theorem \ref{thm:reachability} and Lemma \ref{lemma:convertibility}}
			\label{app:proof_reach_conv}
			In this Appendix we provide the proofs of the necessary and sufficient conditions for pure states to be reachable (convertible) via $LOCC_{\mathbb{N}}$ (by one round of an LOCC protocol) respectively. For better readability we repeat the corresponding Theorem (Lemma) respectively. The proofs are analogous to the ones for finite (unitary) symmetries presented in \cite{SpdV17,dVSp17}. However, for the sake of completeness we recall them here. \\
			
			\noindent{\bf Theorem \ref{thm:reachability}.}
			\textit{A state $\ket{\Phi}\propto h \ket{\Psi_s}$ is reachable via $LOCC_{\mathbb{N}}$, iff there exists $S \in S_{\Psi_s}$ such that the following conditions hold up to permutations of the particles:
				\begin{itemize}
					\item[(i)] For any $i\geq 2$ $ (S^{(i)})^\dagger H_i S^{(i)} \propto H_i$ and
					\item[(ii)] $  (S^{(1)})^\dagger H_1 S^{(1)} \not\propto H_1$.
			\end{itemize}}
			
			\begin{proof} The proof is analogous to the one for unitary symmetries in \cite{SpdV17}. 
				We will first show that the conditions in the theorem are necessary and then present an explicit example of a transformation that allows to reach $\ket{\Phi}$.
				
				The protocol involves only finitely many rounds and is non-trivial, hence, there exists a last non-trivial round of the protocol. In order for the transformation to be deterministic any state $\ket{\chi}$ which is obtained in one branch of the LOCC protocol in the second to last round has to be transformed in the last round deterministically to $\ket{\Phi}$. It is straightforward to see that therefore these states have to be  in the same SLOCC class, i.e., $\ket{\chi}\propto g \ket{\Psi_s}$ for some $g\in G$.  We assume then that w.l.o.g. party $1$ applies a non-trivial measurement in the last step, which is described by the operators $\{A_i\}$, and all the other parties apply (depending on the outcome) a LU. Note that at least two outcomes are not related to each other via a unitary, i.e., it has to hold that $A_2^\dagger A_2 \not\propto A_1^\dagger A_1$, as otherwise the last round would be trivial. Moreover, we have that $(A_1\otimes \identity) g\ket{\Psi_s}=r_1 \otimes_{i=2}^n U_i h \ket{\Psi_s}$ and $(A_2\otimes \identity)g\ket{\Psi_s}=r_2 \otimes_{i=2}^n V_i h \ket{\Psi_s}$
				for some local unitaries $U_i$ and $V_i$ and $r_1,r_2 > 0$. It is easy to see that the latter equations are equivalent to
				\begin{align}
					h^{-1} (\otimes_{i=2}^n U_i^\dagger ) (A_1\otimes \identity) g&=r_1 S_1 \\
					h^{-1} (\otimes_{i=2}^n V_i^\dagger ) (A_2\otimes \identity) g&=r_2 S_2,
				\end{align}
				where $S_1,S_2 \in S_{\Psi_s}$. This implies that 
				\begin{align}\label{eq:A}
					A_1&=r^{(1)}_1 h_1S_1^{(1)}g_1^{-1}, A_2=r^{(1)}_2 h_1S_2^{(1)}g_1^{-1}\\
					g_i&=r^{(i)}_1 U_i h_i S_1^{(i)}=r^{(i)}_2 V_i h_iS_2^{(i)}, \quad \forall i>1,
				\end{align}
				where $r_j=\prod_i r^{(i)}_j$, for $j=1,2$. Using the last equations for $g_i^\dagger g_i$ and  that $h_i$ and $S_2^{(i)}$ are invertible it follows that condition (i) in Theorem \ref{thm:reachability} has to hold for $S = S_1 S_2^{-1}$. Using further that $A_1^\dagger A_1 \not\propto A_2^\dagger A_2$ and Eq. (\ref{eq:A}) one obtains condition (ii) for $S = S_1 S_2^{-1}$.
				
				We will next show that $\ket{\Phi}$ is reachable via $LOCC_\mathbb{N}$ if conditions (i) and (ii) are fulfilled. In order to do so we construct a transformation and choose the initial state $\ket{\Psi}\propto g \ket{\Psi_s}$ of the transformation such that for
				$i>1$ $G_i = H_i \propto (S^{(i)})^\dagger H_i S^{(i)}$, i.e. we choose $g_i=V_i h_i \propto W_i h_i S^{(i)}$, for some unitaries $V_i,W_i$. Note that these have to exist as condition (i) implies that
				$h_i S^{(i)} (h_i)^{-1}$ is up to proportionality factor a
				unitary.  Further, we choose $g_1$ such that $r G_1=p  H_1 + (1-p) (S^{(1)})^\dagger H_1 S^{(1)}$, for some $0<p<1$ and $r=p+(1-p)\tr((S^{(1)})^\dagger H_1 S^{(1)})$ \cite{footnoteAppA}.  Then the following LOCC protocol allows to transform $g\ket{\Psi_s}$ to $h\ket{\Psi_s}$. Party $1$ performs a generalized measurement with measurement operators,
				$ \frac{\sqrt{p}}{\sqrt{r}}h_1 g_1^{-1},
				\frac{\sqrt{1-p}}{\sqrt{r}} h_1 S^{(1)} g_1^{-1}$. All the other parties, $i$, apply then depending on the measurement outcome either $V^\dagger_i$ or $W^\dagger_i$, respectively.
			\end{proof}

			\noindent {\bf Lemma \ref{lemma:convertibility}.}
			\textit{A state $\ket{\Psi}\propto g \ket{\Psi_s}$ is convertible via $LOCC_{1}$ iff there exist $m$ symmetries $S_k \in S_{\Psi_s}$, with $m>1$ and $H_1>0$ and $p_k > 0$ with $\sum_{k=1}^m p_k = 1$, such that the following conditions hold up to permutations of the particles:
				\begin{itemize}
					\item[(i)] $(S_k^{(i)})^\dagger G_i S_k^{(i)} \propto G_i$ for any $i\geq 2$  and for all $k \in \{1, \ldots, m\}$ and
					\item[(ii)] $G_1 = \sum_{k=1}^m p_k \left( S_k^{(1)} \right)^\dagger H_1 S_k^{(1)}$ and  $H_1 \not\propto \left(S^{(1)}\right)^\dagger G_1 S^{(1)}$ for any $S \in S_{\Psi_s}$ fulfilling $(S^{(i)})^\dagger G_i S^{(i)} \propto G_i$ for all $i\geq 2$.
			\end{itemize}}
			
			\begin{proof}
				The proof is a straightforward generalization of the proof of Lemma 3 in \cite{dVSp17}. 
				
				\textit{If:} If conditions (i) and (ii) are fulfilled, then the following $LOCC_1$-transformation allows to non-trivially convert the state. The first party performs a generalized measurement with measurement operators
				$A_k = \sqrt{p_k} h S_k g_1^{-1}$. Condition (ii) ensures that $\sum_k A_k^\dagger A_k = \identity$. For outcome $k$  one obtains the state $h_1 S_k^{(1)} \otimes g_2 \otimes \ldots \otimes g_n \ket{\Psi_s}$. Due to condition (i) there exist unitaries $U_k^{(i)}$ satisfying $g_i S_k^{(i)} \propto U_k^{(i)} g_i$. Depending on the measurement outcome $k$ parties $i\in\{2, \ldots, n\}$ apply now the unitaries $ U_k^{(i)}$ in order to obtain $h_1 \otimes g_2 \otimes \ldots \otimes g_n \ket{\Psi_s}$ for all measurement outcomes. Note that if the second part of condition (ii) is fulfilled this transformation is non-trivial.
				
				\textit{Only if:}
				Wlog we assume in the following that the non-trivial measurement is implemented by the first party.
				In order for a state to be convertible via $LOCC_1$ there has to exist a state $\tilde{h}_1 \otimes \ldots \otimes \tilde{h}_n \ket{\Psi_s}$, a measurement $\{A_k\}_{k=1}^{m}$, $\sum_{k=1}^m A_k^\dagger A_k = \identity$ and unitaries $U_k^{(i)}$ (for $k \in \{2, \ldots, m\}$) such that for all $k\in \{1, \ldots, m\}$,
				\begin{align} \label{eq:LOCC1}
					A_k g_1 \otimes U_k^{(2)} g_2 \otimes \ldots \otimes U_k^{(n)} g_n \ket{\Psi_s} \propto \tilde{h}_1 \otimes \ldots \otimes \tilde{h}_n \ket{\Psi_s}.
				\end{align}
				Note that as the transformation should be non-trivial it has to hold that $\tilde{H} \not\propto S^\dagger G S$ for all $S \in S_{\Psi_s}$. Eq. (\ref{eq:LOCC1}) can only hold true if there exist $\tilde{S}_k \in S_{\Psi_s}$ such that
				\begin{align}
					A_k &\propto \tilde{h}_1 \tilde{S}_k^{(1)} g_1^{-1}, \\
					U_k^{(i)} &\propto \tilde{h}_i \tilde{S}_k^{(i)} g_i^{-1} \ \forall i \geq 2.
				\end{align}
				Using then that $\sum_{k=1}^m A_k^\dagger A_k = \identity$ and that $U_k^{(i)}$ are unitaries we obtain that there have to exist $p_k > 0$ such that (for an appropriately chosen normalization of $\tilde{H}$) 
				\begin{align}
					\sum_k p_k \left(\tilde{S}_k^{(1)} \right)^\dagger \tilde{H}_1 \tilde{S}_k^{(1)} &= G_1 \\
					\left(\tilde{S}_k^{(i)} \right)^\dagger \tilde{H}_i \tilde{S}_k^{(i)} &\propto G_i \ \forall i\geq 2 \ \forall k .
				\end{align}
				With the definitions $S_k^{(i)} = \left(\tilde{S}_1^{(i)}\right)^{-1} \tilde{S}_k^{(i)}$ and $H_1 =  \left(\tilde{S}_1^{(1)} \right)^\dagger \tilde{H_1} \tilde{S}_1^{(1)}$ it therefore has to hold true that 
				$\sum_k p_k \left({S}_k^{(1)} \right)^\dagger {H}_1 {S}_k^{(1)} = G_1$, $\left( S_k^{(i)} \right)^\dagger G_i S_k^{(i)} \propto G_i$ for all $i \geq 2$, and $H_1 \not\propto \left(S^{(1)}\right)^\dagger G_1 S^{(1)}$ for any $S \in S_{\Psi_s}$ such that  $\left( S^{(i)} \right)^\dagger G_i S^{(i)} \propto G_i$ for all $i \geq 2$. This completes the proof.
			\end{proof}
			
			\section{On the quasi-commutation relation}
			\label{app:quasicommutation2}
			Here we provide the proofs of Observation \ref{obs:quasicommutation2} and Observation \ref{obs:quasicommutation3}. In order to improve readability we repeat the observations here.\\
			
			\noindent {\bf Observation \ref{obs:quasicommutation2}.} \textit{
				Let $A$ be a positive $k \times k$ matrix and $B$ be an arbitrary $k \times k$ matrix. Then, $B^\dagger A B \propto A$ if and only if
				$B \propto a^{-1} U a$ for some unitary $U$ and some $a$ s.t. $a^\dagger a = A$.}
			\begin{proof} 
				It can be easily verified that if $B$ and $A$ are as in the statement, then the relation $B^\dagger A B \propto A$ is indeed satisfied. Let us now show the converse direction. Assume $B^\dagger A B \propto A$. As $A$ is positive, there exists a matrix $a$ such that $A = a^\dagger a$. Thus, $B^\dagger a^\dagger a B \propto a^\dagger a$, or, equivalently, $\left(a B a^{-1}\right)^\dagger a B a^{-1}\propto \identity$. Hence, $a B a^{-1}\propto U$ for some unitary $U$ and the statement of the lemma follows.
			\end{proof}

			\noindent {\bf Observation \ref{obs:quasicommutation3}.} \textit{
				Let $A$ be a positive $k \times k$ matrix and $B$ be an arbitrary $k \times k$ matrix. Then, $B^\dagger A B \propto A$ if and only if the following two conditions are met. 
				\begin{enumerate}[(i)]
					\item $B$ is (up to proportionality) similar to a unitary matrix, i.e., $B \propto R \operatorname{diag}(e^{i \phi_1}, \ldots, e^{i \phi_k}) R^{-1}$ for some $k \times k$ matrix $R$ and $\phi_1, \ldots, \phi_k \in [0,2\pi)$.
					\item $A \propto R^{-\dagger} X R^{-1}$, where $X$ is a direct sum of positive matrices acting on the degenerate subspaces of $\operatorname{diag}(e^{i \phi_1}, \ldots, e^{i \phi_k})$.
				\end{enumerate}
			}
			Note that in the particular case in which the eigenvalues of $A$ are non-degenerate, i.e., $\phi_i$ are pairwise different, $X$ is a diagonal matrix with positive entries.
			\begin{proof} 
				We will first show the sufficient part of the observation and afterwards the necessary part.
				Assume that $A$ and $B$ fulfill the conditions (i) and (ii). We then have
				\begin{align}
					&B^\dagger A B \propto \nonumber \\
					&  \quad \propto R^{-\dagger} \operatorname{diag}(e^{-i \phi_1}, \ldots, e^{-i \phi_k})  X  \operatorname{diag}(e^{i \phi_1}, \ldots, e^{i \phi_k}) R^{-1} \nonumber \\
					&  \quad = R^{-\dagger}   X  R^{-1} = A,
				\end{align}
				where we have used that $ \operatorname{diag}(e^{i \phi_1}, \ldots, e^{i \phi_k}) $ and $X$ commute due to the condition on $X$. This shows the sufficient part of the observation.
				
				Let us now prove the necessary part. Assume $B^\dagger A B \propto A$. Statement (i) follows from Observation \ref{obs:quasicommutation2}. Using that $B$ is as in statement (i) we have that
				\begin{align}
					\operatorname{diag}(e^{-i \phi_1}, \ldots, e^{-i \phi_k})  R^\dagger A R  \operatorname{diag}(e^{i \phi_1}, \ldots, e^{i \phi_k}) \nonumber\\
					\propto R^\dagger A R.
				\end{align}
				Using the abbreviation $X=R^\dagger A R$ and denoting the proportionality factor by $\lambda$ we have
				\begin{align}
					\operatorname{diag}(e^{-i \phi_1}, \ldots, e^{-i \phi_k}) X  \operatorname{diag}(e^{i \phi_1}, \ldots, e^{i \phi_k}) = \lambda X.
				\end{align}
				As $X$ is positive, the diagonal entries of $X$ must be positive numbers and we have $\lambda = 1$. We hence obtain that entries $X_{i,j}$ must vanish whenever $\phi_i \neq \phi_j$. Thus, $X$ is a positive matrix that is composed of a direct sum of (positive) matrices acting on the degenerate subspace of $\operatorname{diag}(e^{i \phi_1}, \ldots, e^{i \phi_k})$. Recalling that $A = R^{-\dagger} X R^{-1}$, statement (ii) follows. This completes the proof of the lemma.
			\end{proof}
			
			\section{Hardly any symmetric state is in the GHZ class\label{app:GHZ} }
			In this appendix we prove the following Lemma.
			\begin{lemma}
				The set of symmetric states in the GHZ SLOCC class has Lebesgue measure zero in $\text{Sym}^n(\mathbb{C}^d)$ except when both $d=2$ and $n=3$.
			\end{lemma}
			\begin{proof}
				Obviously, $\rank_S(GHZ)\leq d$, so it suffices to prove that generically the symmetric tensor rank is larger than $d$. This follows from the Alexander-Hirschowitz theorem \cite{AlHi95}, which shows that for $n>2$ the set of states with symmetric tensor rank equal to $\lceil f_n(d)\rceil$ where
				\begin{equation}
				f_n(d)=\frac{1}{d}\left(
				\begin{array}{c}
				n+d-1 \\
				n \\
				\end{array}
				\right)=\frac{(d+n-1)\cdots(d+1)}{n!}
				\end{equation}
				has full Lebesgue measure in $\text{Sym}^n(\mathbb{C}^d)$ (with a few exceptional values of $(n,d)$ which require $\lceil f_n(d)\rceil+1$) \cite{comon}. Thus, to prove the claim it is enough to notice that $f_n(d)>d$ (except for $f_3(2)$), which can be easily verified by induction.
			\end{proof}
			
			\section{Weak isolation in non-ES classes \label{App:proof_iso_NONES}}
			
			We prove here Lemma \ref{lemma:non-es-isolation}, which we restate here in order to increase readability. \\

			\noindent {\bf Lemma \ref{lemma:non-es-isolation}.}
			\textit{Let $\ket{\Psi_s}$ be a non-ES $n$--qudit state ($n\geq5$). Then, there exist weakly isolated states within the SLOCC class of $\ket{\Psi_s}$.}

			\begin{proof}
				We will prove the lemma by constructing a state $g_1 \otimes \ldots \otimes g_n \ket{\Psi_s}$ and showing that it satisfies the condition given in Lemma \ref{lemma:isolation} (and is weakly isolated, thus). Let us denote the local dimension of $\ket{\Psi_s}$ by $d$.
				
				In order to construct the isolated state let us consider three orthonormal bases $\{\ket{v^{(1)}_j}\}_{j=0}^{d-1}$, $\{\ket{v^{(2)}_j}\}_{j=0}^{d-1}$, and $\{\ket{v^{(3)}_j}\}_{j=0}^{d-1}$, such that any two basis vectors belonging to two different bases have non-vanishing overlap, i.e., $\braket{v^{(i_1)}_{j_1}}{v^{(i_2)}_{j_2}} \neq 0$, for any $i_1,i_2 \in \{1,2,3\}$ s.t. $i_1 \neq i_2$ and for any $j_1, j_2 \in \{0,\ldots,d-1\}$. Such a construction is given e.g. by mutually unbiased basis, at least three of them exist for any $d$ \cite{Co06}. Let us moreover consider $d$ pairwise different positive reals $x_0, \ldots, x_{d-1}$. Then, an isolated state is given by $G_1= \sum_j x_j \ket{v^{(1)}_j} \bra{v^{(1)}_j}$ and $G_2= \sum_j x_j \ket{v^{(2)}_j} \bra{v^{(2)}_j}$, $G_3= \sum_j x_j \ket{v^{(3)}_j} \bra{v^{(3)}_j}$, $G_4=G_5 = \identity$, and $G_6, \ldots, G_n$ may be chosen arbitrarily.
				
				Recall that the symmetry group of $\ket{\Psi_s}$ consists of symmetries of the form $A^{\otimes n}$ only. Let us now prove that the state $g_1 \otimes \ldots \otimes g_n \ket{\Psi_s}$ is isolated by showing that there does not exist any non-trivial $A$ such that $A^\dagger G_i A \propto G_i$ for $n-1$ sites $i$. To this end, let us assume that $A^\dagger G_i A \propto G_i$ for at least four out of the five sites $\{1,2,3,4,5\}$ and show that it follows that $A \propto \identity$. As $A^\dagger G_i A \propto G_i$ for at least one $i \in \{4,5\}$, $A$ must be proportional to a unitary. Then, $A^\dagger G_i A \propto G_i$ if and only if $[A,G_i] = 0$. Note that $G_1$, $G_2$, and $G_3$ are constructed such that they possess non-degenerate eigenvalues. Commuting matrices, where (at least) one of the matrices has this property are simultaneously diagonalizable with a unique (up to irrelevant phases) eigenbasis.  We now have that $[A, G_i]=0$ for at least two out of the three sites $i\in\{1,2,3\}$. Let us assume wlog. that $[A, G_1]=0$ and $[A, G_2]=0$. Then, $A \ket{v^{(1)}_0} = \lambda \ket{v^{(1)}_0}$ for some $\lambda \in \mathbb{C}$. Moreover, we have that $A \ket{v^{(2)}_j} = \mu_j \ket{v^{(2)}_j}$ for some $\mu_j \in \mathbb{C}$ for all $j \in \{0,\ldots, d-1\}$. As $A$ is proportional to a unitary, all eigenvectors corresponding to distinct eigenvalues must be orthogonal. However, as $\braket{v^{(1)}_0}{v^{(2)}_j} \neq 0$ for all $j \in \{0,\ldots, d-1\}$ we have that $\mu_j = \lambda$ for all $j$. Hence, $A=\lambda \identity$. This completes the proof of the lemma.
			\end{proof}
			
			\section{Symmetries of non-derogatory ES states}
			\label{app:nonderogatorysyms}
			
			In this appendix, we characterize the symmetries of non-derogatory ES SLOCC classes, which are representated by states of the form given in Eq. (\ref{eq:sumek}). First, we analyse the symmetries of a single $\ket{E_k}$ (Theorem \ref{lemma:eksyms}). Then, we consider direct sums thereof (Theorem \ref{lemma:eksumsyms}). We restate Theorems \ref{lemma:eksyms} and \ref{lemma:eksumsyms} here for readability.\\

			\noindent {\bf Theorem \ref{lemma:eksyms}.}\textit{
				For $n\geq 3$ the stabilizer of the states $\ket{E_k}$ is generated by 
				operators of the form $B \otimes B^{-1}$ and $A^{\otimes n}$, where
				\begin{itemize}
					\item $B$ is an arbitrary invertible upper triangular Toeplitz matrix. 
					\item $A^{\otimes n}$ can be written up to a proportionality factor as $A=D\bar{S}$ where $D$ is a diagonal matrix with $[D]_{l,l} = x^l$ for some $x \in \mathbb{C}$ and $l \in \{0, 1, \ldots, k\}$ and $\bar{S}$ is upper triangular with $[\bar{S}]_{l,l}=1$ for all $l$. Furthermore, $\bar{S}$ is characterized by $[\bar{S}]_{i+1,l}=\sum_{j=1}^{l-i} [\bar{S}]_{i,l-j} y_j$ for some complex parameters $y_j$ and $\bra{E_l}\bar{S}^{\otimes n}\ket {E_k}=0$ for $0\leq l< k$. For $n\geq k-1$ the symmetries $\bar{S}^{\otimes n}$ can be determined by solving linear equations and (except for a measure zero subset) any choice of $y_j$ leads to a solution. Therefore, $\bar{S}$ is a $k$-parameter group (depending on $n$ and clearly on $k$).
			\end{itemize}}
			
			\begin{proof}
				
				Recall that, as already pointed out in the main text, any symmetry of a symmetric state can be written as a product of symmetries of the form $B_{(i)}\otimes B_{(j)}^{-1}\otimes \identity \ldots \identity$  and $A^{\otimes n}$. We prove the theorem by considering these two types of symmetries separately.

				We will first characterize all symmetries of the form $B\otimes B^{-1}\otimes \identity \ldots \identity$. For these symmetries it has to hold that 
				\bea\label{eq_A}
				B_{(1)} \ket{E_k}= B_{(2)}\ket{E_k}.
				\eea 
				Recall that any potential proportionality factor in Eq. (\theequation) must equal one, as discussed in Section \ref{subsec:Prelim_ESandNES}.
				Projecting Eq. (\ref{eq_A}) on some state with $k-m$ excitations ($m\leq k$) on all parties but party $1$ and $2$ we obtain that 
				\bea\label{eq_A2}
				B_{(1)} \ket{E_m}= B_{(2)}\ket{E_m}.
				\eea 
				Projecting Eq. (\ref{eq_A2}) on the the computational basis state $\ket{m+1}$ on party $2$ results in 
				\bea \label{eq_A4}
				\sum_{l=0}^m [B]_{m+1,l}\ket{m-l}=0.
				\eea
				This implies that $[B]_{m+1,l}=0$ for all $l\leq m$ and any $m < k$, i.e., $B$ is upper triangular.
				Next we project this equation on the computational basis state $\ket{m-l}$ with $l\leq m$ on party $2$. We obtain
				\bea\label{eq_A3}
				B\ket{l}= \sum_{x=0}^l [B]_{m-l,m-l+x}\ket{l-x}.
				\eea
				As this equation has to hold true for any $m\leq k$ (and  $l\leq m$) we have that $[B]_{m-l,m-l+x}=[B]_{0,x}$, i.e. $B$ is a Toeplitz matrix. 
				It can be easily seen for an upper triangular Toeplitz matrix Eq. (\ref{eq_A}) is fulfilled and therefore this condition is also sufficient.
				
				It remains to characterize the symmetries of the form $A^{\otimes n}$, i.e., to characterize $A$ such that
				\bea \label{eq_S}
				A^{\otimes n} \ket{E_k}= \lambda \ket{E_k}.
				\eea
				We will first show by induction, that for all $m,j \in \{0, 1, \ldots, k\}$, $[A]_{m, j} = 0$ whenever $m > j$, i.e., $A$ is upper triangular. 
				To this end, let us consider the projection on the computational basis state $\ket{m_1m_2}$ of the first two parties  with $m_1+m_2>k$ and use that $A$ is invertible to obtain
				\bea \label{eq_S2}
				\bra{m_1m_2}A^{\otimes 2}\ket{E_k}=0.
				\eea
				Projecting then the remaining parties on a state with $k-l$ excitations ($l\leq k$) results in
				\bea \label{eq_S3}
				\bra{m_1m_2}A^{\otimes 2}\ket{E_l}=\sum_{i_1 + i_2 = l}  [A]_{m_1, i_1} [A]_{m_2, i_2} =0.
				\eea
				
				We then consider this equation for $m_1 = m > \tilde{n}+1$ and $m_2 \in \{k-(\tilde{n}+1), \ldots, k\}$. Note that for these choices of $m_1$ and $m_2$ it indeed holds that $m_1+m_2>k$.
				
				Let us then assume that $[A]_{m, j} = 0$ whenever $m > j$ holds for columns $j$ up to $\tilde{n}$ and show that the statement then also holds for column $\tilde{n}+ 1$. We will show that the induction hypothesis together with assuming $[A]_{m, \tilde{n}+1} \neq 0$ for some $m>\tilde{n}+1$ leads to the fact that $A$ is not invertible, which is a contradiction, and thus shows that $[A]_{m, \tilde{n}+1} = 0$ for any $m>\tilde{n}+1$. 
				
				Considering Eq. (\ref{eq_S3}) with $l = \tilde{n}+1$ and using the induction hypothesis ($[A]_{m, j} = 0$ for $m>j$ and all $j\leq \tilde{n}$) we obtain the single equation $[A]_{m, \tilde{n}+1} [A]_{m_2, 0} = 0$ for all $m_2 \in \{k-(\tilde{n}+1), \ldots, k\}$, which implies $[A]_{m_2, 0} = 0$. Moreover, by considering Eq. (\ref{eq_S3}) for $l \in \{\tilde{n}+2, \ldots, k\}$ one obtains iteratively $[A]_{m_2, 0} = [A]_{m_2, 1} = \ldots = [A]_{m_2,k-(\tilde{n}+1)} =0$. In other words, in the $m_2$th row of $A$, only the last $\tilde{n}+1$ entries are non-vanishing. Note that this statement holds for $\tilde{n}+2$ rows. This is a contradiction, as $A$ has to be invertible, thus we have completed the induction step by showing  $[A]_{m, \tilde{n}+1} = 0$ for $m>\tilde{n}+1$. Note that the induction basis, $\tilde{n} = 0$, can be shown in a similar way.

				Hence, $A$ has to be upper triangular. Projecting then Eq. (\ref{eq_S}) on $\ket{m_1m_2\ldots m_n}$ with $\sum m_i=k$ we obtain that $\prod_i [A]_{m_i,m_i}=\lambda$. Considering also the projection on $\ket{(m_1+1) (m_2-1)\ldots m_n}$ and choosing $m_2=1$ we have that 
				\bea \label{eq_S5}
				\frac{[A]_{m_1+1,m_1+1}}{[A]_{m_1,m_1}}=\frac{[A]_{1,1}}{[A]_{0,0}}\equiv x
				\eea 
				for some $x\in \mathbb{C}$ and $m_1\in\{0,\ldots, k-1\}$. Normalizing $A$ such that $[A]_{0,0}=1$ this implies that $[A]_{m,m}=x^m$. Note that this implies that we can write $A=D\bar{S}$ with $D=\operatorname{diag}(1,x,x^2,\ldots,x^k)$ for some $x$ and $\bar{S}$ being upper triangular and $[\bar{S}]_{m,m}=1$ for all $m\leq k$. Note further that $D^{\otimes n}$ is a symmetry on its own. Therefore, $\bar{S}^{\otimes n}$ is a symmetry too and it hence only remains to characterize $\bar{S}$. 
				That is, we have that 
				\bea \label{eq_S6}
				\bar{S}^{\otimes n} \ket{E_k}= \ket{E_k},
				\eea
				where we further used that $[\bar{S}]_{m,m}=1$ and $\bar{S}$ is upper triangular in order to see that $\lambda=1$.
				We will next use that for any symmetry $\bar{S}^{\otimes n}$, we have that
				\bea\label{eq_SB}
				J\bar{S}=\bar{S}B
				\eea
				for some $B$ which is a symmetry of the form $B\otimes B^{-1}\otimes \identity^{\otimes n-2}$ and $J$ being here the ($k+1$)-dimensional Jordan block with eigenvalue $\lambda$. This is due to the fact that $(\bar{S}^{-1})^{\otimes n}(J\otimes J^{-1}\otimes \identity^{\otimes n-2})\bar{S}^{\otimes n}$ has to be a symmetry and that it is of the form $B\otimes B^{-1}\otimes \identity^{\otimes n-2}$. It can be straightforwardly seen that Eq. (\ref{eq_SB}) corresponds to
				\bea
				\lambda [\bar{S}]_{i,l}+[\bar{S}]_{i+1,l}=\sum_{j=0}^{l-i} [\bar{S}]_{i,l-j}[B]_{0,j}.
				\eea
				Note that $B$ has to have the same eigenvalue as $J_\lambda$, i.e., $[B]_{0,0}=\lambda$ and therefore
				\bea\label{sym_eq}
				[\bar{S}]_{i+1,l}=\sum_{j=1}^{l-i} [\bar{S}]_{i,l-j}[B]_{0,j}.
				\eea
				
				As we will show this constraint implies for an upper triangular $\bar{S}$ that 
				\bea \label{eq_S7}
				\bar{S}^{\otimes n} \ket{E_k}= \sum_{l\leq k} \alpha_l \ket{E_l},
				\eea 
				with $\alpha_l\in \mathbb{C}$. It then only remains to impose the constraints on $\bar{S}$ that lead to $\alpha_l=0$ for $l<k$. We will present a systematic way of doing so. In particular, we will show that for a generic choice of $[B]_{0,j}$ in Eq. (\ref{sym_eq}) there exists a valid symmetry.
				
				In order to first show  Eq. (\ref{eq_S7}) we will use that $J^{+}J^{-}=\identity - \ket{0}\bra{0}$ with $J^+=\sum_{j=0}^{k-1} \ket{j+1}\bra{j}$ and $J^-=\sum_{j=1}^{k} \ket{j-1}\bra{j}$ and that for $0<\sum_{j=1}^n m_j$ with $0\leq m_j\leq k$ there always exists at least one party $j$ for which $m_j\neq 0$ (which we label without loss of generality in the following by party $1$).
				With this and labeling for better readability here the parties on which the operators are acting on by subindices (as well as omitting all identity operators) we have that
				\begin{align}\nonumber
					&\bra{m_1m_2\ldots m_n} \bar{S}^{\otimes n}\ket{E_k}\\\nonumber
					&=\bra{m_1m_2\ldots m_n}(\identity_{(1)}- \ket{0}_1\bra{0})\bar{S}^{\otimes n}\ket{E_k}\\\nonumber
					&=\bra{m_1m_2\ldots m_n}J_{(1)}^+J_{(1)}^-\bar{S}^{\otimes n}\ket{E_k}\\\nonumber
					&=\bra{m_1m_2\ldots m_n}J_{(1)}^+\bar{S}^{\otimes n}B_{(1)}\ket{E_k}\\\nonumber
					&=\bra{m_1m_2\ldots m_n}J_{(1)}^+\bar{S}^{\otimes n}B_{(2)}\ket{E_k}\\\nonumber
					&=\bra{m_1m_2\ldots m_n}J_{(1)}^+\otimes J_{(2)}^-\bar{S}^{\otimes n}\ket{E_k}
					\\
					&=\bra{(m_1-1)(m_2+1)\ldots m_n}\bar{S}^{\otimes n}\ket{E_k},\label{eq_symm2}
				\end{align}
				where we used that $\bar{S}$ fulfills Eq. (\ref{sym_eq}) and therefore $J^-\bar{S}=\bar{S}B$ for $B$ being a (non-invertible) upper triangular Toeplitz matrix and $B_{(1)}\ket{E_k}=B_{(2)}\ket{E_k}$ by the argumentation above. Note that we chose here to increase the excitation number of party $2$ but we could have chosen here equivalently any other party. Note further that if $\sum_{j=1}^n m_j=0$ this implies that $m_j=0$ for all $j$ which corresponds to a single computational basis state. Hence, we have that any projection of Eq. (\ref{eq_S6}) solely depends on $\sum_{j=1}^n m_j$. We further have that due to the fact that $\bar{S}$ is upper triangular any projection on a computational basis state with $\sum_{j=1}^n m_j>k$ is equal to zero. This completes the proof that Eq. (\ref{eq_S7}) holds true.
				
				With this $\bar{S}$ (being an upper triangular matrix with only ones on the diagonal) is characterized by Eq. (\ref{sym_eq}) and $\bra{E_l}\bar{S}^{\otimes n}\ket{E_k}=0$ for $l<k$. We will next show that for $n\geq k-1$ and for nearly any choice of $[B]_{0,j}$ in Eq. (\ref{sym_eq}) there exists a solution for $\bar{S}^{\otimes n}$. In order to do so, we provide a construction that allows to successively solve this set of equations by solving only linear ones. Instead of projecting on $\bra{E_{k-l}}$, we consider the projection on a particular computational basis state, namely, for any fixed excitation number $\sum_{j=1}^n m_j=k-l$ with $k\geq l\geq 1$, the state for which $k-l$ parties hold the state $\ket{1}$ and the remaining $n-k+l$ parties have the state $\ket{0}$. Note that due to Eq. (\ref{eq_symm2}) these lead to the same equation as $\bra{E_{k-l}}\bar{S}^{\otimes n}\ket{E_k}=0$ if  Eq. (\ref{sym_eq}) is fulfilled.
				It can be easily seen that the resulting equation only depends on $[\bar{S}]_{1,m}$ and $[\bar{S}]_{0,m-1}$ with $m\in\{1,\ldots,l+1\}$ (and $[B]_{0,j}$). Using then Eq. (\ref{sym_eq}) one observes that this equation in fact only depends on $[\bar{S}]_{0,m-1}$ with $m\in\{1,\ldots,l+1\}$ (and $[B]_{0,j}$) and that it is linear in $[\bar{S}]_{0,l}$. Hence, solving the equations for $[\bar{S}]_{0,l}$ in increasing order of $l$ (starting with $l=1$) and inserting the result in the equation for $l+1$, one can determine $[\bar{S}]_{0,l}$ for $1\leq l\leq k$. With this all projections on computational basis states with $\sum_{j=1}^n m_j<k$ (and therefore also on $\ket{E_j}$ with $j<k$) are zero and $\bar{S}^{\otimes n}$ is indeed a symmetry.
				Note that in our construction we solved linear equations in $[\bar{S}]_{0,l}$ for $1\leq l\leq k$ (i.e., these equations have in general a solution except for very special values of $[B]_{0,j}$) and that as $\bar{S}$ is upper triangular and $[\bar{S}]_{m,m}=1$ for all $m$ it is always invertible. Hence, for any $[B]_{0,j}$ (up to a measure zero subset) there exists a solution and  $\bar{S}$ depends only on $n, k$ and $[B]_{0,j}$ for $j\in\{1, \ldots, k\}$. Note that if $\{[B]_{0,j}\}$ is such that in the $l$-th step of the protocol the equation is independent of the value of $[\bar{S}]_{0,l}$, then $[\bar{S}]_{0,l}$ can be chosen to be a free parameter and one can proceed to solve the equations iteratively as discussed above.
			\end{proof}
			
			As we have shown in the proof the following algorithm allows to determine the symmetries $\bar{S}^{\otimes n}$:
			\begin{enumerate}
				\item Start with $l=1$.
				\item Consider the projection of $\bar{S}^{\otimes n}\ket{E_k}= \ket{E_k}$ on the computational basis state for which $k-l$ parties hold the state $\ket{1}$ and the remaining $n-k-l$ parties have the state $\ket{0}$. 
				\item If $l\neq 1$ insert the expressions for $[\bar{S}]_{0,m}$, $m\in \{1,\ldots, l-1\}$ into this equation.
				\item Solve the resulting linear equation for $[\bar{S}]_{0,l}$. If the equation vanishes for any value of $[\bar{S}]_{0,l}$, then $[\bar{S}]_{0,l}$ is in the following a free parameter. If the equation does not vanish for any value of $[\bar{S}]_{0,l}$, then there does not exist a symmetry for this choice of $\{[B]_{0,j}\}$ and the algorithm terminates. Otherwise, replace $l\rightarrow l+1$ and repeat steps 2.-4. until $l+1>k$.
			\end{enumerate}

			In order to illustrate how one can easily determine the symmetry $\bar{S}^{\otimes n}$ with our algorithm, let us consider as an example the case $k=2$. From the projection of $\bar{S}^{\otimes n}\ket{E_k}=\ket{E_k}$ on $\ket{100\ldots 0}$ ($\sum_{j=1}^n m_j=k-1$) one obtains the equation 
			\begin{align}\nonumber
				&[\bar{S}]_{1,2}+(n-1)[\bar{S}]_{0,1}\\\label{eq_example}
				&=([B]_{0,2}+[B]_{0,1}[\bar{S}]_{0,1})+(n-1)[\bar{S}]_{0,1}=0,
			\end{align}
			where we used Eq. (\ref{sym_eq}).  Considering first $[B]_{0,1}\neq -(n-1)$ we have that $[\bar{S}]_{0,1}=- [B]_{0,2}/([B]_{0,1}+n-1)$.  We next consider the projection on $\ket{000\ldots 0}$ ($\sum_{j=1}^n m_j=k-2$) which results in the equation 
			\begin{align}\nonumber&n[\bar{S}]_{0,2}+\frac{n(n-1)}{2}([\bar{S}]_{0,1})^2\\
				&=n[\bar{S}]_{0,2}+\frac{n(n-1)([B]_{0,2})^2}{2([B]_{0,1}+n-1)^2}=0
			\end{align}
			and therefore for $[B]_{0,1}\neq -(n-1)$,
			\begin{align}\nonumber
				[\bar{S}]_{0,2}&=-\frac{(n-1)([B]_{0,2})^2}{2([B]_{0,1}+n-1)^2}
			\end{align}
			Note that if $[B]_{0,1}=-(n-1)$ then $\bar{S}$ can only be a symmetry if $[B]_{0,2}=0$ (see Eq. (\ref{eq_example})) and we further have that $[\bar{S}]_{0,1}$ is a free parameter.
			This illustrates how the procedure above allows to determine the symmetries.

			Note that in Theorem \ref{lemma:eksyms}, the symmetries are normalized in such a way that $A^{\otimes n} \ket{E_k} = x^{k} \ket{E_k}$, a fact which will be relevant when considering the symmetries of direct sums of $\ket{E_k}$ in the following Theorem. 
			We now characterize the symmetries of the general non--derogatory case. \\
			
			\noindent {{\bf Theorem \ref{lemma:eksumsyms}}{\bf.}}\textit{
				For $n\geq 3$ the symmetries of states $\bigoplus_{b=1}^K \ket{E_{k_b}}$ are generated by
				\begin{enumerate}
					\item $\bigoplus_{b=1}^K S_{k_b}$, with $S_{k_b}$ a symmetry of the individual $\ket{E_{k_b}}$, i.e. $S_{k_b}\ket{E_{k_b}}=\alpha \ket{E_{k_b}}$ for any $b\in \{1,\ldots K\}$ 
					\item diagonal matrices $D=\otimes_i D(\vec{\gamma}_i)$ with $D(\vec{\gamma}_i)= \bigoplus_{b=1}^K (\vec{\gamma}_i)_b \identity_{k_b}$ for some $\vec{\gamma}_i \in \mathbb{C}^K$ such that $\prod_i (\vec{\gamma}_i)_b=\gamma$ for all $b$.
					\item simultaneous permutations of whole blocks that have equal size for all parties (i.e. permutations of the form $X_\sigma^{\otimes n}$).
				\end{enumerate}
			}
			\begin{proof}
				Let us assume wlog that the contributions $\ket{E_{k_b}}$ are sorted in order of descending sizes, i.e., $k_1 \geq k_2 \geq \ldots \geq k_K$, where we denote the total number of blocks by $K$.
				Let us use the notation $\ket{i^{(b)}}$ for the single particle basis corresponding to block $b$. It can be easily verified that simultaneous permutations of blocks of equal size are indeed symmetries, that are matrices $X_\sigma^{\otimes n}$, where
				\begin{align}
					X_{\sigma} = \sum_{b \in \{1, \ldots, K\}} \sum_{i \in \{0, \ldots, k_{b}\}} \ket{i^{(\sigma(b))}} \bra{i^{(b)}},
				\end{align}
				for $\sigma \in \mathcal{S}_K$, such that $k_{\sigma(b)} = k_b$ for all $b \in \{1, \ldots, K\}$, where  $\mathcal{S}_K$ denotes the symmetric group. We will prove the Lemma by first showing that any symmetries of $\bigoplus_i \ket{E_{k_i}}$ can be brought into a block-diagonal through acting on them by some $X_\sigma^{\otimes n}$ from the left. Then we will show, that each block in the block-diagonal form must individually satisfy the constraints obtained in Theorem \ref{lemma:eksyms}.
				
				Let us first show, that the symmetries can be brought into a block-diagonal form with $K$ blocks of sizes $k_b+1$. Considering projections of  $\bigoplus_b \ket{E_{k_b}}$ onto $\ket{j_1^{(b_{1})}, j_2^{(b_2)}}$, we obtain, similarly as in the proof of Theorem \ref{lemma:eksyms}, the equations
				\begin{align}
					\label{eq:twoparticlegeneralized}
					\forall c \in &\{1, \ldots, K\}  \ \forall l \in \{0, \ldots, k_c\} \nonumber \\
					&\sum_{i_1 + i_2 = l}  [S_1]_{j^{(b_1)}_1, i^{(c)}_1} [S_2]_{j^{(b_2)}_2, i^{(c)}_2} = 0,
				\end{align}
				which symmetries $S_1 \otimes \ldots \otimes S_n$ necessarily have to satisfy for all $j_1\in \{0, \ldots, k_{b_1}\}$, $j_2\in \{0, \ldots, k_{b_2}\}$ in case $b_1 \neq b_2$, and for $j_1 + j_2 > k_b$ in case $b_1 = b_2 = b$. Note that in order to ease the notation, we denote in this proof the part of the symmetry that is acting on party $i$ via a subindex $i$, i.e., $S_i\equiv S^{(i)}$.
				
				Let us now show by induction, that $S_i$ can be brought into a block-diagonal form, as outlined above. Let us first show, that the statement holds for the first $k_1 +1$ columns of the matrices $S_i$, i.e., the columns that constitute the first block, which has size $k_1 + 1$.
				As $S_1$ has to be invertible, at least one element in the first column of $S_1$ must be non-vanishing. Assuming this element has a block-index of $b_1$, Eq. (\ref{eq:twoparticlegeneralized}) for $c=1$ can be used to show that for all other blocks $b \neq b_1$, for all $j_2 \in \{0, \ldots, k_{b}\}$, and for all $i_2 \in \{0, \ldots, k_1\}$ it must hold
				that  $[S_2]_{j_2^{(b)},i_2^{(1)}} = 0$. From this it follows that $S_2$ has $\sum_{b \in \{1, \ldots, K\} \setminus \{{b_1}\}} (k_b+1)$ rows where one can only find non-vanishing elements on the last $\sum_{b \in \{2, \ldots, K\} } (k_b+1)$ columns. Thus, unless $k_{b_1} = k_1$, $S_2$ cannot be invertible, which implies that non-vanishing elements in the first column of $S_1$ can only appear in rows which belong to blocks of the largest possible size, $k_1$. Moreover, in order for $S_2$ to be invertible, it must then also contain a non-vanishing element in the first column, in a row corresponding to the same block, $b_1$. Thus, the same argument can be used to show that in $S_1$, and actually in all $S_i$, the elements $[S_i]_{j_2^{(b)},i_2^{(1)}}$ vanish for all $b \neq b_1$, for all $j_2 \in \{0, \ldots, k_{b}\}$, and for all $i_2 \in \{0, \ldots, k_1\}$. In case the block $b_1$, that contains the non-vanishing entries, is not the first one, we can swap the positions of block $1$ and block $b_1$ by applying an appropriate $X_\sigma^{\otimes n}$ from the left. Thus, we have successfully brought $S_i$ into block-diagonal form for the first $k_1 +1$ columns.
				
				Let us now prove the induction step, i.e., assume that $S_i$ are in block diagonal form for the columns corresponding to the first $m$ blocks, and show that $S_i$ can be brought into block-diagonal form also on columns corresponding to the following block, block $m+1$. Again, some entry on column $1+\sum_{b \in \{1, \ldots, m\} } (k_b+1)$ of $S_1$ must be non-vanishing. Let us assume this element has a block-index of $b_{m+1}$. Again, considering Eq. (\ref{eq:twoparticlegeneralized}) (for $c=m+1$) leads to the fact that for all other blocks $b \neq b_{m+1}$, for all $j_2 \in \{1, \ldots, k_{b}\}$, and for all $i_2 \in \{0, \ldots, k_{m+1}\}$ it must hold
				that  $[S_2]_{j_2^{(b)},i_2^{(m+1)}} = 0$. Similarly as before, it can be shown that this must hold not only for $S_2$, but in fact for all $S_i$. Let us now argue that the block index of the non-vanishing element, $b_{m+1}$, must be one of the largest-sized blocks among the remaining blocks with indices $m+1, \ldots, K$. First, note that the block index $b_{m+1}$ cannot be one of the first blocks $1, \ldots m$, in other words, the non-vanishing element cannot be on the first $\sum_{b \in \{1, \ldots, m\} } (k_b+1)$ rows of $S_1$. The reason for that is that in this case there would be $2(k_{b_{m+1}}+1)$ columns of $S_i$ with non-vanishing elements only on the same $k_{b_{m+1}}+1$ rows, leading to a contradiction as $S_i$ must be invertible. Moreover, similarly as before, the block index $b_{m+1}$ cannot index a block that has a size $k_{b_{m+1}}$ that is smaller than $k_{m+1}$, as we find  $\sum_{b \in \{m+1, \ldots, K\} \setminus \{{b_{m+1}}\}} (k_b+1)$ rows of $S_i$ where one can only find non-vanishing elements on the last $\sum_{b \in \{m+2, \ldots, K\} } (k_b+1)$ columns, which implies that $S_i$ cannot be invertible when $k_{b_{m+1}} < k_{m+1}$. As the block sizes of block $m+1$ and $b_{m+1}$ coincide, we can find an appropriate block permutation operator $X_\sigma^{\otimes n}$ from the left in order to swap the two blocks $m+1$ and $b_{m+1}$. Thus, we have brought $S_i$ into a block-diagonal form for the columns corresponding to the first $m+1$ blocks.
				The following sketch illustrates the induction proof showing that $S_i$ can be brought into block-diagonal form. In the induction proof, we show that the ``gray-shaded areas" in Eq.\;(\ref{eq:prettypencil}) have to vanish. Moreover, it is sketched how a contradiction arises when the next non-vanishing element does not appear in a block that has largest size among the remaining blocks, as then the matrix $S_i$ cannot be invertible,
				\begin{widetext}
					\begin{equation}
					\label{eq:prettypencil}
					\begin{tikzpicture}[baseline=(current  bounding  box.center),
					style1/.style={
						matrix of math nodes,
						every node/.append style={text width=#1,align=center,minimum height=2ex},
						nodes in empty cells,
						left delimiter=(,
						right delimiter=),
					},
					]
					\matrix[style1=0.15cm] (1mat)
					{
						& & & & & & & & & & & & & & & & & & & & & & & & & & & & & & &  \\
						& & & & & & & & & & & & & & & & & & & & & & & & & & & & & & &  \\
						& & & & & & & & & & & & & & & & & & & & & & & & & & & & & & &  \\
						& & & & & & & & & & & & & & & & & & & & & & & & & & & & & & &  \\
						& & & & & & & & & & & & & & & & & & & & & & & & & & & & & & &  \\
						& & & & & & & & & & & & & & & & & & & & & & & & & & & & & & &  \\
						& & & & & & & & & & & & & & & & & & & & & & & & & & & & & & &  \\
						& & & & & & & & & & & & & & & & & & & & & & & & & & & & & & &  \\
						& & & & & & & & & & & & & & & & & & & & & & & & & & & & & & &  \\
						& & & & & & & & & & & & & & & & & & & & & & & & & & & & & & &  \\
						& & & & & & & & & & & & & & & & & & & & & & & & & & & & & & &  \\
						& & & & & & & & & & & & & & & & & & & & & & & & & & & & & & &  \\
						& & & & & & & & & & & & & & & & & & & & & & & & & & & & & & &  \\
						& & & & & & & & & & & & & & & & & & & & & & & & & & & & & & &  \\
						& & & & & & & & & & & & & & & & & & & & & & & & & & & & & & &  \\
						& & & & & & & & & & & & & & & & & & & & & & & & & & & & & & &  \\
						& & & & & & & & & & & & & & & & & & & & & & & & & & & & & & &  \\
						& & & & & & & & & & & & & & & & & & & & & & & & & & & & & & &  \\
						& & & & & & & & & & & & & & & & & & & & & & & & & & & & & & &  \\
						& & & & & & & & & & & & & & & & & & & & & & & & & & & & & & &  \\
						& & & & & & & & & & & & & & & & & & & & & & & & & & & & & & &  \\
						& & & & & & & & & & & & & & & & & & & & & & & & & & & & & & &  \\
						& & & & & & & & & & & & & & & & & & & & & & & & & & & & & & &  \\
						& & & & & & & & & & & & & & & & & & & & & & & & & & & & & & &  \\
						& & & & & & & & & & & & & & & & & & & & & & & & & & & & & & &  \\
					};
					
					\draw[solid]
					(1mat-1-1.north west) -- (1mat-1-4.north east);
					\draw[solid]
					(1mat-4-1.south west) -- (1mat-4-4.south east);
					\draw[solid]
					(1mat-1-1.north west) -- (1mat-4-1.south west);
					\draw[solid]
					(1mat-1-4.north east) -- (1mat-4-4.south east);

					\draw[solid]
					(1mat-5-5.north west) -- (1mat-5-8.north east);
					\draw[solid]
					(1mat-8-5.south west) -- (1mat-8-8.south east);
					\draw[solid]
					(1mat-5-5.north west) -- (1mat-8-5.south west);
					\draw[solid]
					(1mat-5-8.north east) -- (1mat-8-8.south east);

					\draw[solid]
					(1mat-13-9.north west) -- (1mat-13-11.north east);
					\draw[solid]
					(1mat-15-9.south west) -- (1mat-15-11.south east);
					\draw[solid]
					(1mat-13-9.north west) -- (1mat-15-9.south west);
					\draw[solid]
					(1mat-13-11.north east) -- (1mat-15-11.south east);

					\draw[dashed]
					(1mat-1-4.north east) -- (1mat-25-4.south east);
					\draw[dashed]
					(1mat-1-8.north east) -- (1mat-25-8.south east);
					\draw[dashed]
					(1mat-1-12.north east) -- (1mat-25-12.south east);
					\draw[dashed]
					(1mat-1-15.north east) -- (1mat-25-15.south east);

					\draw[dashed]
					(1mat-4-1.south west) -- (1mat-4-32.south east);
					\draw[dashed]
					(1mat-8-1.south west) -- (1mat-8-32.south east);
					\draw[dashed]
					(1mat-12-1.south west) -- (1mat-12-32.south east);
					\draw[dashed]
					(1mat-15-1.south west) -- (1mat-15-32.south east);

					\fill [black ,opacity=0.1] (1mat-5-1.north west) rectangle (1mat-25-4.south east);
					\fill [black ,opacity=0.1] (1mat-1-5.north west) rectangle (1mat-4-8.south east);
					\fill [black ,opacity=0.1] (1mat-9-5.north west) rectangle (1mat-25-8.south east);
					\fill [black ,opacity=0.15] (1mat-1-9.north west) rectangle (1mat-12-12.south east);
					\fill [black ,opacity=0.15] (1mat-16-9.north west) rectangle (1mat-25-12.south east);

					\node[font=\large]
					at ([xshift=0pt,yshift=-4pt]1mat-2-2.east) {block 1};
					\node[font=\large]
					at ([xshift=0pt,yshift=-4pt]1mat-6-6.east) {block 2};

					\draw[stealth-]
					(1mat-14-10.south east) -- (1mat-20-20.south east);
					\node[font=\small]
					at ([xshift=4pt,yshift=-4pt]1mat-21-23.east) {(incorrectly placed) block 3};

					\node[font=\huge]
					at (1mat-17-18) {$\ddots$};

					\draw[decoration={brace,mirror,raise=5pt},decorate]
					(1mat-25-1.south west) --
					node[below=7pt] {$k_{1} + 1$}
					(1mat-25-4.south east);
					\draw[decoration={brace,mirror,raise=5pt},decorate]
					(1mat-25-5.south west) --
					node[below=7pt] {$k_{2} + 1$}
					(1mat-25-8.south east);
					\draw[decoration={brace,mirror,raise=5pt},decorate]
					(1mat-25-9.south west) --
					node[below=7pt] {$k_{3} + 1$}
					(1mat-25-12.south east);
					\draw[decoration={brace,mirror,raise=5pt},decorate]
					(1mat-25-13.south west) --
					node[below=7pt] {$k_{4} + 1$}
					(1mat-25-15.south east);
					
					\node at ([xshift=-30pt,yshift=-1.2pt]1mat.west) {$S_i =$};
					\end{tikzpicture}.
					\end{equation}
				\end{widetext}
				This completes the proof that the symmetries can be brought into a block-diagonal form with $K$ blocks of sizes $k_1+1, \ldots, k_K+1$.
				
				With this the equation  $S_1 \otimes  \ldots \otimes S_n (\bigoplus_{b=1}^K \ket{E_{k_b}}) \propto \bigoplus_{b=1}^K \ket{E_{k_b}}$ now separates into $K$ individual equations $S_1^{[b]} \otimes  \ldots \otimes S_n^{[b]}  \ket{E_{k_b}} = \alpha  \ket{E_{k_b}}$ for some $\alpha \in \mathbb{C}$. Thus, it is clear that the individual blocks have to satisfy the constraints derived in Theorem \ref{lemma:eksyms}. In particular, the diagonal entries of $S_i^{[b]}$ must be proportional to powers of some complex number $x_b$, i.e.,  $[S_i^{[b]}]_{l,l} = \lambda_{i,b} x_b^l$ for all $l \in \{0, \ldots, k_b\}$ for all parties $i$. In order to get the same $\alpha$ for all $b$, and keeping the normalization of $S^{[b]}$ assumed in Theorem \ref{lemma:eksyms}, we need to introduce a factor $\alpha /(x_b^{k_b} \prod_i \lambda_{i,b})$ for block $b$.
				
				Moreover, note, that besides the ``block-global'' scaling factor discussed above, scaling factors of the individual $S_i^{[b]}$ did not play a role in Theorem \ref{lemma:eksyms}, as in the case of a single block these scaling factors are only a global factor, which does not play a role. Here, however, each block $b$ at party $i$ can have a different scaling factor $(\vec{\gamma}_i)_b$, as long as taking the product of this factor over all parties for a fixed $b$ yields the same value for all $b$. Thus additional symmetries somehow resembling the diagonal symmetries of GHZ states arise, which are of the form
				\begin{align}
					D(\vec{\gamma}_1) \otimes \ldots \otimes D(\vec{\gamma}_n),
				\end{align}
				where
				\begin{align}
					D(\vec{\gamma}_i) = \bigoplus_{b=1}^K (\vec{\gamma}_i)_b \identity_{k_b},
				\end{align}
				for some $\vec{\gamma}_i \in \mathbb{C}^K$ for $i \in \{1, \ldots, n-1\}$, and $(\vec{\gamma}_n)_b = \frac{\gamma}{(\vec{\gamma}_1)_b \ldots (\vec{\gamma}_{n-1})_b }$ for all $b \in \{1, \ldots, K\}$ and some $\gamma \in \mathbb{C}$.
				This completes the proof of the lemma.
			\end{proof}

			\section{Proof of Lemma \ref{lemma:commutationlemma}}
			\label{app:commutationlemma}
			
			In this appendix we prove Lemma  \ref{lemma:commutationlemma} from the main text, which we restate here for readability.
			
			\noindent {{\bf Lemma \ref{lemma:commutationlemma}}{\bf.}}\textit{
				Let $k \geq 1$. Given an $(k+1) \times (k+1)$ upper triangular matrix A for which $A_{l,l}$ = $x^l$ for some $x\in \mathbb{C}$, and  an $(k+1) \times (k+1)$ matrix
				\begin{align}
					g =\begin{pmatrix}
						1&   &   &     \\
						& \ddots &   &     \\
						&   &  1  & a   \\
						&   &   & 1   \\
					\end{pmatrix},
				\end{align}
				where $a \in \mathbb{C}$. Then it holds that $A^\dagger g^\dagger g A \propto g^\dagger g$ iff $|x|=1$ and A is of the form
				\begin{align}
					\label{eq:Aformapp}
					A \propto \begin{pmatrix}
						1&   &   & &     \\
						& x  &   &  &   \\
						&  & \ddots &   &     \\
						&  &   &  x^{k-1}  & a  \left( \frac{1}{{x^*}^{k-1}} - x^k  \right) \\
						&  &   &   & x^k   \\
					\end{pmatrix}.
				\end{align}
			}
			
			\begin{proof}
				First, note that $G = g^\dagger g = \identity_{k-1} \oplus \tilde{G}$ with $\tilde{G} =\begin{pmatrix}
				1&  a   \\
				a^* & 1 + |a|^2 \end{pmatrix}$.
				Let us now rewrite  the equation $A^\dagger g^\dagger g A \propto g^\dagger g$ taking the block structure into account, i.e.,
				\begin{widetext}
					\begin{equation}
					\begin{tikzpicture}[baseline=(current  bounding  box.center),
					style1/.style={
						matrix of math nodes,
						every node/.append style={text width=#1,align=center,minimum height=3.0ex},
						nodes in empty cells,
						left delimiter=(,
						right delimiter=),
					},
					]
					\matrix[style1=0.25cm] (1mat)
					{
						& & & &\\
						& & & &\\
						& & & &\\
						& & & &\\
						& & & &\\
					};
					
					\draw[solid]
					(1mat-1-3.north east) -- (1mat-5-3.south east);
					\draw[solid]
					(1mat-3-1.south west) -- (1mat-3-5.south east);
					
					\node[font=\large]
					at (1mat-2-2) {$\tilde{A}_{0,0}^\dagger$};
					\node[font=\huge]
					at (1mat-2-5) {$0$};
					\node[font=\large]
					at (1mat-5-5) {$\tilde{A}_{1,1}^\dagger$};
					\node[font=\large]
					at (1mat-5-2) {$\tilde{A}_{0,1}^\dagger$};
					\end{tikzpicture}
					\begin{tikzpicture}[baseline=(current  bounding  box.center),
					style1/.style={
						matrix of math nodes,
						every node/.append style={text width=#1,align=center,minimum height=3.0ex},
						nodes in empty cells,
						left delimiter=(,
						right delimiter=),
					},
					]
					\matrix[style1=0.25cm] (1mat)
					{
						& & & &\\
						& & & &\\
						& & & &\\
						& & & &\\
						& & & &\\
					};
					\draw[solid]
					(1mat-1-3.north east) -- (1mat-5-3.south east);
					\draw[solid]
					(1mat-3-1.south west) -- (1mat-3-5.south east);
					
					\node[font=\huge]
					at (1mat-2-2) {$\identity$};
					\node[font=\huge]
					at (1mat-2-5) {$0$};
					\node[font=\large]
					at (1mat-5-5) {$\tilde{G}$};
					\node[font=\huge]
					at (1mat-5-2) {$0$};
					\end{tikzpicture}
					\begin{tikzpicture}[baseline=(current  bounding  box.center),
					style1/.style={
						matrix of math nodes,
						every node/.append style={text width=#1,align=center,minimum height=3.0ex},
						nodes in empty cells,
						left delimiter=(,
						right delimiter=),
					},
					]
					\matrix[style1=0.25cm] (1mat)
					{
						& & & &\\
						& & & &\\
						& & & &\\
						& & & &\\
						& & & &\\
					};
					
					\draw[solid]
					(1mat-1-3.north east) -- (1mat-5-3.south east);
					\draw[solid]
					(1mat-3-1.south west) -- (1mat-3-5.south east);
					
					\node[font=\large]
					at (1mat-2-2) {$\tilde{A}_{0,0}$};
					\node[font=\large]
					at (1mat-2-5) {$\tilde{A}_{0,1}$};
					\node[font=\large]
					at (1mat-5-5) {$\tilde{A}_{1,1}$};
					\node[font=\huge]
					at (1mat-5-2) {$0$};
					\end{tikzpicture}
					=
					\begin{tikzpicture}[baseline=(current  bounding  box.center),
					style1/.style={
						matrix of math nodes,
						every node/.append style={text width=#1,align=center,minimum height=3.0ex},
						nodes in empty cells,
						left delimiter=(,
						right delimiter=),
					},
					]
					\matrix[style1=0.25cm] (1mat)
					{
						& & & &&&&\\
						& & & &&&&\\
						& & & &&&&\\
						& & & &&&&\\
						& & & &&&&\\
					};
					
					\draw[solid]
					(1mat-1-4.north east) -- (1mat-5-4.south east);
					\draw[solid]
					(1mat-3-1.south west) -- (1mat-3-8.south east);
					
					\node[font=\large]
					at (1mat-2-2) {$\tilde{A}_{0,0}^\dagger \tilde{A}_{0,0}$};
					\node[font=\large]
					at (1mat-2-7) {$\tilde{A}_{0,0}^\dagger\tilde{A}_{0,1}$};
					\node[font=\large]
					at (1mat-5-7) {$\tilde{A}_{1,1}^\dagger \tilde{G} \tilde{A}_{1,1}$};
					\node[font=\large]
					at (1mat-5-2) {$\tilde{A}_{0,1}^\dagger\tilde{A}_{0,0}$};
					\end{tikzpicture},
					\end{equation} as
					\begin{equation}
					\begin{tikzpicture}[baseline=(current  bounding  box.center),
					style1/.style={
						matrix of math nodes,
						every node/.append style={text width=#1,align=center,minimum height=3.0ex},
						nodes in empty cells,
						left delimiter=(,
						right delimiter=),
					},
					]
					\matrix[style1=0.25cm] (1mat)
					{
						& & & &&&&\\
						& & & &&&&\\
						& & & &&&&\\
						& & & &&&&\\
						& & & &&&&\\
					};
					
					\draw[solid]
					(1mat-1-4.north east) -- (1mat-5-4.south east);
					\draw[solid]
					(1mat-3-1.south west) -- (1mat-3-8.south east);
					
					\node[font=\large]
					at (1mat-2-2) {$\tilde{A}_{0,0}^\dagger \tilde{A}_{0,0}$};
					\node[font=\large]
					at (1mat-2-7) {$\tilde{A}_{0,0}^\dagger\tilde{A}_{0,1}$};
					\node[font=\large]
					at (1mat-5-7) {$\tilde{A}_{1,1}^\dagger \tilde{G} \tilde{A}_{1,1}$};
					\node[font=\large]
					at (1mat-5-2) {$\tilde{A}_{0,1}^\dagger\tilde{A}_{0,0}$};
					\end{tikzpicture}
					\propto
					\begin{tikzpicture}[baseline=(current  bounding  box.center),
					style1/.style={
						matrix of math nodes,
						every node/.append style={text width=#1,align=center,minimum height=3.0ex},
						nodes in empty cells,
						left delimiter=(,
						right delimiter=),
					},
					]
					\matrix[style1=0.25cm] (1mat)
					{
						& & & &\\
						& & & &\\
						& & & &\\
						& & & &\\
						& & & &\\
					};
					\draw[solid]
					(1mat-1-3.north east) -- (1mat-5-3.south east);
					\draw[solid]
					(1mat-3-1.south west) -- (1mat-3-5.south east);
					
					\node[font=\huge]
					at (1mat-2-2) {$\identity$};
					\node[font=\huge]
					at (1mat-2-5) {$0$};
					\node[font=\large]
					at (1mat-5-5) {$\tilde{G}$};
					\node[font=\huge]
					at (1mat-5-2) {$0$};
					\end{tikzpicture} \text{, where } A =
					\begin{tikzpicture}[baseline=(current  bounding  box.center),
					style1/.style={
						matrix of math nodes,
						every node/.append style={text width=#1,align=center,minimum height=3.0ex},
						nodes in empty cells,
						left delimiter=(,
						right delimiter=),
					},
					]
					\matrix[style1=0.25cm] (1mat)
					{
						& & & &\\
						& & & &\\
						& & & &\\
						& & & &\\
						& & & &\\
					};
					
					\draw[solid]
					(1mat-1-3.north east) -- (1mat-5-3.south east);
					\draw[solid]
					(1mat-3-1.south west) -- (1mat-3-5.south east);
					
					\node[font=\large]
					at (1mat-2-2) {$\tilde{A}_{0,0}$};
					\node[font=\large]
					at (1mat-2-5) {$\tilde{A}_{0,1}$};
					\node[font=\large]
					at (1mat-5-5) {$\tilde{A}_{1,1}$};
					\node[font=\huge]
					at (1mat-5-2) {$0$};
					\end{tikzpicture}.
					\end{equation}
				\end{widetext}
				From the upper left part of Eq.\;(\theequation) we get that $\tilde{A}_{00}$, which is an upper triangular matrix, must be proportional to a unitary matrix. This implies that $\tilde{A}_{00}$ must be diagonal. From the lower right part we obtain
				$\tilde{A}_{1,1} \propto \begin{pmatrix}
				x^{k-1}&   a  \left( \frac{1}{{x^*}^{k-1}} - x^k  \right)   \\
				0& x^k
				\end{pmatrix}$ and $|x|=1$.
				Moreover, as $\tilde{A}_{0,0}$ is invertible, the upper right part of Eq. (\theequation) implies $\tilde{A}_{0,1} = 0$. This proves the ``only if'' part of the Lemma. However, all steps in the proof are in fact equivalences, which proves the lemma.
			\end{proof}

			\section{Weak isolation in non--derogatory ES classes} 
			
			\label{App:proof_iso_NDES}
			
			We prove here Theorem \ref{theo:sumsisolation}, which we restate here to increase readability. \\
			
			\noindent \textit{{\bf Theorem \ref{theo:sumsisolation}.}
				In the SLOCC classes represented by $\bigoplus_{b=1}^K \ket{E_{k_b}}$ with $K\geq 2$ and $k_b\neq 0$ for at least one $b\in\{1,\ldots, K\}$, there are weakly isolated states present. If further $k_b\neq 2$ for all $b\in\{1,\ldots, K\}$, then there exist symmetric weakly isolated states.}

			\begin{proof}
				Similarly as in the proof of  Theorem \ref{theo:isolation}, we will prove the theorem by constructing a family of states $g_1 \otimes \ldots \otimes g_n \left(\bigoplus_{b=1}^K \ket{E_{k_b}}\right)$ with the property that for all symmetries  $S_1 \otimes S_2 \otimes \ldots \otimes S_n$ of the state $\bigoplus_{b=1}^K \ket{E_{k_b}}$, it holds that either $S_i^\dagger g_i^\dagger g_i S_i \propto g_i^\dagger g_i$ is true for no more than $n-2$ parties, or $S\propto\identity^{\otimes n}$. This implies that the states $g_1 \otimes \ldots \otimes g_n \left(\bigoplus_{b=1}^K \ket{E_{k_b}}\right)$  are neither reachable via $LOCC_{\mathbb{N}}$, nor convertible via $LOCC_1$.
				Note that in order to ease the notation, we denote in this proof the part of the symmetry that is acting on party $i$ via a subindex $i$, i.e., $S_i\equiv S^{(i)}$.
				
				States of the form $\bigoplus_{b=1}^K \ket{E_{k_b}}$ for which the block sizes and their number of occurrences coincides and which differ therefore only by the order of the blocks are SLOCC equivalent and we choose as representative in the following the state in which the blocks are ordered with increasing size, i.e., $k_b\geq k_{b'}$ for $b> b'$.
				We will construct $G_i$ based on direct sums of $G_i^{[b]}$ from the proof of Theorem  \ref{theo:isolation} (for $k_b\geq 2$) and Observation \ref{obs_GHZd} (for more than two blocks with $k_b=0$) respectively, but with some additional non-vanishing elements, that appear at the junction of two blocks.  As will be clear later on, these off-diagonal elements ensure that $G_i^{[b]}$ cannot quasi-commute with the symmetries in the case of qubit-GHZ and W blocks and will impose additional constraints on the diagonal gate in Theorem \ref{lemma:eksumsyms}. More precisely, let us construct
				\begin{align}
					\label{eq:gconstruction}
					G_i = \bigoplus_{b=1}^K G_i^{[b]} + \sum_{b=2}^{K} (c_b \ket{0^{(b-1)}}\bra{0^{(b)}}+ h.c.),
				\end{align}
				where for $b \in \{2, \ldots, K\}$, $c_b \in \mathbb{C}$ and $c_b \neq 0$ and small enough such that $G_i>0$. Moreover, for any $b \in \{1, \ldots, K\}$ for which $k_b\geq 2$ we choose  $g_i^{[b]}=\sqrt{G_i^{[b]}}$ to be of the form given in Eq. (\ref{eq_G}) with $a$ being replaced by a free parameter $a_i$, which is different for different blocks, we will thus denote it by $a_{i,b}$. Let us further choose $a_{i,b}\neq 0$ with pairwise different absolute value for all $b$ for which $k_b\geq 2$ and in case $k_b = 2$, moreover $a_{i,b}$ should be pairwise different for all $i \in \{1, \ldots, n\}$. For any $b \in \{1, \ldots, K\}$ for which $k_b=1$ we choose 
				\begin{align}\label{eq_GW}
					G_i^{[b]} = \begin{pmatrix}
						1& p \\
						p & d_b     \\
					\end{pmatrix},
				\end{align}
				with $G_i^{[b]}>0$ and $d_b\neq d_{b'}$ for $b\neq b'$. For any $b \in \{1, \ldots, K\}$ for which $k_b=0$ we choose $G_i^{[b]}=1$. Moreover, for blocks $b$ with $k_b=0$ we choose the corresponding coefficients  $c_b$  such that all absolute values are pairwise different (as in Observation \ref{obs_GHZd}), i.e. $|c_b|\neq |c_{b'}|$ for $b, b'$ such that $k_b=k_{b'}=0$.
				
				Let us now, based on study of symmetries of $\bigoplus_b \ket{E_{k_b}}$ in Theorem \ref{lemma:eksumsyms}, show that the condition $S_i^\dagger g_i^\dagger g_i S_i \propto g_i^\dagger g_i$ for $i \in \{1, \ldots, n-1\}$ implies that $S \propto \identity$. Recall that the symmetries are generated by the symmetries of single blocks, certain scalings of the individual blocks and permutations of blocks of equal size. More precisely, any symmetry can be written as
				\begin{align}
					\label{eq:symdecomp}
					S= \left( \bigoplus_{b=1}^K S^{[b]} \right)   \left( D(\vec{\gamma}_1) \otimes \ldots \otimes D(\vec{\gamma}_n) \right) \left( X_\sigma^{\otimes n} \right)
				\end{align}
				for some $\sigma$ permuting blocks of equal size, $\vec{\gamma}_1, \ldots, \vec{\gamma}_{n-1} \in \mathbb{C}^K$ \cite{footnoteAppD}, and symmetries  $S^{[b]}$ of $\ket{E_{k_b}}$.
				
				In the following we will consider the diagonal and off-diagonal blocks of the quasi-commutation relation, i.e. we partition our matrices into blocks,  $M=(M_{b_1b_2})_{b_1b_2}$, where the block structure is determined by the block structure of $\bigoplus_{b=1}^K \ket{E_{k_b}}$. With this we obtain that the quasi-commutation relation is fulfilled only if 
				\begin{align}\label{eq_blockdiag}
					\left| \left(\vec{\gamma}_i\right)_b  \right|^2  \left(S_i^{[b]}\right)^\dagger G_i^{[b]}  S_i^{[b]}   \propto  G_i^{[\sigma(b)]}
				\end{align}
				and
				\begin{align}\label{eq_blockoffdiag}
					&c_b \left(S_i^{[b-1]}\right)^\dagger\ket{0^{(b-1)}}\bra{0^{(b)}}S_i^{[b]}\left(\vec{\gamma}_i\right)_{b-1}^\ast\left(\vec{\gamma}_i\right)_b\propto \nonumber\\ & \qquad c_{\sigma(b)}\ket{0^{(b-1)}}\bra{0^{(b)}}
				\end{align}
				and the permutation $X_\sigma$ is such that it maps off-diagonal blocks $(G_i)_{b_1b_2}$ that only contain zeros solely to blocks $(G_i)_{b_1'b_2'}$ that are the zero matrix (of same size as $(G_i)_{b_1b_2}$).

				Let us first note that if a block $b$ with $k_b=1$  occurs, we obtain from Eq. (\ref{eq_blockoffdiag}) that $S_i^{[b]}$, which in general can be an arbitrary upper triangular matrix, has to be diagonal in order for the quasi-commutation relation to possibly hold true. Using then that $G_i^{[b]}$ is of the form in  Eq. (\ref{eq_GW}) it is straightforward to see that from Eq. (\ref{eq_blockdiag}) it follows that $S_i^{[b]}$ has to be proportional to the identity and that there cannot be any non-trivial permutation among blocks with $k_b=1$, i.e. $\sigma(b)=b$ for blocks with $k_b=1$.

				Considering Eq. (\ref{eq_blockdiag}) for blocks with $k_b\geq 2$ we observe that these equations are of similar form in Lemma \ref{lemma:commutationlemma}. Here, in contrast to there, the operator $\tilde{G}_i^{[b]}$ on the left-hand side may differ from the one on the right-hand side of Eq. (\ref{eq_blockdiag}).
				Nevertheless, similar techniques as in Lemma \ref{lemma:commutationlemma} can be used to see that $S_i^{[b]}$ must have the form given in Eq. (\ref{eq:Aform}) with $|x_b|=1$, with the only difference that the $(k_b-1,k_b)$ matrix element of $S_i^{[b]}$ now reads $[S_i^{[b])}]_{k_b-1,k_b} =\frac{a_{i,\sigma(b)}}{(x_b^*)^{k-1}} - a_{i,b} (x_b)^k$. The same techniques as in the proof of Theorem \ref{theo:isolation} can be used to prove that the $S_i^{[b]}$ must be diagonal. Whereas there, this condition implied, and in fact was satisfied by, the choice $x=1$, here we obtain the conditions
				\begin{align}
					a_{i,\sigma(b)} - a_{i,b} {x_b} = 0
				\end{align}
				for all $i \in \{1, \ldots, n\}$ and $b \in \{1, \ldots, K\}$. As for all $b$, $a_{i,b}$ have pairwise different absolute value, Eq. (\theequation) cannot be satisfied unless $\sigma(b) = b$ for all $b$, i.e., the symmetry $S$ in Eq. (\ref{eq:symdecomp}) can only contain trivial block permutations for blocks with $k_b\geq 2$. Moreover, it follows that $x_b=1$ for all $b$, which directly implies $S^{[b]} \propto \identity$ in Eq. (\ref{eq:symdecomp}) for $k_b\geq 2$. 
				
				If there are more than two blocks with $k_b=0$, one obtains analogously to the argumentation in Observation \ref{obs_GHZd} that there can be no permutations of blocks with $k_b=0$ (and that for b with $k_b=0$, $(\vec{\gamma}_i)_b$ is independent of b).
				
				Let us next consider the case that there are exactly two blocks with $k_b=0$ and at least one further block with $k_b>0$, which we label without loss of generality in the following by $b=1, 2$ (for $k_b=0$) and $b=3$ (for a block with $k_b>0$). 
				As we will show there cannot be any non-trivial permutation of the two blocks with $k_b=0$. Note first that as we have already shown for our choice of $G_i$ no non-trivial permutations among blocks with $k_b>0$ are possible (even if we are in the case that there are several of the same size). Hence, it is only possible to permute the two blocks with $k_b=0$. However, this would map $c_{3} \ket{0^{(2)}}\bra{0^{(3)}}$ to $c_{3} \ket{0^{(1)}}\bra{0^{(3)}}$. As $(G_i)_{13}$ is the zero matrix and $c_3\neq 0$ no such permutation is possible. 
				
				Hence, for any block structure we have reduced the form of $S$ to $S = D(\vec{\gamma}_1) \otimes \ldots \otimes D(\vec{\gamma}_n)$. It remains to be shown that
				\begin{align}
					D(\vec{\gamma}_i)^\dagger  G_i D(\vec{\gamma}_i) \propto  G_i
				\end{align}
				for $i \in \{1, \ldots, n-1\}$ implies that $S \propto \identity$. From the diagonal of Eq. (\theequation) we still get some restriction, namely $|(\vec{\gamma}_i)_b|$ is independent of b for all $i \in \{1, \ldots, n-1\}$. Considering now the matrix elements  $\ket{0^{(b-1)}}\bra{0^{(b)}}$ of Eq. (\theequation) for all $b \in \{2, \ldots, K\}$, we finally get $(\vec{\gamma}_i)_{b-1} = (\vec{\gamma}_i)_b$ for all $b \in \{2, \ldots, K\}$ and for all $i \in \{1, \ldots, n-1\}$, which implies $S \propto \identity$. Similarly as in Theorem \ref{theo:isolation} and Observation \ref{obs_GHZd}, the states considered here constitute examples of symmetric states that are weakly isolated, as long as there is no $\ket{E_2}$ contribution. This proves the theorem.
			\end{proof}
			
			Let us remark, that in the construction of isolated states in Theorem \ref{theo:sumsisolation}, it seems that populating some elements in $g_i$ that lie outside the block-diagonal structure is necessary, as otherwise unitary $D(\vec{\gamma}_i)$-symmetries remain and presumably, transformations become possible.

			\section{Symmetries of $\ket{\psi_{\text{derogatory}}}$}
			\label{sec:symmetries}
			
			In this appendix, we characterize the symmetries of the state 
			\begin{align}
				\ket{\psi_{\text{derogatroy}}} = \ket{E_0^{5,0}} + \ket{E_0^{3,2}} + \ket{E_0^{2,3}}  + \ket{E_1^{0,5}}, \nonumber
			\end{align}
			introduced in Section \ref{sec:derog} within the main text.
			
			Let us first recall the individual contributions,
			\begin{align}
				\ket{E_0^{5,0}} &= \ket{0 0 0 0 0} \\
				\ket{E_0^{3,2}} &= \ket{0 0 0 1 1} + \ket{0 0 1 0 1} + \ket{0 1 0 0 1} + \ket{1 0 0 0 1} \nonumber \\
				& \quad + \ket{0 0 1 1 0} + \ket{0 1 0 1 0} + \ket{1 0 0 1 0} + \ket{0 1 1 0 0} \nonumber \\
				& \quad  + \ket{1 0 1 0 0} + \ket{1 1 0 0 0} \\
				\ket{E_0^{2,3}} &= \ket{0 0 1 1 1} + \ket{0 1 0 1 1} + \ket{0 1 1 0 1} + \ket{0 1 1 1 0} \nonumber \\
				& \quad + \ket{1 0 0 1 1} + \ket{1 0 1 0 1} + \ket{1 0 1 1 0} + \ket{1 1 0 0 1} \nonumber \\
				& \quad  + \ket{1 1 0 1 0} + \ket{1 1 1 0 0} \\
				\ket{E_1^{0,5}} &= \ket{1 1 1 1 2} + \ket{1 1 1 2 1} + \ket{1 1 2 1 1} + \ket{1 2 1 1 1} \nonumber \\
				& \quad + \ket{2 1 1 1 1}.
			\end{align}
			
			Let us now characterize the symmetries of $\ket{\psi_{\text{derogatory}}}$, i.e., the invertible matrices $S^{(i)}$ satisfying 
			\begin{align}
				\label{eq:sym}
				\bigotimes_{i=1}^5 S^{(i)} \ket{\psi_{\text{derogatory}}} = \ket{\psi_{\text{derogatory}}}.
			\end{align}
			
			Let us first derive some necessary conditions for $S^{(i)}$. To this end, note that $_{i,j}\braket{22}{\psi_{\text{derogatory}}} = 0$ for any $i, j$. Hence, also $_{i,j}\bra{22}S^{(i)} \otimes S^{(j)} \ket{\psi_{\text{derogatory}}} = 0$. From this, one obtains the equations
			\begin{align}
				[S^{(i)}]_{2,1} [S^{(j)}]_{2,1} &= 0, \\
				[S^{(i)}]_{2,0} [S^{(j)}]_{2,0}  + [S^{(i)}]_{2,1} [S^{(j)}]_{2,1} &= 0,  \\
				[S^{(i)}]_{2,0} [S^{(j)}]_{2,0}  + [S^{(i)}]_{2,2} [S^{(j)}]_{2,1}  + [S^{(i)}]_{2,1} [S^{(j)}]_{2,2} &= 0 
			\end{align}
			for all $i, j$. Using that all $S^{(i)}$ are invertible one obtains $[S^{(i)}]_{2,1} = 0$ for all $i \in \{1, \ldots, 5\}$ and wlog $[S^{(i)}]_{2,0} = 0$ for all $i \in \{1, \ldots, 4\}$. Let us moreover wlog normalize $S^{(1)} \ldots, S^{(4)}$ such that $[S^{(i)}]_{2,2} = 1$ for all $i \in \{1, \ldots, 4\}$.
			Reinserting these $S^{(i)}$ into Eq. (\ref{eq:sym}), it is now straightforward to obtain
			$S^{(i)} = \begin{pmatrix} 1 & 0 & 0\\ 0 & 1 & a_i \\ 0 & 0 & 1 \end{pmatrix}$, where $a_1, \ldots, a_4 \in \mathbb{C}$ are arbitrary and $a_5 = -(a_1 + a_2 + a_3 + a_4)$. Thus, all symmetries are generated by  $B_{(i)} \otimes B^{-1}_{(j)}$, where   $B = \begin{pmatrix} 1 & 0 & 0\\ 0 & 1 & 1 \\ 0 & 0 & 1 \end{pmatrix}$ (and analytic functions thereof).

			\section{Symmetries of symmetric states \label{app:generators}}
			For completeness we show here that any symmetry of a symmetric state is generated by $A^{\otimes n}$ and $B \otimes B^{-1}$-type symmetries. This result straightforwardly follows from  \cite{MiRo13}.
			\begin{lemma}
				The symmetry group of a symmetric state $\ket{\psi}$ is generated by symmetries of the form $A^{\otimes n}$ and $B_{(i)} \otimes B_{(j)}^{-1} \otimes \identity^{\otimes (n-2)}_{\{1,\ldots, n\} \setminus \{i,j\}}$.
			\end{lemma}
			\begin{proof}
				The proof of this lemma follows very much the techniques introduced in \cite{MiRo13}, in particular on the construction in Section IIA therein. Let us assume $S= S^{(1)}\otimes S^{(2)}\otimes \ldots \otimes S^{(n)}$ is a symmetry of $\ket{\psi}$. We will prove the lemma by showing that $S$ can be written as a product of the constituents listed in the lemma. It holds that $S^{(1)} \otimes S^{(2)} \otimes \ldots \otimes S^{(n)}= \tilde{A}^{\otimes n} (\identity \otimes \tilde{S}^{(2)} \otimes \ldots \otimes \tilde{S}^{(n)})$
				where $\tilde{A}= S^{(1)}$ and $\tilde{S}^{(i)} = (S^{(1)})^{-1}S^{(i)}$ for $i \in \{2, \ldots, n\}$. As $\tilde{A}^{\otimes n}$ is symmetric, we have $\identity \otimes \tilde{S}^{(2)} \otimes \tilde{S}^{(3)} \otimes \ldots \otimes \tilde{S}^{(n)} \ket{\psi} = \tilde{S}^{(2)}  \otimes \identity \otimes \tilde{S}^{(3)} \otimes \ldots \otimes \tilde{S}^{(n)} \ket{\psi}$, from which it follows that $\tilde{S}^{(2)} \otimes \identity^{\otimes (n-1)} \ket{\psi}$ is a symmetric state and, moreover, $\tilde{S}^{(i)} \otimes \identity^{\otimes (n-1)} \ket{\psi}$ is symmetric for all $i \in \{2, \ldots, n\}$. Using Lemma 1 in  \cite{MiRo13}, we thus have that  $\tilde{S}^{(i)} \otimes (\tilde{S}^{(i)})^{-1}$ is a symmetry of $\ket{\psi}$. Let us now reduce $S$ by acting on it from the right with $\identity \otimes (\tilde{S}^{(2)})^{-1} \otimes \ldots \otimes  (\tilde{S}^{(n-1)})^{-1} \otimes \tilde{S}^{(n-1)} \ldots \tilde{S}^{(2)}$ to the form $\tilde{A}^{\otimes n} (\identity^{\otimes (n-1)} \otimes \tilde{S}^{(n)} \tilde{S}^{(n-1)} \ldots \tilde{S}^{(2)})$. Now we utilize Theorem 2 introduced in \cite{MiRo13}, by which $\tilde{S}^{(n)} \tilde{S}^{(n-1)} \ldots \tilde{S}^{(2)}$ can be rewritten as $P^n$, where $P$ is a {$n$th} root of $\tilde{S}^{(n)} \tilde{S}^{(n-1)} \ldots \tilde{S}^{(2)}$, such that $\identity^{\otimes (n-1)} \otimes P \ket{\psi}$ is symmetric. It follows from Lemma 1 in \cite{MiRo13} that therefore $P_{(i)}\otimes P_{(j)}^{-1}$ are symmetries of $\ket{\psi}$. We can thus further reduce $S$ by acting on it with $P^{\otimes (n-1)} \otimes P^{-n+1}$ from the right to the form $(\tilde{A} P)^{\otimes n}$. We have thus reduced an arbitrary symmetry $S^{(1)} \otimes S^{(2)} \otimes \ldots \otimes S^{(n)}$ by acting on it with symmetries of the form $B_{(i)} \otimes B_{(j)}^{-1}$ from the right to the form  $A^{\otimes n}$, which proves the lemma.
			\end{proof}
			
			\section{4-qutrit derogatory ES SLOCC classes}
			This appendix is meant to provide technical details that are omitted in Sec. \ref{subsec:4qutritDerog}. We first show explicit constructions of $A^{\otimes4}$ that allow the 5 SLOCC representatives in Eq.\;(\ref{eq:4qutritType2Reps}) to reach all states $\ket{\Psi_2}$ in Eq.\;(\ref{eq:psi2derog}) with $b_5\neq0$. We need not consider the states with $b_5=0$ as they are 4-qubit, and therefore must be non-derogatory \cite{footnotefullrank}. We then fill in the details for characterizing all the local symmetries of each representative.
			
			\subsection{All 4-qutrit states $\ket{\Psi_2}$ can be reached via SLOCC}\label{subsec:4qutritsAx4}
			
			For the representatives $\ket{\psi}\in\{\ket{F^4_1},\;\ket{S^4_2}+\ket{F^4_1},$ $\ket{1^4}+\ket{F^4_1}, \;\ket{1^4}+\ket{S^4_2}+\ket{F^4_1},\;\ket{S^4_3}+\ket{F^4_1}\}$, we construct invertible matrices $A$ such that
			\begin{align}
				A^{\otimes4}\ket{\psi} = b_0|0^4\rangle + \sum_{k=1}^{3} b_k|S^4_k\rangle + b_4|1^4\rangle + b_{5}|F^4_1\rangle \nonumber
			\end{align}
			for any $b_i\in\mathbb{C}$ for $i=0,...,5$ and $b_5\neq0$. We define
			\begin{align}
				A_1=\begin{pmatrix} 1 & 0 & \frac{b_0}{2}\\ 0 & x & b_1 \\ 0 & 0 & b_5 \end{pmatrix} \text{ and } A_2=\begin{pmatrix} a & d & g\\ 0 & e & h \\ 0 & 0 & \frac{b_5}{a^3} \end{pmatrix}.
			\end{align}
			\begin{enumerate}[(a)]
				\item With $\ket{\psi}=\ket{F^4_1}$ and $A_1$ where $x=1$, $A_1^{\otimes4}\ket{\psi}$ reaches all states with $b_2=b_3=b_4=0$ and $b_5\neq0$.
				
				\item With $\ket{\psi}=\ket{S^4_2}+\ket{F^4_1}$ and $A_1$ where $x=\sqrt{b_2}$ and $\det A_1=\sqrt{b_2}b_5\neq0$, $A_1^{\otimes4}\ket{\psi}$ reaches all states with $b_2,b_5\neq0$ and $b_3=b_4=0$.
				
				\item With $\ket{\psi}=\ket{1^4}+\ket{F^4_1}$ and $A_2$ where $a=1$, $e=b_4^{1/4}$, $d=\frac{b_3}{2e^3}$, $b_2=\sqrt{6}d^2e^2$, $g=\frac{b_0-d^4}{2}$, $h=b_1-2d^3e$ and $\det A_2=b_4^{1/4}b_5\neq0$, $A_2^{\otimes4}\ket{\psi}$ reaches all states with $b_2=\frac{\sqrt{6}b_3^2}{4b_4}$ and $b_4,b_5\neq0$.
				
				\item With $\ket{\psi}=\ket{1^4}+\ket{S^4_2}+\ket{F^4_1}$ and $A_2$ where $e=b_4^{1/4}$, $d=\frac{b_3}{2e^3}$, $a=\pm\sqrt{\frac{b_2}{e^2}-\sqrt{6}d^2}$, $g=\frac{b_0-d^4-\sqrt{6}a^2d^2}{2a^3}$, $h=\frac{b_1-2d^3e-\sqrt{6}a^2de}{a^3}$ and $\det A_2=\frac{b_5 e}{a^2}=\frac{4b_4^{7/4}b_5}{4b_2b_4-\sqrt{6}b_3^2}\neq0$, $A_2^{\otimes4}\ket{\psi}$ reaches all states with $b_2\neq\frac{\sqrt{6}b_3^2}{4b_4}$ and $b_4,b_5\neq0$.
				
				\item With $\ket{\psi}=\ket{S^4_3}+\ket{F^4_1}$ and $A_2$ where $a=1$, $e=b_3^{1/3}$, $d=\frac{b_2}{\sqrt{6}e^2}$, $g=\frac{b_0}{2}-d^3$, $h=b_1-3d^2e$ and $\det A_2=b_3^{1/3}b_5\neq0$, $A_2^{\otimes4}\ket{\psi}$ reaches all states with $b_3,b_5\neq0$ and $b_4=0$.
			\end{enumerate}
			Therefore, all 4-qutrit states $\ket{\Psi_2}$ with $b_5\neq0$ can be reached via SLOCC from one of the 5 SLOCC representatives in Eq.\;(\ref{eq:4qutritType2Reps}).
			
			\subsection{Structure of local symmetries}\label{subsec:structureOfS_j}
			We now prove that any local invertible symmetries $S^{(j)}$ of type-1 (type-2) derogatory ES SLOCC representatives must be block-diagonal (upper triangular), i.e., they are of the form given in Eq.\;(\ref{eq:symmPsi12}). 
			
			For all type-1 representatives $\ket{\psi_1}\in\{\ket{S^4_2}+\ket{2^4},$ $\ket{0^4}+\ket{S^4_2}+\ket{2^4}\}\cup\{\ket{0^4}+\ket{1^4}+\ket{2^4}+\mu\ket{S^4_2}:\mu\in\mathbb{C} \text{ and } \mu\neq0, \pm\sqrt{\frac{2}{3}}, \pm\sqrt{6}\}$, it holds for $k=0,1$ that
			\begin{align*}
				\prescript{}{i,j}{\bra{2,k}}S^{(i)}\otimes S^{(j)}\ket{\psi_1}=\prescript{}{i,j}{\bra{2,k}}\psi_1\rangle=0,
			\end{align*}
			which imposes constraints on the matrix elements of $S^{(j)}$:
			\begin{align}
				[S^{(i)}]_{2,m}[S^{(j)}]_{0,m}=0,\;[S^{(i)}]_{2,m}[S^{(j)}]_{1,m}=0,\label{eq:type1const1}\\
				[S^{(i)}]_{2,0}[S^{(j)}]_{0,1}+[S^{(i)}]_{2,1}[S^{(j)}]_{0,0}=0,\\
				[S^{(i)}]_{2,0}[S^{(j)}]_{1,1}+[S^{(i)}]_{2,1}[S^{(j)}]_{1,0}=0,\label{eq:type1const2}
			\end{align}
			for $m=0,1,2$ and $i\neq j$. If $[S^{(i)}]_{2,0}\neq0$ or $[S^{(i)}]_{2,1}\neq0$ for any index $i$, then it follows from the above equations that $[S^{(j)}]_{0,0}=[S^{(j)}]_{1,0}=[S^{(j)}]_{0,1}=[S^{(j)}]_{1,1}=0$ which implies that $S^{(j)}$ is not invertible. To avoid that, we must have $[S^{(i)}]_{2,0}=[S^{(i)}]_{2,1}=0$ for all $i$. Invertibility of $S^{(i)}$  requires then that $[S^{(i)}]_{2,2}\neq0$, thereby forcing $[S^{(i)}]_{0,2}=[S^{(i)}]_{1,2}=0$ for all $i$ due to Eqs.\;(\ref{eq:type1const1}). Therefore, it must be that $S^{(i)}=\begin{pmatrix} a & d & 0 \\ b & e & 0 \\ 0 & 0 & p \end{pmatrix}$ where $ae-bd\neq0$ and $p\neq0$ for all $i=1,2,3,4$. Since the constraints in Eqs.\;(\ref{eq:type1const1})--(\ref{eq:type1const2}) hold for all type-1 representatives, the block-diagonal restriction for $S^{(i)}$ applies to all of them.
			
			For all type-2 representatives $\ket{\psi_2}\in\{\ket{1^4}+\ket{S^4_2}+\ket{F^4_1},$ $\ket{S^4_3}+\ket{F^4_1}\}$, it holds for $k'=1,2$ that
			\begin{align*}
				\prescript{}{i,j}{\bra{2,k'}}S^{(i)}\otimes S^{(j)}\ket{\psi_2}=\prescript{}{i,j}{\bra{2,k'}\psi_2\rangle}=0,
			\end{align*}
			which imposes constraints on the matrix elements of $S^{(j)}$:
			\begin{align}
				[S^{(i)}]_{2,m'}[S^{(j)}]_{2,m'}=0,\;[S^{(i)}]_{2,m'}[S^{(j)}]_{1,m'}=0,\label{eq:type2const1}\\
				[S^{(i)}]_{2,0}[S^{(j)}]_{2,m'+1}+[S^{(i)}]_{2,m'+1}[S^{(j)}]_{2,0}=0,\label{eq:type2const2}\\
				[S^{(i)}]_{2,0}[S^{(j)}]_{1,m'+1}+[S^{(i)}]_{2,m'+1}[S^{(j)}]_{1,0}=0,\label{eq:type2const3}
			\end{align}
			for $m'=0,1$ and $i\neq j$. The first constraint implies that $[S^{(i)}]_{2,0}=0$ and $[S^{(j)}]_{2,1}=0$ for at least 3 indices $i,j$. If $[S^{(i)}]_{2,0}\neq0$ for one index $i$, then $[S^{(j)}]_{2,0}=0$ for all $j\neq i$ and $[S^{(i)}]_{2,0}[S^{(j)}]_{2,1}=0=[S^{(i)}]_{2,0}[S^{(j)}]_{2,2}$ from constraints (\ref{eq:type2const2}). This leads to $[S^{(j)}]_{2,1}=[S^{(j)}]_{2,2}=0$ and $\det S^{(j)}=0$ (non-invertible) for all $j\neq i$, so $[S^{(i)}]_{2,0}=0$ must hold for all $i$. For the 3 or 4 indices $i$ where $[S^{(i)}]_{2,1}=0$, $[S^{(i)}]_{2,2}$ has to be non-zero for $\det S^{(i)}\neq0$, so $[S^{(j)}]_{1,0}=0$ must hold for all $j$ to satisfy Eq.\;(\ref{eq:type2const3}). If $[S^{(i)}]_{2,1}\neq0$ for one index $i$, then from Eq.\;(\ref{eq:type2const1}), $[S^{(j)}]_{1,1}=0$ and (recall $[S^{(j)}]_{1,0}=[S^{(j)}]_{2,0}=[S^{(j)}]_{2,1}=0$) $\det S^{(j)}=0$ for all $j\neq i$. Therefore, we have that $[S^{(i)}]_{2,1}=0$ and $S^{(i)}=\begin{pmatrix} a & d & g \\ 0 & e & h \\ 0 & 0 & p \end{pmatrix}$ where $aep\neq0$ for all $i=1,2,3,4$. Since the constraints in Eqs.\;(\ref{eq:type2const1})--(\ref{eq:type2const3}) hold for both type-2 representatives, the upper triangular restriction for $S^{(i)}$ applies to all of them. This concludes our proof.
			
			\subsection{Symmetries of the form $A^{\otimes4}$}\label{subsec:AllA4}
			
			We sketch the proof for the complete characterization of $A^{\otimes4}$ symmetries that are listed in Sec.\;\ref{par:4qutritSymm} for each 4-qutrit derogatory representative $\ket{\psi}$. The proof involves solving the equation
			\begin{align}
				(A^{\otimes2}\otimes\identity^{\otimes2}- \identity^{\otimes2}\otimes \widetilde{A}^{\otimes2})\ket{\psi}=0\label{eq:A2Atilde2}
			\end{align}
			by restricting both matrices $A$ and $\widetilde{A}$ to be invertible and block-diagonal or upper triangular depending on whether the representative is type-1 or -2 and, subsequently, imposing that $\widetilde{A}=A^{-1}$. 
			
			For all type-1 representatives $\ket{\psi_1}$, it is clear that $A^{\otimes4}\ket{\psi_1}=\ket{\psi_1}$ requires the matrix entry $A_{2,2}=\pm1$ or $\pm i$ since $A$'s block-diagonal structure ensures that the state $\ket{2}$ cannot mix with $\ket{0}$ and $\ket{1}$. Hence, we focus only on the qubit subspace of $A$ and $\widetilde{A}$, which resembles the problem studied in \cite{SpdV16}, and define
			\begin{align}
				A=\begin{pmatrix} a & c \\ b & d \end{pmatrix}\oplus A_{2,2},\;\widetilde{A}=\begin{pmatrix} \tilde{a} & \tilde{c} \\ \tilde{b} & \tilde{d} \end{pmatrix}\oplus \widetilde{A}_{2,2}\;.
			\end{align}
			\begin{enumerate}[(a)]
				\item For $\ket{\psi_1}=\ket{S^4_2}+\ket{2^4}$, Eq.\;(\ref{eq:A2Atilde2}) allows only:
				\begin{enumerate}[(i)]
					\item $b=\tilde{b}=c=\tilde{c}=0$, $a=\pm\tilde{d}$, $d=\pm\tilde{a}$ and $a d=\tilde{a}\tilde{d}$. Set $\widetilde{A}=A^{-1}\Rightarrow d=\pm\frac{1}{a}$\;.
					\item $a=\tilde{a}=d=\tilde{d}=0$, $b=\pm\tilde{b}$, $c=\pm\tilde{c}$ and $b c=\tilde{b}\tilde{c}$. Set $\widetilde{A}=A^{-1}\Rightarrow c=\pm\frac{1}{b}$\;.
				\end{enumerate}
				These correspond to $A_1$ and $A_2$ in Eq.\;(\ref{eq:A1A2A3}).
				
				\item For $\ket{\psi_1}=\ket{0^4}+\ket{S^4_2}+\ket{2^4}$, Eq.\;(\ref{eq:A2Atilde2}) allows only the condition in (a)(i) above. Imposing $\widetilde{A}=A^{-1}$ in Eq. (\ref{eq:A2Atilde2}) gives $a^2=1/a^2$, which results in $A_1$ in Eq.\;(\ref{eq:A1A2A3}) with $a=\pm1$ or $\pm i$.
				
				\item  For $\ket{\psi_1}=\ket{0^4}+\ket{1^4}+\ket{2^4}+\mu\ket{S^4_2}$ with $\mu\neq0, \pm\sqrt{\frac{2}{3}}, \pm\sqrt{6}, \pm\sqrt{2}i$, Eq.\;(\ref{eq:A2Atilde2}) together with $\tilde{A}=A^{-1}$ allows only the two conditions in (a) above (corresponding to $ad-bc=\pm1$) and requires $a^2=\pm1$ and $b^2=\pm1$, which result in $A_1$ and $A_2$ in Eq.\;(\ref{eq:A1A2A3}) with $a,b=\pm1$ or $\pm i$.
				
				\item For $\ket{\psi_1}=\ket{0^4}+\ket{1^4}+\ket{2^4}+\sqrt{2}i\ket{S^4_2}$, Eq.\;(\ref{eq:A2Atilde2}) allows only the conditions in (c) and an additional symmetry (corresponding to $a d-b c=e^{i(\theta\pm\frac{\pi}{3})}$ for $\theta=0,\pi$), which demands $d=ae^{i\delta}$ and $a^2=-e^{-i(\delta+\varphi)}b^2=e^{-i(\delta+\varphi)}c^2$ for $\delta,\varphi\in\{\frac{\pi}{2},\frac{3\pi}{2}\}$. Introducing $\alpha,\beta\in\{0,\pi\}$, we have $b=aie^{i(\frac{\delta+\varphi}{2}+\alpha)}$ and $c=ae^{i(\frac{\delta+\varphi}{2}+\beta)}$. For $A$ to be invertible, $ad-bc=a^2e^{i\delta}(1-ie^{i(\alpha+\beta+\varphi)})\neq0\Rightarrow e^{i(\alpha+\beta+\varphi)}=i$. Finally, imposing $(A^{\otimes4}- \identity^{\otimes4})\ket{\psi_1}=0$ gives 
				\begin{align*}
					a^4=\begin{cases}\frac{1}{4}e^{- i\frac{\pi}{3}}, \text{\; if } e^{i(\delta+\varphi)}=1\\
						\frac{1}{4}e^{i\frac{\pi}{3}}, \text{\; if } e^{i(\delta+\varphi)}=-1,\end{cases}
				\end{align*}
				\begin{align*}
					\Rightarrow a=\begin{cases}\frac{1}{\sqrt{2}}e^{i\pi(\frac{m}{2}-\frac{1}{12})}, \text{\; if } e^{i(\delta+\varphi)}=1,\\
						\frac{1}{\sqrt{2}}e^{i\pi(\frac{m}{2}+\frac{1}{12})}, \text{\; if } e^{i(\delta+\varphi)}=-1,\end{cases}
				\end{align*}
				for $m=0,1,2,3$, thereby obtaining $A_3$ in Eq.\;(\ref{eq:A1A2A3}).
			\end{enumerate}
			
			For all type-2 representatives $\ket{\psi_2}$, we define
			\begin{align}
				A=\begin{pmatrix} a & d & g \\ 0 & e & h \\ 0 & 0 & p \end{pmatrix} \text{ and } \widetilde{A}=\begin{pmatrix} \tilde{a} & \tilde{d} & \tilde{g} \\ 0 & \tilde{e} & \tilde{h} \\ 0 & 0 & \tilde{p} \end{pmatrix}.
			\end{align}
			\begin{enumerate}[(a)]
				\item For $\ket{\psi_2}=\ket{1^4}+\ket{S^4_2}+\ket{F^4_1}$, Eq.\;(\ref{eq:A2Atilde2}) and the invertibility of $A$ and $\widetilde{A}$ allow only $d=\tilde{d}=h=\tilde{h}=0$, $a=\pm\tilde{e}$, $e=\pm\tilde{a}$, $ag=\tilde{a}\tilde{g}$, $a^2=\tilde{a}\tilde{p}$ and $\tilde{a}^2=ap$. Set $\widetilde{A}=A^{-1}\Rightarrow e=\pm\frac{1}{a}$, $g=0$, and $p=\frac{1}{a^3}$. Finally, $A^{\otimes4}\ket{\psi_2}=\ket{\psi_2}$ requires $a^4=1$. Thus, the only option for $A$ is $A=a\oplus\pm a\oplus a$ where $a=\pm1$ or $\pm i$. 
				
				\item For $\ket{\psi_2}=\ket{S^4_3}+\ket{F^4_1}$, Eq.\;(\ref{eq:A2Atilde2}) and the invertibility of $A$ and $\widetilde{A}$ allow only $d=\tilde{d}=h=\tilde{h}=0$, $e^2=\tilde{a}\tilde{e}$, $\tilde{e}^2=ae$, $ag=\tilde{a}\tilde{g}$, $a^2=\tilde{a}\tilde{p}$ and $\tilde{a}^2=ap$. Set $\widetilde{A}=A^{-1}\Rightarrow g=0$, $p=\frac{1}{a^3}$ and $e=\frac{1}{a^{1/3}}e^{i\frac{2m\pi}{3}}$ for $m=0,1,2$. Thus, $A$ can only be $A=a\oplus\frac{1}{a^{1/3}}e^{i\frac{2m\pi}{3}}\oplus\frac{1}{a^3}$ where $a\in\mathbb{C}$. 
			\end{enumerate}
			This concludes the characterization of all $A^{\otimes4}$ symmetries of every 4-qutrit derogatory representative.
			
			\section{Conversion not achievable by all-deterministic-$\text{LOCC}_\mathbb{N}$}\label{app:NoAllDetLOCC}
			
			In this appendix, we explain why the initial state $\ket{\Psi}\propto \sqrt{G_1}\otimes \sqrt{G_2}\otimes\identity\otimes\identity\ket{\Psi_s}$ in Sec.\;\ref{sec:probStepProtocol} cannot be converted to $\ket{\Phi}\propto \sqrt{H_1}\otimes \sqrt{H_2}\otimes\identity\otimes\identity\ket{\Psi_s}$ via any all-det-$\text{LOCC}_\mathbb{N}$ protocol. In Sec.\;\ref{sec:probStepProtocol}, we defined $\ket{\Psi_s}\coloneqq|\psi(\mu=\sqrt{2}i)\rangle$ and the positive matrices $G_i$ and $H_i$ for $i=1,2$ in Eqs.\;(\ref{eq:G1_def})--(\ref{eq:H2_def}). It is clear that $\ket{\Psi}$ is not LU equivalent to $\ket{\Phi}$ since ${S^{(i)}}^\dagger G_i S^{(i)}\not\propto H_i$ for $i=1,2$ and for any symmetry $\bigotimes_{j=1}^4 S^{(j)}\in S_{\Psi_s}$.
			
			We first recall that all these protocols constitute of a sequence of $\text{LOCC}_1$ transformations \cite{footnoteLOCC1}. Then, we observe that all non-trivial symmetries $S\not\propto\identity^{\otimes4}$ of $|\psi(\mu=\sqrt{2}i)\rangle$ [see item (iii) under ``type-1 symmetries" in Sec.\;\ref{par:4qutritSymm})] that quasi-commute with at least 3 sites of $G_1\otimes G_2\otimes\identity^{\otimes2}$ must take the form, $\bigotimes_{i=1}^4 D^{(i)}$, where $D^{(1)}$ or $D^{(2)}=\identity$ due to the specific forms of $G_1$ and $G_2$ [see Eqs.\;(\ref{eq:G1_def}) and (\ref{eq:G2_def})] and $D^{(j)}=\text{diag}(1,1,e^{i\theta^{(j)}})$ with $\theta^{(j)}\in[0,2\pi)$ for other sites $j$ such that $\prod_{j=1}^4 e^{i\theta^{(j)}}=1$. According to the proof of Lemma \ref{lemma:convertibility}, an $\text{LOCC}_1$ conversion uses only symmetries $\bigotimes_{j=1}^4 S^{(j)}$ that quasi-commute with $G_i$ on 3 sites where the corresponding parties apply only LU, whereas for the site that a non-trivial measurement occurs, the symmetries may or may not quasi-commute with $G_i$. 
			
			In the first round, party $i$ performs a non-trivial measurement, which results in a state with $G_i$ replaced by $G'_i$ which fulfills the equation
			\begin{equation}
			\sum_k p_k (D^{(i)}_k)^\dagger G'_i D^{(i)}_k=G_i
			\label{eq:depolarizeG}
			\end{equation} 
			for some choices of probabilities $\{p_k\}$ and diagonal unitaries $D^{(j)}_k=\text{diag}(1,1,e^{i\theta^{(j)}_k})$ [see Lemma \ref{lemma:convertibility}]. If $i=1$ or 2, then the entries of $G'_i$ must satisfy: $(G_i')_{jk}=(G_i)_{jk}\;\forall\;j,k\in\{1,2,3\}$ except for $(j,k)\in\{(1,3),(3,1)\}$ where $(G_i')_{13}=f(G_i)_{13}$ and $(G_i')_{31}=f^*(G_i)_{31}$ with $f=(\sum_k p_k e^{i\theta^{(j)}_k})^{-1}$, so $|f|>1$. If $i=3$ or 4, then $G_i'=\begin{pmatrix}
			1 & 0 & \gamma^*\\
			0 & 1 & \epsilon^*\\
			\gamma & \epsilon & 1
			\end{pmatrix}$ with $\gamma\neq0$ or $\epsilon\neq0$, which satisfies Eq.\;(\ref{eq:depolarizeG}) for $\sum_k p_k e^{i\theta^{(j)}_k}=0$. One can easily verify that ${S^{(i)}}^\dagger G'_i S^{(i)}\not\propto H_i$ for any $\bigotimes_{j=1}^4 S^{(j)}\in S_{\Psi_s}$, so the resulting state is LU inequivalent to the target state.
			
			In any subsequent rounds of $\text{LOCC}_1$ transformations, the only symmetries in $S_{\Psi_s}$ that can quasi-commute with 3 or more sites are again in the form $\bigotimes_{i=1}^4 D^{(i)}$ where $D^{(j)}=\text{diag}(1,1,e^{i\theta^{(j)}})$ with $\theta^{(j)}\in[0,2\pi)$. Hence, party $j$'s measurement will update $G_j$ or $G'_j$ of the state from the previous round to $G''_j$ of the current state in the same way as $G_i$ is updated to $G'_i$ in the first round. It is easy to check that ${S^{(j)}}^\dagger G''_j S^{(j)}\not\propto H_j$ for any $\bigotimes_{k=1}^4 S^{(k)}\in S_{\Psi_s}$. Therefore, regardless of the ordering and which parties perform the measurements, the resulting state} is not LU equivalent to the target state $\ket{\Phi}\propto\sqrt{H_1}\otimes\sqrt{H_2}\otimes\identity^{\otimes2}\ket{\Psi_s}$. We can then conclude that $\ket{\Psi}$ cannot be transformed into $\ket{\Phi}$ with any all-det-$\text{LOCC}_\mathbb{N}$ protocol.


\begin{thebibliography}{99}
\bibitem{footnote1} Typically, the dimension of the union of SLOCC classes will be even smaller as SLOCC classes will not only contain a single symmetric state, and might have an infinite stabilizer. Note that we are considering here complex parameters.
\bibitem{footnote2} This can be easily seen by noting that the permutation symmetry of $\ket{\psi}$ implies that $B\otimes(B^{-1}B)\otimes B^{-1}\otimes\identity^{\otimes n-3}\ket{\psi}=\lambda'^2\ket{\psi}=B\otimes B^{-1}\otimes\identity^{\otimes n-2}\ket{\psi}=\lambda'\ket{\psi}$, which gives $\lambda'=\lambda'^2=1$ for $\lambda'\neq0$.
\bibitem{footnote3} This follows from the proof of Theorems 4 and 7 in \cite{MiRo13} and the fact that under action of $B$ on a symmetric state as in Eq. (\ref{eq:symstate}), subspaces of fixed block distribution $n_1, \ldots, n_K$ do not couple with each other.
\bibitem{footnote5} In the non-derogatory case, SL operations are constructed such that $\left|E_k\right>$ can be transformed into an arbitrary superposition of excitation states $\sum_j \left|E_j\right>$ \cite{MiRo13}. In the derogatory case, the same trick can be utilized, but only works for a single block-configuration $(n_1, \ldots, n_K)$.
\bibitem{footnote6} Note that if the stabilizer is finite, it can always be chosen to be unitary by considering a different representative of the SLOCC--class.
\bibitem{footnoteLOCC1} Contrary to \cite{SpdV17,dVSp17} $LOCC_1$ corresponds here to a protocol where one party applies a non-trivial measurement and the other parties implement (depending on the outcome) a unitary, i.e., to a single round in a LOCC protocol, and not to a single-round protocol where party $1$ is the party that performs the non-trivial measurement.
\bibitem{footnote7} Note that in \cite{GoKr17} some of us showed that generically an $n$--qubit state possesses only trivial symmetries by utilizing symmetric seed states. However, this result does not imply that within the subspace of symmetric states the results holds generically.
\bibitem{footnote8} Note that the additional constraint $\prod_ja_j=1$ is irrelevant and we will ignore it in what follows. The reason for that is that one can then redefine the symmetry operator to be $A'=A/(\prod_ja_j)^{1/n})$.
\bibitem{footnote9} Note that this corresponds, e.g. to the diagonal part of the symmetry of the GHZ--state.
\bibitem{footnote10} Note that here in contrast to the previous sections we allow for states which do not have full local rank or matrices B that are only invertible on a subspace.
\bibitem{footnotefullrank} Note that states that are non-derogatory on the subspace which they are supported on are also non-derogatory on the whole Hilbert space. As they do not have full rank, this does not contradict the statement that the unique states (up to SLOCC) that are stabilized by $B\otimes B^{-1}$ for a non-derogatory $B$ in Jordan form are given by the ones in Eq. (\ref{eq:sumek}) (see \cite{MiRo13}).
\bibitem{footnote11} All $n$-qubit ES states are non-derogatory since there are no derogatory matrices in $d=2$.
\bibitem{footnote12} This is not an element of the local symmetry group of the type-2 representative $\ket{S^4_3}+\ket{F^4_1}$, which is why we exclude the SLOCC class of $\ket{S^4_3}+\ket{F^4_1}$ in this example.
\bibitem{footnoteAppA} Note that the proportionality factor $r$ is required here as $S^{(1)}$ is not necessarily unitary.
\bibitem{footnoteAppD} Recall that $\vec{\gamma}_{n}$ is fixed once $\vec{\gamma}_1, \ldots, \vec{\gamma}_{n-1}$ are chosen.
\end{thebibliography}
	\end{document}